\documentclass[a4paper,onecolumn,11pt,accepted=2023-03-07]{quantumarticle}

\pdfoutput=1
\usepackage[utf8]{inputenc}
\usepackage[english]{babel}
\usepackage[T1]{fontenc}
\usepackage{amsmath}
\usepackage{hyperref}
\usepackage{tikz}
\usepackage{lipsum}
\usepackage{enumitem}
\usepackage{graphicx}
\usepackage{caption}
\usepackage{latexsym}
\usepackage[nointegrals]{wasysym}
\usepackage{color}
\usepackage{bbold}
\usepackage{dsfont}
\usepackage{multirow}
\usepackage{array,booktabs}
\usepackage{minibox}
\usepackage{epstopdf}
\usepackage{epsfig}
\usepackage{amsthm}
\usepackage{amsmath}
\usepackage{amsfonts}
\usepackage{color}

\newtheorem{Xeorem}{Theorem}[section]
\newtheorem{definition}[Xeorem]{Definition}
\newtheorem{result}[Xeorem]{Smallest Contextual Set.}
\newtheorem{resultt}[Xeorem]{\bf Noncontextual
  $\boldsymbol{\longrightarrow}$ Contextual.}
\newtheorem{resulttt}[Xeorem]{\bf Operator
  $\boldsymbol{\leftrightarrow}$ MMP Rule.}
\newtheorem{resultttt}[Xeorem]{\bf Quantum Indeterminacy Postulate.}
\newtheorem{resulttttt}[Xeorem]{\bf Peres-Mermin contradiction.}
\newtheorem{resultttttt}[Xeorem]{\bf Graph $\leftrightarrow$ MMP
           Hypergraph}
\newtheorem{lemma}[Xeorem]{Lemma}

\newtheorem{statistics}[Xeorem]{Hypergraph Statistics}
\newtheorem{discussion}[Xeorem]{Discussion}
\usepackage{float}

\definecolor{Blue}{rgb}{0.3,0.3,0.9}
\def\crta{\vrule height1.41ex depth-1.27ex width0.34em}
\def\dj{d\kern-0.36em\crta}
\def\Crta{\vrule height1ex depth-0.86ex width0.4em}
\def\Dj{D\kern-0.73em\Crta\kern0.33em}
\begin{document}

\title{Quantum Contextuality}

\author{Mladen Pavi{\v c}i{\'c}}
\affiliation{Center of Excellence CEMS, Photonics and Quantum Optics
  Unit, Ru\dj er Bo\v skovi\'c Institute and Institute of Physics,
  Zagreb, Croatia}
\orcid{0000-0003-1915-6702}
\email{mpavicic@irb.hr}
\homepage{http://www2.irb.hr/users/mpavicic}
\thanks{Funded by the Ministry of Science and Education of Croatia
through the Center of Excellence for Advanced Materials and Sensing
Devices (CEMS) funding, and by MSE grants Nos.~KK.01.1.1.01.0001
and 533-19-15-0022.}

\begin{abstract}
  Quantum contextual sets have been recognized as resources
for universal quantum computation, quantum steering and quantum
communication. Therefore, we focus on engineering the sets that
support those resources and on determining their structures and
properties. Such engineering and subsequent implementation rely
on discrimination between statistics of measurement data of
quantum states and those of their classical counterparts. The
discriminators considered are inequalities defined for hypergraphs
whose structure and generation are determined by their basic
properties. The generation is inherently random but with the
predetermined quantum probabilities of obtainable data. Two
kinds of statistics of the data are defined for the hypergraphs
and six kinds of inequalities. One kind of statistics, often
applied in the literature, turn out to be inappropriate and
two kinds of inequalities turn out not to be noncontextuality
inequalities. Results are obtained by making use of universal
automated algorithms which generate hypergraphs with both odd
and even numbers of hyperedges in any odd and even dimensional
space---in this paper, from the smallest contextual set with
just three hyperedges and three vertices to arbitrarily many
contextual sets in up to 8-dimensional spaces. Higher dimensions
are computationally demanding although feasible.
\end{abstract}
  

\maketitle 

\section{\label{sec:intro}Introduction}

A series of experiments with state-independent
\cite{beng-blan-cab-pla12} contextual sets has been
carried out recently, using photons
\cite{amselem-cabello-09,liu-09,d-ambrosio-cabello-13,ks-exp-03,canas-cabello-8d-14,canas-cabello-14},
neutrons \cite{h-rauch06,b-rauch-09},
trapped ions \cite{k-cabello-blatt-09}, and
solid state molecular nuclear spins  \cite{moussa-09}.

These experiments pave the road for applications of contextual
sets in quantum computation \cite{magic-14,bartlett-nature-14},
quantum steering \cite{tavakoli-20}, and quantum communication
\cite{saha-hor-19} by measuring yes-no outputs of quantum systems  
and contrasting them with predetermined 0--1 values of the 
corresponding classical systems. Such systems might be organized
in contextual sets represented by graphs or hypergraphs and their
properties and features are the main subject of the present paper.

Intuitively speaking, a (hyper)graph is a set of points and a
set of subsets of these points. The points are called the vertices
of the (hyper)graph and the subsets are called the (hyper)edges of
the (hyper)graph. Vertices might be represented by vectors,
operators, subsets, or other objects, and (hyper)edges by a
relation between vertices contained in them such as orthogonality,
inclusion, or geometry. We follow Berge \cite{berge-73,berge-89},
Bretto \cite{bretto-13}, and Voloshin \cite{voloshin-09} in all
details except in several restrictions needed for hypergraph
description of contextual sets which we introduce in
Sec.~\ref{sec:h-lang}. Historically, representations used to
describe and depict contextual sets appeared in several different
forms and definitions, e.g., partial Boolean algebra
\cite{koch-speck}, operator and projectors
\cite{cabello-08,badz-cabel-09}, lists or tables of
vectors and their orthogonalities \cite{peres,planat-saniga-12},
Greechie diagrams
\cite{svozil-tkadlec,svozil-book-ql,svozil-18,cabello-svozil-18,svozil-20,budroni-cabello-rmp-22},
Kochen-Specker (KS) proofs  \cite{planat-12}, parity proofs
\cite{waeg-aravind-jpa-11}, MMP diagrams \cite{pmmm05a,pmmm05a-corr},
graphs with cliques \cite{yu-oh-12}, node-context graphs
\cite{lisonek-14,budroni-cabello-rmp-22}, etc. However,
as shown in this paper, when some discrepancies between these
definitions as well as their possible inner limitations are
smoothed out, all of them boil down to hypergraphs and in this
paper we provide a hypergraph platform for major results and
achievements in the field of quantum contextual sets.  

Connections between contextuality and universal quantum
computation \cite{magic-14} and steering \cite{tavakoli-20}
that have recently been established ask for a quantification
of properties of contextual sets, e.g., robustness to noise
\cite{cabello-bengtsson-12}, size of maximal independent sets
of stabilizer states \cite{magic-14}, or suitability for
implementation in general. It has been shown that
inequalities are an efficient tool for the purpose
\cite{yu-oh-12,beng-blan-cab-pla12,xu-chen-su-pla15,ram-hor-14,cabell-klein-budr-prl-14,pavicic-entropy-19,yu-tong-14,yu-tong-15}.
Yu, Guo, and Tong prove \cite{yu-tong-15} that operator
formulations of KS  contextual sets can always be
converted to state-independent noncontextuality inequalities. The
problem with the inequalities in these references is that either
no definite inequality is given or that they were given for chosen
particular contextual sets previously specified via sets of
vectors/rays, or that they have not been formulated for
probabilities applicable to genuine YES-NO quantum experiments.
As a consequence, while billions of hypergraph-defined contextual
sets are known, a straightforward automated way of generation of
operator-based inequalities from them is missing. On the other
hand, there are operator-defined sets, e.g., the Peres-Mermin
square \cite{peres90,mermin90}, for which a proper underlying
vector set awaits to be defined.

Therefore we broaden the scope of the contextuality so as to
cover both operators and hypergraphs. We compare their features
and their inequalities. We define and/or reconsider six
different kinds of hypergraph inequalities that correspond to
the aforementioned operator inequalities: two are based on
hyperedges, two on vertices, and two are mixed; we compare them
with the known operator inequalities.

Hyperedge-based inequalities are well-known (it stems
directly from the Kochen-Specker theorem) and, essentially, they
boil down to our impossibility of assigning exactly one `1'
to vertices in each hyperedge of a contextual hypergraph. So, there
are always fewer such hyperedges than there are hyperedges altogether
in the set, and the inequalities just confirms this discrepancy
\cite{cabello-08,badz-cabel-09,yu-tong-15}. They correspond to
the operator noncontextuality inequalities.

Some vertex and mixed inequalities rely on two different kinds
of statistics of the outcomes of quantum YES-NO measurements:
raw data statistics and postprocessed data statistics. These
yield four different kinds of inequalities: the original
Gr{\"o}tschel-Lov{\'a}sz-Schrijver (GLS) inequality, the quantum
forms of the GLS inequality, and inequalities that we call v- and
e-inequalities. The original GLS inequality holds for any graph
or hypergraph for variable probabilities within each clique or
hyperedge, respectively. However, these probabilities are not
variable but constant within any quantum YES-NO measurements and
under them arbitrary many contextual graphs and hypergraphs
violate the GLS inequality, i.e., the GLS inequality is not a
noncontextuality inequality. The v- and e-inequalities are
satisfied for all contextual (hyper)graphs and violated for all
noncontextual ones and unlike the GLS-like inequalities they
correspond to the existing operator-based inequalities.

The aforementioned types of inequalities are determined by
structural properties of the hypergraphs that define them.
Structural properties we obtain characterize contextual as well
as noncontextual hypergraphs and are relevant for application of
contextual sets in quantum computation and quantum communication.
The properties serve us to
\begin{enumerate}[label=(\alph*)]
  \item characterize the hypergraphs themselves;
  \item analyze hypergraphs probability and randomness characteristic
    for obtaining small contextual sets from big master sets;
  \item obtain hypergraphs from elementary vector components in any
    odd or even dimensional space (in this paper in 3- to 8-dim
    spaces);
  \item obtain contextual hypergraphs from noncontextual ones
    by deleting a certain number of vertices from them;
  \item establish a correspondence between hypergraph and operator
    approaches;
  \item obtain state independent hypergraph inequalities
    where operator approach gives state dependent ones;
  \item introduce new hypergraph-defined measurements based on
    multiplicity of vertices and postprocessing of multiple
    detection at the ports of the gates;
  \item obtain the smallest critical contextual hypergraph with
    just 3 edges and 3 vertices;
  \item prove that one of the graphs which are considered to be
    a source of quantum computer’s power is a subhypergraph of a
    non-critical KS hypergraph;
  \item derive a vector-hypergraph underlying the 3x3
    Peres-Mermin operator square.
\end{enumerate}

An outline of the paper is given by the following organisational
flow.

In Sec.~\ref{sec:h-lang} we give the definitions of a general
hypergraph and of its McKay-Megill-Pavi\v ci\'c (MMP) hypergraph
restriction; then we introduce notation, language, algorithms,
and programs for MMP hypergraphs and compare them with other
notations and formalisms of contextual sets from the literature.

In Sec.~\ref{sec:h-ext}, we state the Kochen-Specker
and Bell theorems and introduce several generalizations of
theirs.

In Sec.~\ref{subsec:methods}, we present three methods of
MMP hypergraph generation we make use of in this paper. 

In Sec.~\ref{sec:op-ineqal}, we review the operator-based
inequalities from the literature some of which we correlate with
our results in subsequent sections.

In Sec.~\ref{sec:hyp-op}, we compare the MMP hypergraph and
operator approaches to contextual sets using the example of a
3-dim pentagon set.

In Sec.~\ref{sec:structure}, we analyze the structure of MMP
hypergraphs and introduce notions and theorems and lemmas that
characterize them; we consider two kinds of quantum statistics:
the raw data and postprocessed data statistics and six kinds of
inequality: GLS-like-, v-, e$_{Max}$-, and e$_{min}$-inequalities;
we also compare operator and hypergraph approach to the introduced
notions and features.

In Sec.~\ref{sec:structure-e}, we present several examples of
MMP hypergraph structure introduced in Sec.~\ref{sec:structure}. 

In Sec.~\ref{sec:exampl}, we apply the results and notions
obtained in the previous sections to the MMP hypergraph
multiplicity in Sec.~\ref{subsec:4dm}, to the 3-dim MMP
hypergraphs in Sec.~\ref{subsec:3d}, to chosen 4-dim MMPs
in Sec.~\ref{subsec:4d}, to $\Gamma$ set that has recently
been used to prove that contextuality is the source of quantum
computer's power in Sec.~\ref{subsec:magic}, to the
Peres-Mermin square Sec.~\ref{subsec:mermin}, and to the
5- to 8-dim contextual MMP hypergraphs in
Secs.~\ref{subsec:5d}, \ref{subsec:6d}, and \ref{subsec:78d},
respectively.

In Sec.~\ref{sec:impl} we consider possible general
implementation schemes.

In Sec.~\ref{sec:disc} we discuss the obtained results.

In the Appendices we give a comparison of historical as well as
contemporary hypergraph formalisms from the literature with the
MMP hypergraphs language as well as strings and coordinatizations
of bigger MMP hypergraphs to avoid visual clutters in the main
body of the paper.

\section{\label{sec:h-lang}MMP hypergraph language}

In this section, we start with a general definition of a
hypergraph, which we then narrow down to the MMP hypergraph.
After giving specifics of the MMP hypergraph language which
will be the language of our presentation) we review other
formalisms that have been used for generation of contextual
sets in the literature and show that they all reduce to the
MMP hypergraph formalism.

A general {\em hypergraph\/} is defined as follows
\cite{berge-73,berge-89,voloshin-09,bretto-13}.
Let $V=\{v_1,v_2,\dots ,v_k\}$ be a finite set of elements called
{\em vertices} and let $E=\{e_1,e_2,\dots ,e_l\}$ be a family of
subsets of $V$ called {\em hyperedges}. The pair ${\cal{H}}=(V,E)$
is called a {\em hypergraph\/} with {\em vertex set\/} $V$ also
denoted by $V({\cal{H}})$, and {\em hyperedge set\/}
$E$ also denoted by $E({\cal{H}})$. A hypergraph ${\cal{H}}$ may
be drawn as a set of points representing the vertices subsets of
which represent hyperedges as follows: a hyperedge $e_j$ is
represented by a continuous curve joining two elements if the
cardinality (number of elements, vertices) within the hyperedge
is $|e_j|=2$, by a loop if $|e_j|=1$, and by a closed curve
enclosing the elements if $|e_j|>2$. Numerically they are
represented by the incidence matrices \cite[p.~2,Fig.~1]{berge-89}
in which columns are hyperedges and rows are vertices.
Intersection of hyperedge columns with vertex rows contained in
hyperedges are assigned `1' and those not contained are assigned `0'.

The number of vertices within a hypergraph ($k$), i.e.,
the cardinality of $V$ ($|V|$), is called the {\em order\/} of
a hypergraph, and the number of hyperedges within a hypergraph
($l$), i.e., the cardinality of $E$ ($|E|$), is called the
{\em size\/} of a hypergraph. 

To arrive at MMP hypergraphs we restrict the
general hypergraphs numerically and graphically. Numerically,
we substitute ASCII characters for vectors, operators, or
elements within tables and matrices from the literature and attach 
these ASCII characters to vertices. Mutually related vertices
are collected in one-line strings representing hyperedges. The
relation might be orthogonality, inclusion, geometry, etc.
Thus, numerically, an MMP hypergraph, defined in
Def.~\ref{def:MMP-string}, is a string of characters corresponding
to vertices which are organized in substrings separated by commas
(``,'') corresponding to hyperedges; the string ends with a period
(``.''). Graphically, vertices are dots and hyperedges are lines or
curves connecting vertices by passing through them; we dispense with
hyperedges of cardinality 0 and 1 ($|e_j|=0,1$) and since each
contextuality contradiction occurs within a single connected set we
do not have unconnected subhypergraphs and we do not have hyperedges
attached to the main body of an MMP hypergraph at only one vertex.
Also, because the Hilbert space in which contextual sets reside when
equipped with a coordinatization must have at least 3 dimensions
(3-dim), we introduce the {\em hypergraph-dimensionality\/} $n\ge 3$.
Thus we arrive at the following formal definition of an MMP hypergraph.

\begin{definition}\label{def:MMP-string}
  An {\rm MMP} hypergraph is a connected {\em hypergraph}
  ${\cal{H}}=(V,E)$ (where $V=\{V_1,V_2,\dots ,V_k\}$ is a set of
  {\em vertices} and $E=\{E_1,E_2,\dots ,E_l\}$ sets of
  {\em hyperedges}) of {\em hypergraph-dimension} $n\ge 3$ in which
\begin{enumerate}
  \item Every vertex belongs to at least one hyperedge;
  \item Every hyperedge contains at least $2$ and at most $n$
    vertices;
  \item No hyperedge shares only one vertex with another
    hyperedge;
  \item Hyperedges may intersect each other in at most $n-2$
    vertices.
  \item Graphically, vertices are represented as dots and
    hyperedges as (curved) lines passing through them.
\end {enumerate}
\end{definition}

\begin{definition}\label{def:MMP-dim}{\bf MMP hypergraph-dimension}
  $n$ is a predefined (for an assumed task or purpose) maximal
  possible number ($n$) of vertices within a hyperedge even when
  none of the actually processed hyperedges include $n$ vertices.
\end{definition}

This is operationally requested for any
implementation since a full coordinatization of vertices
turn hypergraph-dimension $n$ into a dimension of a Hilbert
space determined by vectors each vertex is assigned to. 
But until we invoke a coordinatization of an MMP hypergraph we
can handle it solely by means of Defs.~\ref{def:MMP-string}
and \ref{def:MMP-dim}.

Notice that in our previous papers we did not have the condition
3 in the definition of the MMP hypergraph but all our results
in that papers were obtained by excluding the corresponding
hyperedges from the calculations explicitly or implicitly.
Therefore, we need not introduce a different name for the MMP
hypergraph from Def.~\ref{def:MMP-string} and from now on we
shall assume that the condition 3 holds in the definition of
the MMP hypergraph. Another formulation of the condition would
be that all hyperedges of an MMP hypergraph with two or more
of hyperedges must share at least two vertices, i.e., that no
hyperedge should be attached to the main body of the MMP
hypergraph at just one vertex. A single hyperedge is
therefore an MMP hypergraph, as well as two hyperedges that
share two or more vertices.

We encode MMP hypergraphs with the help of ASCII characters
\cite{pmmm05a,pmmm05a-corr}. Vertices are denoted by one of the
following 90 characters:
{{\tt 1 2 \dots\ 9 A B \dots\ Z a b
 \dots\ z !\ " \#} {\$} \% \& ' ( ) * - / : ; \textless\ =
\textgreater\ ? @ [ {$\backslash$} ] \^{} \_ {`} {\{}
{\textbar} \} $\sim$}\ \ \cite{pmmm05a,pmmm05a-corr}. A 91st character
`+', is used for the following purpose: when all aforementioned 
characters are exhausted, we reuse them prefixed by `+', then again
by `++', and so on (See Appendices). An $n$-dim contextual
hypergraph with $k$ vertices and $l$ hyperedges, a hypergraph
of order $k$ and size $l$, we denote as a $k$-$l$ hypergraph.
There is no limit on the size of an MMP hypergraph.

In Fig.~\ref{fig:st-hyp} we illustrate the difference between the
standard and the MMP hypergraph formalism. In the standard hypergraph
formalism, hyperedges between two vertices are represented by
straight lines as taken over from the graph theory ($e_3,e_7$).
Hyperedges containing three or more vertices encircle the vertices
($e_1,e_2,e_4,e_8$). Hyperedges containing only one vertex have two
representations: $e_5$ \cite{voloshin-09,bretto-13} and 
$e_6$ \cite{berge-73,berge-89}. In the MMP hypergraph formalism
$e_3$ hyperedge is represented as the line (or curve) connecting
vertices {\tt 1} and {\tt 6}: $E_3={\tt 16}$. Hyperedges $e_1,e_2$
and $e_4$ are represented as curves (or lines) passing through
vertices they contain: $E_1={\tt 2345}$, $E_2={\tt 1236}$, and
$E_4={\tt 1456}$, respectively. Vertex corresponding to vertex 10 in
Fig.~\ref{fig:st-hyp}(a) does not exist in an MMP representation due
to Def.~\ref{def:MMP-string}.1; the same is with $e_5$,$e_6$,$e_9$
due to Def.~\ref{def:MMP-string}.2., and with hyperedges $e_7$,$e_8$
and vertices {\tt 7,8,9} due to Def.~\ref{def:MMP-string}.1.\/
\&\ Def.~\ref{def:MMP-string}.3.

\begin{figure}[h]
\begin{center}
  \includegraphics[width=0.8\textwidth]{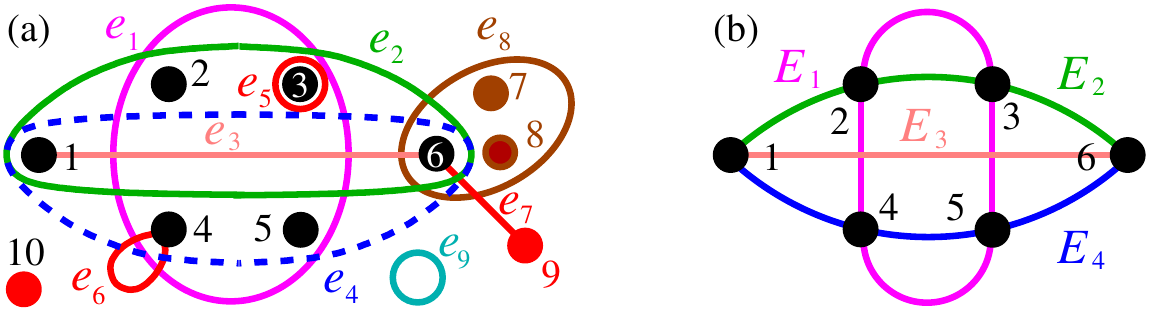}
\end{center}
\caption{(a) Representation of a general hypergraph
  \cite{berge-73,berge-89,voloshin-09,bretto-13};
  (b) representation of a corresponding MMP hypergraph whose ASCII
  string is {\tt 16,1236,6541,2354.}}
\label{fig:st-hyp}
\end{figure}

Taken together, an MMP hypergraph is a special kind of a general
hypergraph in which none of the aforementioned points
({\em 1.-5.}\/) holds within its definition.

\begin{resultttttt}\label{result:g-h}
  We turn a graph whose every edge contains just two vertices
  into an {\rm MMP} hypergraph so as to substitute hyperedges
  for (interwoven) cliques of related vertices and for isolated
  edges. Conversely, we turn an MMP hypergraph into a graph so
  as to substitute a cliques of edges for MMP hyperedges, where
  vertices within a clique correspond to three related vertices
  within a hyperedge, until we exhaust all related triples within
  the hyperedge; isolated graph edges are substituted for
  the hypergraphs of cardinality 2.
\end{resultttttt}

We generate, deal with, and handle MMP hypergraphs by means
of automated algorithms implemented into the programs
{\textsc{One, Mmpstrip, Vecfind, States01}}, and others which
have MMP strings as their inputs and outputs. 
In contrast, there are no such automated algorithms and programs
for incidence matrices of the general hypergraph formalisms known
to us. We stress here, that for a $k$-$l$ $n$-dim KS hypergraph in
a general hypergraph representation, its incidence matrix contains
$k$ vertex rows each $l$ columns long, while an MMP hypergraph is
represented by a single line containing $l$ $n$-tuples of vertices.
E.g., the incidence matrix of the original 4-dim 192-118
Kochen-Specker hypergraph contains 118 hyperedge columns and 192
vertex rows/lines ($192\times118$ matrix), while its single
MMP string (one line) contains 118 triples of ASCII characters
\cite[Supp.~Material]{pavicic-pra-22}.

MMP hypergraph language has been developed over the last 20 years
with the goal of making handling and generation of contextual
sets as efficient as possible. Here we give a comparison of
historical as well as contemporary formalisms with the MMP language. 

\begin{itemize}
\item Partial Boolean algebra used in \cite{koch-speck} generates
  graphs with cliques whose edges contain only two vertices and
  whose computer and graphical processing is therefore more
  demanding than those of MMP hypergraphs to which they can be
  straightforwardly reduced; the same problem applies to all
  other graphs with clique approaches
  \cite{yu-oh-12,cabello-klein-port-prl-14}. Graphical
  representations of such graphs, especially big ones, are
  often unintelligible---compare Figs.~\ref{fig:gr-hyp}(b) and
  \ref{fig:gr-hyp}(c) and Figs.~\ref{fig:graph}(e) and
  \ref{fig:graph}(g).
\item Operators or projectors used to generate contextual
  inequalities are mostly constructed manually by means of
  states/vectors/vertices of chosen contextual hypergraphs,
  meaning that they make use of already known hypergraphs
  \cite{cabello-08,badz-cabel-09} which can serve us to obtain
  those operators and their inequalities in an automated way.
\item A direct treatment of lists or tables of vectors and their
  orthogonalities \cite{peres,planat-saniga-12} as well as the
  diagrams of KS-proofs \cite{planat-12} are notoriously
  cumbersome. A paradigmatic example is Peres' 24-24 set
  \cite{peres}. Peres himself tried to obtain smaller sets via
  a computer program but failed \cite[p.~199]{peres-book}.
  It took three years until Kernaghan obtained one smaller set
  \cite{kern} and two more years until Cabello, Estebaranz, and
  Garc{\'\i}a-Alcaine obtained a second one \cite{cabell-est-96a};
  a straightforward translation of Peres' 24-24 set into an MMP
  24-24 hypergraph string \cite{mpglasgow04-arXiv-0} immediately
  enables one to obtain all 1,233 KS MMP subhypergraphs on any
  laptop in seconds \cite{pmm-2-09},
  \cite[22:00]{pavicic-paris-video-2019};
  a graphical representation of the MMP 24-24
  \cite{mpglasgow04-arXiv-0} even enables one to obtain desired
  subhypergraphs by hand hand in minutes
  \cite[24:00]{pavicic-paris-video-2019}.
\item Parity proofs, developed by Aravind and Waegell and applied
  to contextual sets, read off particular polytopes 
  \cite{waeg-aravind-jpa-11}. However, they exist only for sets
  with an odd number of hyperedges. Still, they are very efficient
  and fast. Their data lists and tables, for both even and odd
  number of hyperedges, can be straightforwardly and automatically
  mapped to MMP hypergraph strings via our programs to enable
  further processing. Notice that for sets with even number of
  hyperedges the MMP hypergraph algorithms remain the only tool.
\item MMP diagrams \cite{pmmm05a,pmmm05a-corr} are predecessors of MMP
  hypergraphs; they required that all hyperedges have the
  cardinality equal to the dimension of the space in which the
  hypergraph vertices reside.
\item Nodes, rays, tests, or vertices in contexts, bases, or edges
  within set or graph approaches are introduced in a series of
  papers but they are vaguely defined and are at odds with the
  standard terminology. In 2014 Lison{\v e}k, Badzi{\c a}g,
  Portillo, and Cabello defined a context as a ``subset of
  jointly measurable tests;'' then as a ``number of bases''
  \cite{lisonek-14}. They graphically present their set
  in \cite[Fig.~1]{lisonek-14}, and we can recognize that a context
  is a (hyper)edge and that a node or a test
  \cite{d-ambrosio-cabello-13} or a ray is a vertex. They do make
  use of the term vertex, but not of the term (hyper)edge. In
  \cite[Fig.~1]{cabello-08} MMP hypergraph (diagram) representation
  from \cite{pmmm05a,pmmm05a-corr} is being used but not cited. In
  \cite[Fig.~1]{cabello-severini-winter-14} MMP hypergraph (diagram)
  representation from \cite{pmmm05a,pmmm05a-corr} are being used but are called
  a ``simplified representation of a graph'' where ``events are
  represented by vertices'' and instead of making use of the term
  hyperedge, they just write that ``sets of pairwise exclusive
  events are represented by vertices in the same straight line or
  circumference rather than by cliques.'' Amaral and Cuncha 
  even mix up graph and MMP hypergraph representations in the same
  figure \cite[Fig.~A.8,~p.115]{amaral-cunha-18}---see
  Fig.~\ref{fig:gr-hyp}(d).
  \begin{figure}[ht]
    \begin{center}
  \includegraphics[width=\textwidth]{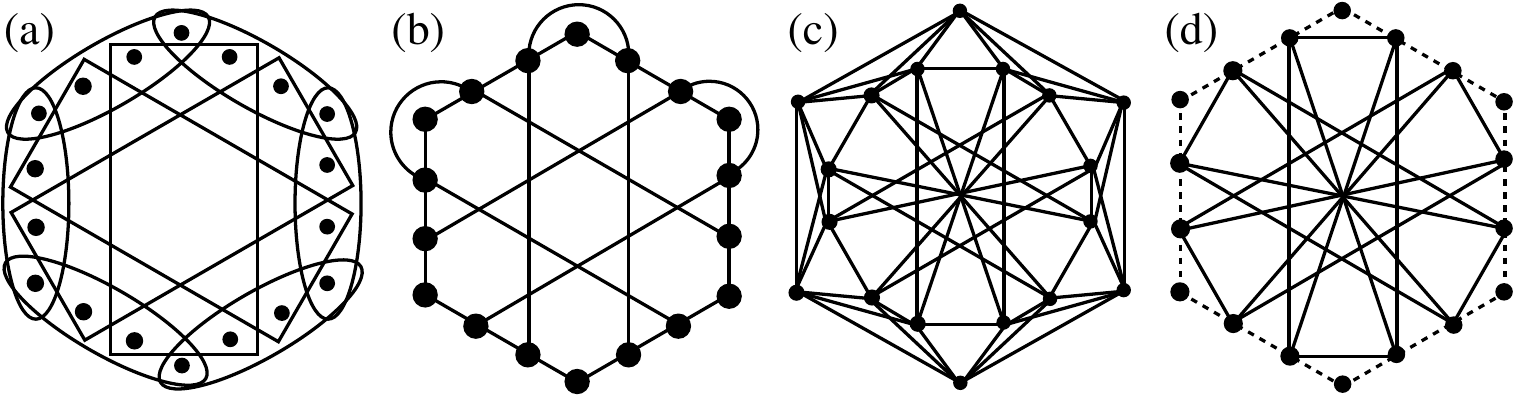}
\end{center}
\caption{(a) General hypergraph representation of the 18-9 KS set
  found in 1996 by Cabello, Estebaranz and Garc{\'\i}a-Alcaine
  \cite{cabell-est-96a} as given in
  \cite[Fig.~A.7,~p.115]{amaral-cunha-18}; (b) the smallest
  (18-9) of all exhaustively generated MMP hypergraphs with the
  (0,-1,1) coordinatization as given in
  \cite[Fig.~3(a)]{pmmm05a,pmmm05a-corr}; it is isomorphic to
  the first 18-9 KS
  set shown in (a); (c) graph representation of the 18-9 set; 
  (d) mixed graph-MMP-hypergraph representation of the set
  as given in \cite[Fig.~A.8,~p.115]{amaral-cunha-18}
  (``The dashed edges correspond to a clique of size 4'').}
\label{fig:gr-hyp}
\end{figure}
  In \cite[Fig.~A.11,~p.117]{amaral-cunha-18} the whole 21-7 6-dim
  MMP hypergraph is called a ``simplified version'' of a graph and
  an MMP hyperedge is said to ``correspond to a clique of  size 6''
  but the term MMP hyperedges is neither mentioned nor cited.
  Budroni, Cabello, G{\"u}hne, Kleinmann and Larsson
  \cite{budroni-cabello-rmp-22} make use of both terms, graphs
  and hypergraphs, interchangeably. They state that ``contexts can
  be represented as graphs, or more generally hypergraphs''
  \cite[p.~30]{budroni-cabello-rmp-22}. Still, they do not mention any
  contextual MMP hypergraph paper published in the last ten years
  \cite{bdm-ndm-mp-fresl-jmp-10,pmm-2-09,mp-7oa,pavicic-pra-17,pm-entropy18,pavicic-entropy-19,pwma-19}
  where billions of contextual 3- to 32-dim MMP hypergraphs were
  generated. Such an approach is deleterious since most known
  contextual vector sets in whatever other formalism turn out to
  be definable by and reducible to MMP hypergraphs or generated by
  them and since nothing comparable has been achieved by any other
  formalism apart from the parity proofs for the KS sets with an
  odd number of hyperedges by Aravind and Waegell in the 4-dim
  Hilbert spaces as derived from polytopes and Lie algebras
  \cite{aravind10,waeg-aravind-pra-13,waeg-aravind-fp-14,waeg-aravind-jpa-15,waeg-aravind-megill-pavicic-11}. Disadvantageously, less than
  5\%\ of all known MMP hypergraphs have parity proofs
  \cite{pavicic-pra-17}.
\item {\em Greechie diagrams\/} have recently been used as a name
  for what are actually MMP hypergraphs
\cite{svozil-18,cabello-svozil-18,svozil-20,budroni-cabello-rmp-22}.
This is a misnomer since Greechie diagrams---connected {\em Hasse
  diagrams\/}---belong to the field of partially ordered sets and
  can represent neither graphs, nor general hypergraphs, nor MMP
  hypergraphs. A Hasse diagram is a graphical representation of a
  poset (partially ordered set)---a collection of whose elements is
  called a {\em block}---where an element $y$ is drawn above an
  element $x$ if and only if $y>x$ ($y$ covers $x$). In a poset with
  the least element 0 an {\em atom\/} is an element that covers it.
  The orthogonality $x\perp y$ is defined as $y'>x$, where $y'$ is
  an orthocomplement of $y$; in a 3-atom poset $y'=x\vee z$;
  $y\vee y'=x\vee y\vee z=1$; $y\wedge y'=x\wedge y\wedge z=0$.
  A Greechie diagram is a shorthand notation for a collection of
  connected Hasse diagrams in which atoms within each block are
  represented as dots and blocks as lines or smooth curves
  connecting them. The following conditions must be
  satisfied: ($\alpha$) All blocks share common 0 and 1;
  ($\beta$) If an atom $x$ belongs to an intersection of blocks and,
  therefore, to both of them, then the blocks also share $x'$;
  ($\delta$) Blocks contain three or more atoms; ($\epsilon$) Two
  blocks may not share more than one atom; ($\zeta$)
  Diagrams cannot contain loops of order 2 or 3
  \cite{greechie71,kalmb74,mp-7oa}. In Fig.~\ref{fig:greechie} we
  give several examples of diagrams that were called Greechie
  diagrams in the literature although they are not Greechie
  diagrams but MMP hypergraphs.
  \begin{figure}[ht]
\begin{center}
  \includegraphics[width=0.97\textwidth]{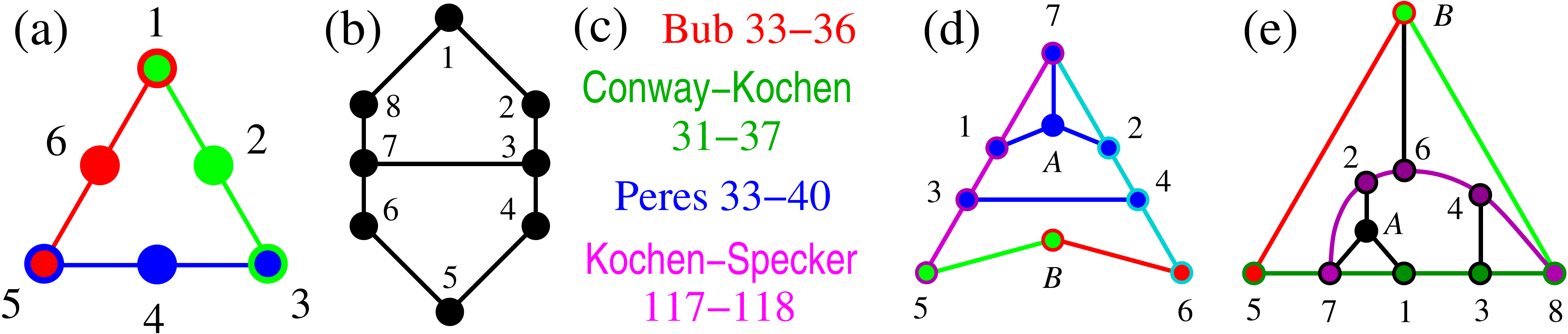}
\end{center}
\caption{(a) \cite[Fig.~2(a)]{svozil-18},
  \cite[Fig.~3(a)]{svozil-20} is not a Greechie diagram since it
  violates condition ($\zeta$); (b) ``bug''
  \cite[Fig.~2(a)]{budroni-cabello-rmp-22} is not a Greechie
  diagram since it violates condition $\delta$---e.g., block
  {\tt 12} contains only two atoms; (c) some 3-dim KS sets with
  vertices that belong to only one hyperedge dropped
  \cite{svozil-tkadlec,svozil-book-ql,budroni-cabello-rmp-22},
  are not Greechie diagrams because they violate condition
  $\delta$; (d) \cite[Fig.~4(a)]{cabello-svozil-18} is not a
  Greechie diagram since it violates conditions $\delta$ and
  $\zeta$---e.g., block {\tt 5B} contains only two atoms and
  loop {\tt 1-A-7} is of order 3; similarly, the ``orthogonality
  hypergraph'' \cite[Fig.~1]{svozil-21} cannot be said to
  represent a ``hypergraph introduced by Greechie''
  \cite[Sec.~II]{svozil-21} for the same $\delta$ and $\zeta$
  reasons; (e) \cite[Fig.~5(a)]{cabello-svozil-18} violates
  conditions $\delta$, $\epsilon$, and $\zeta$---e.g., block
  {\tt 5B} contains only two atoms, blocks {\tt 57138} and
  {\tt 72648} share two atoms ({\tt 7} and {\tt 8}), and loops
  {\tt 2-7-A}, {\tt 1-A-7}, and {\tt 3-8-4} are of order 3.}
  \label{fig:greechie}
\end{figure}
Notice that all MMP hypergraphs from Fig.~\ref{fig:greechie} are
contextual. 

As for Fig.~\ref{fig:greechie}(c), note that full 3-dim KS sets,
Bub's 49-36, Conway-Kochen's 51-37, Peres' 57-40, and
Kochen-Specker's 192-118 are all genuine Greechie diagrams
\cite{mp-7oa}.
\item The term {\em graph of orthogonality\/} has recently been
  used to rename the MMP hypergraph or to avoid using it
  \cite[Supp.~Material, Figs.~3,7,8,9]{cabello-21}, while
  the term {\em orthogonality hypergraph\/} has been used to denote
  the MMP hypergraph arguing that the latter kind of a hypergraph
  is actually a general hypergraph \cite[Sec.~2.4]{bretto-13}
  ``introduced by Greechie'' \cite[Sec.~2.3]{svozil-22}. For
  instance, the KS 10-5 MMP critical hypergraph shown in
  \cite[Fig.~2(b)]{pmmm05a,pmmm05a-corr} Fig.~\ref{fig:pent-cab-6d}(c)---which
  is equivalent to the graph Fig.~\ref{fig:pent-cab-6d}(a,b) given
  in \cite[Supp.~Material, Fig.~7]{cabello-21}---is not
  referred to in \cite{cabello-21}. This set apparently does not
  have a coordinatization in the 4-dim space. An over-complicated
  6-dim coordinatization from 
  $\{0,\pm 1,2,\pm\sqrt{2},\pm\sqrt{3}\}$ components for just 10
  vertices in Fig.~\ref{fig:pent-cab-6d}(a-c) is offered in
  \cite[Eq.~(7)]{cabello-13a} and from
  $\{0,1,-\frac{1}{2}\pm i\frac{\sqrt{3}}{2}\}$ in
  \cite[Supp.~Material, Eq.~(4)]{cabello-21}, but it
  can actually be generated from $\{0,\pm 1,2\}$ components
  (the 6-dim MMP hypergraph itself is shown in
  Fig.~\ref{fig:pent-cab-6d}(d). The latter components   
  yield the following coordinatization: {\tt 1}=(0,0,0,0,0,1),{\tt 2}=(0,0,0,0,1,0),{\tt 3}=(0,1,0,1,0,-1),\break  {\tt 4}=(0,1,0,0,0,1), {\tt 5}=(1,0,0,0,-1,0), {\tt 6}=(0,0,0,1,0,0), {\tt 7}=(0,1,0,-1,0,0), {\tt 8}=(0,0,1,0,0,0), {\tt 9}=(1,0,1,0,0,0), {\tt A}=(1,0,-1,0,1,0), {\tt B}=(0,1,0,0,0,0), {\tt C}=(1,0,0,0,0,0), {\tt D}=(0,1,0,1,0,2),\break  {\tt E}=(1,0,-1,0,0,0), {\tt F}=(0,-1,0,2,0,1),{\tt G}=(1,0,0,0,1,0), {\tt H}=(0,1,0,0,0,-1), {\tt I}=(-1,0,1,0,2,0), {\tt J}=(0,1,0,1,0,0), {\tt K}=(1,0,2,0,1,0).
\begin{figure}[ht]
  \begin{center}
  \includegraphics[width=0.83\textwidth]{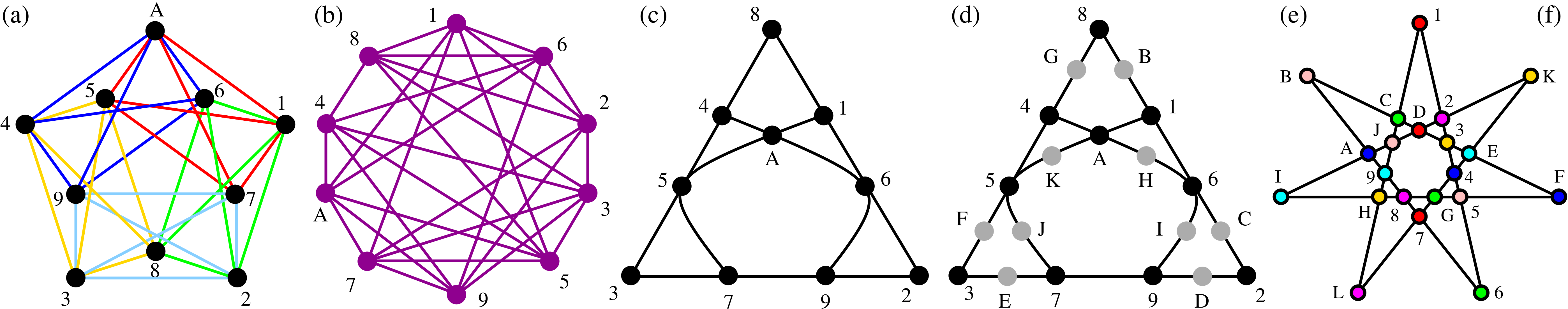}
  \includegraphics[width=0.16\textwidth]{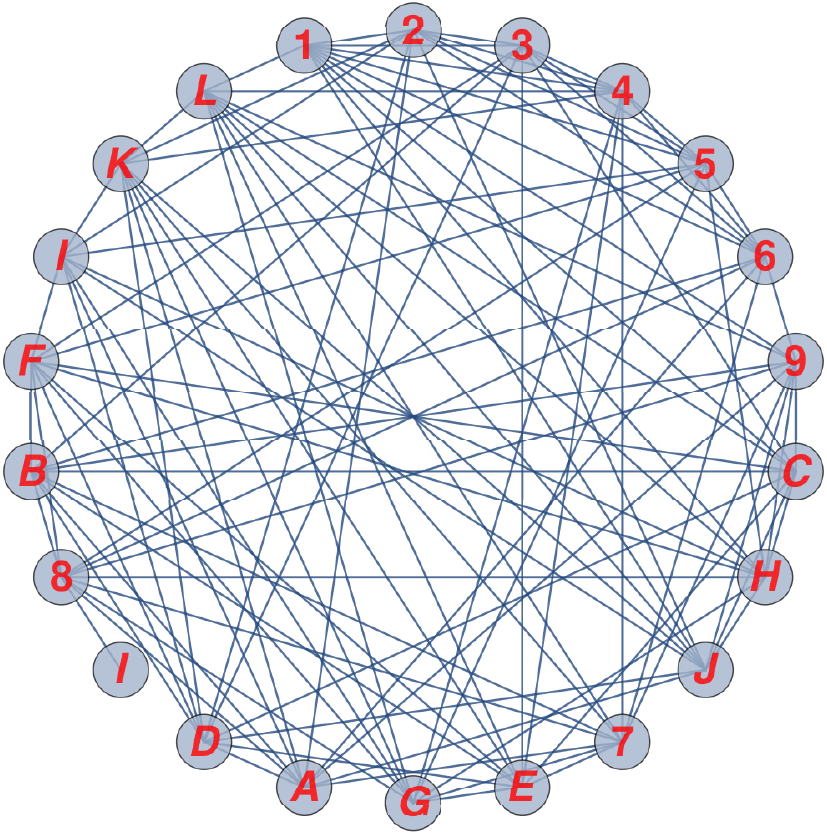}
\end{center}
\caption{(a) Cabello's graph of orthogonality which he also
  recognizes as Johnson graph $J(5,2)$
  \cite[Supp.~Mat., Fig.~7]{cabello-21}; (b) the
  same Johnson graph standardly represented in an automated way
  (e.g., by means of Mathematica's $GraphPlot$) via circular  
  embedding; (c) the same graph represented as an MMP hypergraph
  and first given in \cite[Fig.~2(b)]{pmmm05a,pmmm05a-corr} but
  not referred to neither in \cite{cabello-13a} nor in
  \cite{cabello-21}; the set 10-5 apparently does not have a
  coordinatization; (d) 6-dim 20-5 noncontextual MMP hypergraph
  needed to accommodate the 6-dim
  coordinatization offered in \cite[Eq.~(7)]{cabello-13a} and
  \cite[Supp.~Mat., Eq.~(4)]{cabello-21} to just
  10 vertices from the 10-5 set; (e) a 6-dim 21-7 KS MMP
  hypergraph improperly called a ``graph of orthogonality'' in
  \cite[Supp.~Mat., Figs.~8,9]{cabello-21}; (f) 21-7
  graph of the previous 21-7 MMP hypergraph that should have
  been used in \cite[Supp.~Material, Figs.~8,9]{cabello-21};
  here obtained by Mathematica's $GraphPlot$.}
\label{fig:pent-cab-6d}
\end{figure} 
  Also, the relevance of the recognition of the 10-5 graph
  (a truncated 20-5 graph) as the Johnson graph $J(5,2)$ should
  have been explained.

  On the other hand, the graphics
  presented in \cite[Supp.~Material, Fig.~3]{cabello-21} is
  called a ``graph of orthogonality'' but it is neither a graph nor
  a hypergraph as explained in Sec.~\ref{subsub:3d} and
  Fig.~\ref{fig:3d-bud-ks}. Finally, in \cite[Supp.~Material,
  Figs.~8,9]{cabello-21} the 21-7 MMP hypergraph is shown
  (Fig~\ref{fig:pent-cab-6d}(e)) but improperly called a
  ``graph of orthogonality.'' Its proper graph is shown in
  Fig~\ref{fig:pent-cab-6d}(f). The 21-7 is in
  \cite{cabello-21} recognized to be equivalent to the Johnson
  graph $J(7,2)$, but it is not clear what is the relevance of this
  recognition, since the graph is only one of many from the 216-153
  and 834-1609 KS classes \cite{pm-entropy18}, none of which, apart
  from the 21-7 itself, is a Johnson graph. 

  As for the ``orthogonality hypergraph,'' the term unfolds in
  \cite[Secs.~II,III]{svozil-21} where a graphical appearance
  of Greechie diagrams is offered as a vindication for calling
  MMP diagrams ``orthogonality hypergraphs'' and the MMP language
  the ``hypergraph nomenclature.'' We have already shown in the
  previous bulleted item above that not all MMP hypergraphs
  correspond to Greechie diagrams and that therefore the
  Greechie diagrams cannot be used as their substitutes or
  transliterations. 
\end{itemize}

The most important aspect of MMP hypergraph representation is not its
striking visualization, but its encoding together with its algorithms
and programs that enable their automated generation, handling, and
manipulation. The results obtained with their help gives the MMP
hypergraph language a favorable margin over competitive formalisms. 

\section{\label{sec:obj-met}Objectification and
generation of contextual sets}

As we pointed out in the Introduction, the contextual sets
juxtapose measurement outputs of quantum systems with
those of classical ones. We consider them as represented
by MMP hypergraphs and such a representation can be given
several different objectifications starting with the well
known Kochen-Specker one, that in turn can be generated by
several methods in any $n$-dim space as abundant as needed. 

\subsection{\label{sec:h-ext}The {K}ochen-{S}pecker
  theorem and its Extensions}

In this section we consider the Kochen-Specker (KS) theorem,
its equivalent operator-defined Bell theorem, and its three
extensions: the KS MMP hypergraph theorem, the Weakened Bell
theorem, and the Non-Binary MMP hypergraph definition.

\smallskip
\parindent=0pt
\begin{Xeorem}
  {\bf Kochen-Specker} {\bf (KS)}
  {\rm \cite{koch-speck,zimba-penrose,fine-teller-78}.}
  In ${\cal H}^n$, $n\ge 3$, there are sets of $n$-tuples of
  mutually orthogonal vectors to which it is impossible to
  assign $1$s and $0$s in such a way that
\begin{enumerate}[label=(\roman*)]
\item No two orthogonal vectors are both assigned the value $1$;
\item In no $n$-tuple of  mutually orthogonal vectors, all of
the vectors are assigned the value $0$.
\end{enumerate}
\label{th:KS}
\end{Xeorem}

\smallskip
\parindent=9pt
These sets are called KS {\em sets\/} and the vectors
KS {\em vectors\/}. Since KS sets are constructive
counterexamples that prove the KS theorem, some
authors in the literature call them KS {\em proofs}, e.g.,
\cite{waegell-aravind-fp11,waeg-aravind-pra-13,waeg-aravind-megill-pavicic-11,yu-tong-15}.

The first extension of the KS theorem and KS sets is the one which
makes use of MMP hypergraphs whose vertices are not represented by
vectors.

\smallskip
\parindent=0pt
\begin{Xeorem}
  {\bf KS MMP hypergraphs}. There are {\rm MMP} hypergraphs of
  $\,${\em hypergraph-dimension} $n\ge3$ (Def.~\ref{def:MMP-string})
  to whose vertices it is impossible to assign $1$s and $0$s in
  such a way that
\begin{enumerate}[label=(\roman*$\>'\!$)]
\item No two vertices from the same hyperedge are both assigned
  the value $1$;
\item In no hyperedge, each containing $n$ vertices, all of
  the vertices are assigned the value $0$.
\end{enumerate}
\label{th:h-KS}
\end{Xeorem}

\smallskip
\parindent=9pt

That means that there might be KS MMP hypergraphs that are not
KS sets, as e.g., {\tt 1234,4561,2356.} \cite[Fig.~1]{pm-entropy18}.
It is not a KS set because vectors that would represent its vertices
do not exist, i.e., this KS MMP hypergraph does not have a
coordinatization (vector representation). Our algorithms and
programs are partly based on Theorem \ref{th:h-KS} meaning that they
can detect the contextuality of an MMP hypergraph no matter whether
its coordinatization is given (or even existent) or not. Handling KS
MMP hypergraphs without taking their coordinatization into account
give us a computational advantage over handling them with a
coordinatization because processing bare hypergraphs is faster than
processing them via vectors assigned to their vertices.

In Refs.~\cite{yu-oh-12,cabello-bengtsson-12,pavicic-entropy-19}
$n$-dim contextual sets with hyperedges containing less than $n$
vertices that still violate the rules (i) and (ii) above are
considered. They are not KS sets \cite{pavicic-pra-17}, though,
since the KS theorem \ref{th:KS} assumes that each hyperedge
contains $n$ vertices. We call such sets {\em non-binary}
MMP hypergraphs (see Def.~\ref{def:n-b}).

The original KS theorem holds for vectors and defines a KS set as
its constructive proof. On the other hand, Bell proved the so-called
Bell theorem \cite{bell66} as a corollary to the Gleason theorem
\cite{gleason} which is a projector formulation of the KS theorem
\cite{fine-teller-78,marzlin-15}.

\smallskip

\parindent=0pt
\begin{Xeorem}
  {\bf Bell}. In ${\cal H}^n$, $n\ge 3$, there is no valuation
  function $v$ defined on the projectors $P_i$ on the
  one-dimensional subspaces such that
\begin{enumerate}[label=(\roman*$\>''\!$)]
\item $v(P_i) = 0$ or $1$ for each $i$ and
\item $\sum_{i\in B}v(P_i)=1$ for each orthonormal basis $B$ of
the space ${\cal H}^n$.
\end{enumerate}
\label{th:bell}
\end{Xeorem}

\parindent=9pt
\smallskip
Fine and Teller gave the following extensions of the Bell theorem
for general observables $A$ and $B$ instead of projectors via
the following rules \cite{fine-teller-78}.

\smallskip

\parindent=0pt
\begin{Xeorem}
  {\bf Weakened Bell}. In ${\cal H}^n$, $n\ge 3$, there exists no
  valuation function $v$ defined on the general observables
  $A,B,\dots$ such that
\begin{enumerate}[label=(\roman*$\>'''\!$)]
\item $v(A)$ is an eigenvalue of $A$ {\em (spectrum rule)}; and either
\item $v(A+B)=v(A)+v(B)$, for commuting $A$ and $B$
  {\em (finite sum rule)}, or
\item $v(AB)=v(A)v(B)$, \quad for commuting $A$ and $B$
   {\em (finite product rule)}.
\end{enumerate}
\label{th:bell-sum-pro}
\end{Xeorem}

\parindent=9pt

Therefore, the statement by Yu, Guo, and Tong that the rules
``are usually called the KS theorem'' \cite{yu-tong-15} seems to
be incorrect. On the other hand, they do claim that the operator
formulation of the Kochen-Specker theorem via sum and product
rules is more general than the vector/ray/hypergraph one but they
do not give a rigorous proof of the claim, such as providing us
with a sum- or product-rule KS set that cannot be obtained from
the standard KS theorem. In any case, if the Bell theorem were
``weakened'' by the rules, then it would cease to be equivalent
to the KS theorem and we would possibly lose the universal 0-1
valuation of the valuation function for arbitrary observables
of considered contextual sets.

Since we want to keep the 0-1 valuations for computational purposes,
we shall not pursue the Bell sum/product extension further. Instead,
we introduce another extension of KS sets and MMP hypergraphs based
on hypergraph vertices and their 0-1 valuation. That circumvents the
eigenvalue problem and gives us structural properties of MMP
hypergraphs as well as their measures, valuations, and inequalities
via automated procedures, in contrast to many other valuations and
inequalities obtained in the literature, mostly by hand, for
particular vectors and projectors.

Our hypergraph extension applies to the MMP hypergraph conditions
(i$'$) and (ii$'$) from Theorem \ref{th:h-KS}. It simply consists in
allowing particular hyperedges to contain less than $n$ vertices, as,
e.g., those in\ \ \cite{klyachko-08,yu-oh-12,pavicic-entropy-19}.
Notice that the extension covers both KS and non-KS MMP hypergraphs.
In such extended $n$-dim MMP hypergraphs, hyperedges need not always
contain $n$ vertices but they might still satisfy or violate the two
rules (i$'$) and (ii$'$) from Theorem \ref{th:h-KS}. However, then
we cannot call them the KS rules. Equally so, we cannot call their
inequalities the KS inequalities. Instead, we consider hypergraphs
with truncated vertices in the following way
(cf.~\cite{pavicic-entropy-19}).

\smallskip
\parindent=0pt
\begin{definition} {\bf Non-binary MMP hypergraph.}
  A $k$-$l$ {\rm MMP} hypergraph of hypergraph-dimen\-sion $n\ge 3$
  with $k$ vertices and $l$ hyperedges, whose $i$-th hyperedge
  contains $\kappa(i)$ vertices\break  $(2\le\kappa(i)\le n$,\
  $i=1,\dots,l)$ to which it is impossible to assign {\rm 1}s and
  {\rm 0}s in such a way that the following {\em hypergraph rules}
  hold
\begin{enumerate}[label=(\roman*)]
\item No two vertices within any of its hyperedges are both
assigned the value $1$;
\item In any of its hyperedges, not all of the vertices are
assigned the value $0$.
\end{enumerate}
is called a {\em non-binary hypergraph} ({\rm NBMMP} hypergraph).
\label{def:n-b}
\end{definition}

\parindent=9pt

The MMP hypergraph above is defined without a coordinatization,
i,e., neither their vertices nor their hyperedges are related to
either vectors or operators. We say that an MMP hypergraph is in
an MMP-hypergraph-$n$-dim space (we call it an MMP hypergraph space)
when either all its hyperedges contain $n$ vertices or when we 
add vertices to hyperedges so that each contains $n$ vertices. Many
of our programs handle MMP hypergraphs without any reference to
either vectors or projectors. In an MMP hypergraph with a
coordinatization an $n$-dim MMP hypergraph space becomes an $n$-dim
Hilbert space spanned by a maximal number of vectors within
hyperedges. Whether we speak about an MMP hypergraph with or
without a coordinatization will be clear from the context.

\smallskip
\parindent=0pt
\begin{definition} {\bf Binary MMP hypergraph}
  An {\rm MMP} hypergraph to which it is possible to assign $1$s and
  $0$s so as to satisfy the rules {\rm(i)} and {\rm(ii)} of
  Def.~\ref{def:n-b} is called a {\em binary hypergraph}
  ({\rm BMMP} hypergraph).
\label{def:bin}
\end{definition}

\smallskip
\parindent=9pt

\smallskip
\parindent=0pt
\begin{definition} {\bf Critical MMP hypergraph} is a non-binary
  {\rm MMP} hypergraph which is minimal in the sense that removing
  any of its hyperedges turns it into a binary MMP hypergraph.
\label{def:critical}
\end{definition}

\smallskip
\parindent=9pt

Critical MMP hypergraphs represent non-redundant blueprints for
their implementation since bigger MMPs that contain them only add
orthogonalities that do not change their non-binary property.

An $n$-dim non-binary $k$-$l$ MMP hypergraph ${\cal{H}}$ need not
have a coordinatization, but when it does, the vertices in every
hyperedge have definite mutually orthogonal vectors assigned to
them. That means that each hyperedge $E_j$, $j=1,\dots,l$ should
have not only $\kappa(j)$ vectors corresponding to its $\kappa(i)$
vertices specified, but also $n-\kappa(j)$ ones that must exist by
the virtue of orthogonality in the $n$-dim space so as to form an
orthogonal basis of the space. Such an extended ${\cal{H}}$ we
call a {\em filled\/} ${\cal{H}}$.

In order to handle the pentagon hypergraph {\tt 12,23,34,45,51.}~we
first have to assign a coordinatization of the filled pentagon
{\tt 162,273,384,495,5A1.}~so as to be able to  implement the
pentagon. Once we determined all its vectors/states we can discard
the vertices {\tt 6,7,8,9,A} from further consideration and
processing. 

The following Lemma follows straightforwardly.

\smallskip
\parindent=0pt
\begin{lemma}{\bf KS MMP vs.~non-binary MMP}. \ An $n$-dim
  non-binary {\rm MMP} hypergraph---in which each hyperedge
  contains $n$ vertices---with a coordinatization is a
  {\rm KS MMP} hypergraph.
\end{lemma}

\parindent=9pt
In other words, a standard Kochen-Specker set is a non-binary
MMP hypergraph in which the size of each hyperedge is $n$ and
in which every one of $k$ vertices corresponds to a vector/state.   

\smallskip

\subsection{\label{subsec:methods}Methods of
  MMP hypergraph generations}

To generate NBMMPHs in $n$-dim spaces we mostly rely on
three methods---{\bf M1-M3} which make use of algorithms and
programs developed in
\cite{pavicic-pra-17,pavicic-entropy-19,pwma-19,pavicic-pra-22}.

{\bf M1} consists in an automated dropping of vertices contained
in single hyperedges (multiplicity $m=1$; see Def.~\ref{def:qh-mi})
of BMMPHs and a possible subsequent stripping of their hyperedges.
The obtained smaller MMPHs are often NBMMPH although never KS.  

{\bf M2} consists in an automated random addition of hyperedges to
MMPHs so as to obtain larger ones which then serve us to generate
smaller KS MMPHs by stripping hyperedges randomly again;

{\bf M3} consists in combining simple vector components so as to
exhaust all possible collections of $n$ mutually orthogonal $n$-dim
vectors. These collections form big {\em master\/} MMPHs which
consist of single or multiple MMPHs of different sizes. Master
MMPHs may or may not be NBMMPH, what we find out by applying
filters to them. NBMMPHs serve us to massively generate a
{\em class} of smaller MMPHs via our algorithms and programs.

The algorithms the methods employ are applicable to any
$n$-dim space but the computational barrier of the present day
supercomputers allowed our programs based on these algorithms
to reach no further than a 32-dim space. In this paper we
limit ourselves to 6-dim spaces. 

\section{\label{sec:discr}Discriminators of the
  contextuality: Noncontextuality 
  inequalities}

In general, the noncontextuality inequality as a
distinguisher between contextual and noncontextual MMP
hypergraphs is defined as follows:

\begin{definition}
  {\bf Noncontextuality inequality} defined for an
  arbitrary\/ {\rm MMP} hypergraph ${\cal{H}}$
  \begin{eqnarray}
    \mathrm{A}<\Omega,
  \label{def:non-c-i-s}
  \end{eqnarray}
  where $\mathrm{A}$ and $\Omega$ are some terms defined on
  ${\cal{H}}$, is an inequality whose satisfaction implies that
  ${\cal{H}}$ is a contextual non-binary\/ {\rm MMP} hypergraph
  (Def.~\ref{def:n-b}), and whose violation:
  \begin{eqnarray}
    \mathrm{A}\ge\Omega,
  \label{def:non-c-i-v}
  \end{eqnarray}
  implies that ${\cal{H}}$ is a noncontextual binary {\rm MMP}
  hypergraph (Def.~\ref{def:bin}).
\label{def:non-c-i}
\end{definition}

Noncontextuality inequalities exist in operator- and
hypergraph-based forms.

\subsection{\label{sec:op-ineqal}Operator-based
  inequalities}

Before we dwell on our hypergraph- and vector-evaluation-based
MMP structures and their inequalities let us first briefly present
how the operator-based ones are defined in the literature. There
are three approaches:

\begin{enumerate}[label=(\roman*)]
\item{\em hyperedge-based approach} considers the operators
  defined via vertices organized within hyperedges of MMP
  hypergraphs; such are the majority of operators in the present
  section; noncontextuality inequalities generated by these
  operators correspond to the KS MMP hypergraph hyperedge
  inequalities, abbreviated e-inequalities and defined by
  Defs.~\ref{def:ein} and \ref{def:einm} in
  Sec.~\ref{sec:structure};
\item{\em vertex-based approach} considers the operators
  defined via vertices of an MMP hypergraph by particular algorithms
  and independently of the organization of vertices within
  hyperedges; see Yu-Oh's operators below; their contextuality
  corresponds to the contextuality of the underlying MMP hypergraph; 
  the inequalities generated by these operators correspond to the
  KS MMP hypergraph the $\alpha_r^*$-inequality given by
  Eq.~(\ref{eq:alpha-alpha}); 
\item{\em direct approach} considers operators defined
  without a reference to either vertices or their coordinatization;
  see Peres-Mermin's operators below; we design a
  vertex-hyperedge structure for them in Sec.~\ref{subsec:mermin}.
\end{enumerate}

Most of the operators are defined so as to have eigenvalues
$\pm 1$; their classical counterparts are classical observables
with noncontextual values $\pm 1$. Thus, in the conditions (i) and
(ii) of the KS theorem, the value 1 assigned to vertices/vectors of
a KS set corresponds to an operator-defined variable value 1 (or -1)
and value 0 corresponds to value -1 (or 1), meaning that in (i) and
(ii) we would have 1 (or -1) assigned to one of the vertices and -1
(or 1) to all the others.

Cabello defines 4-dim operators by means of KS states/vectors
$A_{ij}$ \cite[Eq.~(2)]{cabello-08}
\begin{eqnarray}
  A_{ij}=2|v_{ij}\rangle\langle v_{ij}|-I
\label{eq:a-cab}
\end{eqnarray}
via vector coordinatization of the 4-dim KS 18-9 hypergraph
shown in \cite[Fig.~1]{cabello-08} and then he shows that the
inequality defined in \cite[Eq.~(1)]{cabello-08}
\begin{eqnarray}
  -\langle A_{12} A_{16}A_{17} A_{18}\rangle-\dots-\langle A_{29}
  A_{39} A_{59} A_{69}\rangle\le\ 7,
\label{eq:a-cab2}
\end{eqnarray}
defined on the smallest 4-dim KS set 18-9 \cite{cabell-est-96a},
is violated by probabilities of the outcomes of quantum
measurements which give 9 at the right hand side of the
inequality. Value 7 in the inequality is the maximum value
we obtain when we interpret the observables $A_{ij}$ as
classical variables.

Yu and Oh use a similar operator defined for 13 vectors from
a 25-16 non-binary MMP hypergraph (which is a non-KS set, though)
and define the following inequality for them \cite[Eq.~(4)]{yu-oh-12}
\begin{eqnarray}
  \sum_i^{13}\langle A_i\rangle-
  \frac{1}{4}\sum_{i,j}^{13}\Gamma_{ij}\langle A_iA_j\rangle\le 8
\label{eq:yu-oh}
\end{eqnarray}
where $\Gamma$ is a weight function. It is violated by quantum
measurements which yield $25/3=8.\dot3$ for the left hand side
of the inequality.

Badzi{\'a}g, Bengtsson, Cabello, and Pitowsky define
$n$-dim operators by means of states/vectors of a $k$-$l$
KS hypergraph (with $k$ vertices and $l$ hyperedges).
\begin{eqnarray}
  A_i^j=I-2|v^{j,i}\rangle\langle v^{j,i}|,
\label{eq:b-cab2}
\end{eqnarray}
where $\langle v^{j,i}|v^{j,i'}\rangle=\delta_{ii'}$ for every
$1\le j\le l$ \cite[Eq.~(5)]{badz-cabel-09}. They
calculate the following operator expression
\begin{eqnarray}
\beta_q(n,l)=\sum_{j=1}^l\left\langle\left(\sum_{i=1}^nA_i^j-
  \prod_{i=1}^n(I+A_i^j)\right)\right\rangle=\langle l(n-2)I\rangle=l(n-2),
\label{eq:b-cab22}
\end{eqnarray}
which is the result one obtains by a quantum measurement
\cite[Eq.~(8)]{badz-cabel-09}. Classical observables must
satisfy the inequality \cite[Eq.~(1)]{badz-cabel-09}
\begin{eqnarray}
\beta_c(n,l)\le l(n-2)-2,
\label{eq:b-cab-c}
\end{eqnarray}
which is violated by the quantum operators. Badzi{\'a}g,
Bengtsson, Cabello, and Pitowsky then calculate the
classical $\beta_c(n,l)$ for Peres' 24-24 KS set and obtain
Max$[\beta_c(4,24)]=40$ which is clearly violated by the quantum
mechanical $\beta_q(4,24)=24(4-2)=48$.

Yu, Guo, and Tong define noncontextuality KS inequalities
\cite[Eqs.~(3,7,10)]{yu-tong-15} for operators and projectors
which are implicitly defined via vectors of KS sets, but they
do not specify any of them. We can only say that their
inequality \cite[Eq.~(7)]{yu-tong-15} is equivalent to
our e$_{Max}$-inequality in Def.~\ref{def:ein}.

Peres and Mermin's \cite{peres90,mermin90} set was used
\cite{cabello-08} to yield noncontextuality inequalities for
operators which are not constructed with the help of
vectors/states that might underlie them. Instead, one makes use
of the tensor products of the 2-dim Pauli operators given by
Eq.~(\ref{eq:mermin}) and defines the noncontextuality
inequality as follows. The operators $\Sigma_i$, $i=1,\dots 9$,
satisfy the equation
\begin{eqnarray}
  \Sigma=\Sigma_1\Sigma_2\Sigma_3+\Sigma_4\Sigma_5\Sigma_6+
  \Sigma_7\Sigma_8\Sigma_9
  +\Sigma_1\Sigma_4\Sigma_7
  +\Sigma_2\Sigma_5\Sigma_8-\Sigma_3\Sigma_6\Sigma_9=6I
\label{eq:square-S}
\end{eqnarray}
while their classical counterparts $S_i$,
$i=1,\dots 9$ (observables with two possible results
$\pm 1$) satisfy
\begin{eqnarray}
   S= S_1S_2S_3+ S_4S_5S_6+
   S_7S_8S_9
  + S_1S_4S_7+
   S_2S_5S_8- S_3S_6S_9\le 4
\label{eq:square-s}
\end{eqnarray}
Thus, the noncontextuality inequality should read
\begin{eqnarray}
  \langle S\rangle\le 4\le\langle\Sigma\rangle=6.
\label{eq:square-sin}
\end{eqnarray}

Taken together, most operator-based inequalities in the literature
rely on coordinatizations of vertices/states of MMP hypergraphs by
means of which the operators are defined and then measured through
an application on arbitrary states. In contrast, in the next
sections, we consider the inequalities which are defined directly by
means of the coordinatizations of vertices of hypergraphs which are
measured directly.

\subsection{\label{sec:hyp-op}Hypergraph-based
  inequalities; Hypergraphs vs.~operators}

The approaches in Sec.~\ref{sec:op-ineqal} consider vertices either
directly (Yu and Oh), or via their inclusion in hyperedges (other
approaches). In Sec.~\ref{sec:structure} we connect the
operator-based approach vertex structure of contextual sets with a
hypergraph-based approach. But in order to introduce particular
vertex-based features of the latter approach, in this section we
reconsider a simple set---Klyachko, Can, Binicio{\u g}lu, and
Shumovsky's pentagon \cite{klyachko-08,cabello-13}---to pinpoint
the required features and to serve us as an introduction to
Sec.~\ref{sec:structure}. 

As we pointed out in \cite{pavicic-entropy-19} the 3-dim 
10-5 MMP pentagon from whose hyperedges the vertices that belong to
just one hyperedge ({\tt 162,273,384,495,5A1.}, shown in
Fig.~\ref{fig:pentagon}(a)), are dropped is a contextual non-binary
5-5 MMP pentagon ({\tt 12,23,34,45,51.}). 

\begin{figure}[h]
  \begin{center}
  \includegraphics[width=1\textwidth]{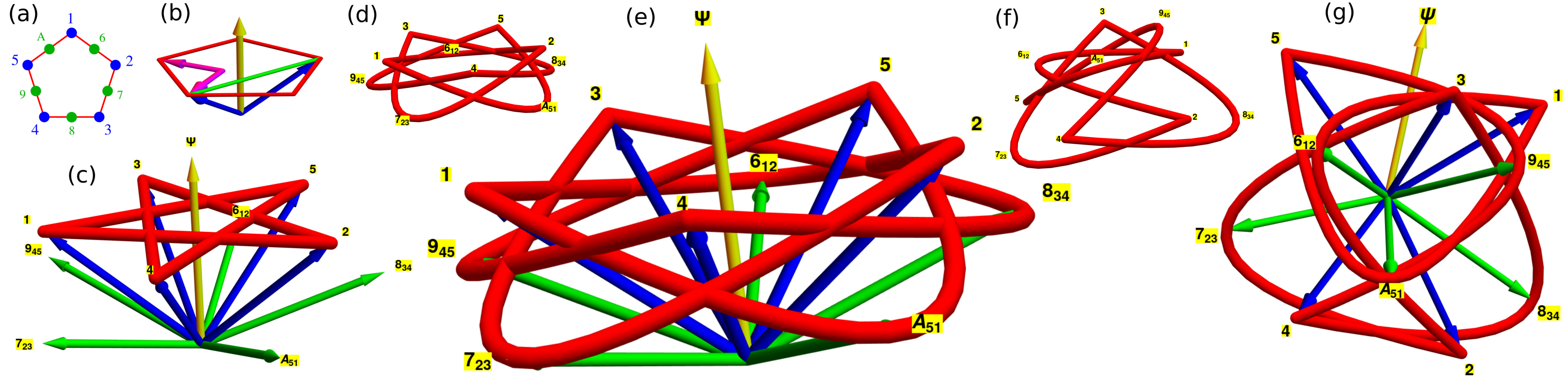}
\end{center}
\caption{(a) MMP hypergraph representation of the pentagon;
  (b) impossible vector representation of a planar regular pentagon;
  (c) impossible star-shaped planar vector-pentagon;
  (d) Proper non-planar vector-pentagon;
  (e) its vector representation;
  (f) bare form of a pentagon obtained from $\{0,\pm 1,2\}$
  vector components; see text;
  (g) its vector representation;
  see text.}
\label{fig:pentagon}
\end{figure}

Operator vs.~hypergraph approaches to the pentagon will serve us
to introduce a distinction between operator-based measurements
and direct measurements of quantum systems exiting quantum ports
determined via vertices within each hyperedge of an MMP hypergraph.

We can assign vectors to the pentagon vertices in many ways.
Fig.~\ref{fig:pentagon}(b) shows that 5-5 cannot be a regular
planar pentagon, irrespective of chosen vectors, since the mutually
orthogonal vectors which would span its hyperedges, cannot reach it
from the $z$-axis (left magenta pair). The green diagonals do allow
for such vectors (blue ones at the bottom) and Klyachko et
al.~\cite[Fig.~1]{klyachko-08} attempted to use the star-shaped
pentagon (Fig.~\ref{fig:pentagon}(c)) instead and called it a
pentagram. However, since the pentagon is a 3-dim one, the vectors
that belong to just one hyperedge  should also be taken into account,
what makes the pentagram contained in a plane inconsistent as it is
obvious from Fig.~\ref{fig:pentagon}(c). The proper pentagon with
curved hyperedges is given in Fig.~\ref{fig:pentagon}(d). Together
with vectors that span them it is shown in Fig.~\ref{fig:pentagon}(e)
and vectors themselves are given above Fig.~\ref{fig:pent-cab-6d}.
The need to take all the states into account experimentally was also
stressed in \cite{klyachko-08}. Since vectors/states belong to the
spin-1 system all three ports of each gate should be measured even
if only two of the outcomes are postprocessed.

Now, Klyachko et al.\ \cite{klyachko-08}~consider the states
corresponding to vectors {{\tt 1},\dots,{\tt 5}} in classical
vs.~quantum representations. Assumed classical measurements demand
that each vertex within an hyperedge either receives an experimental
detection or not, i.e., that it is assigned a value 1 or 0 (a
preassigned truth value), in such a way that the above hypergraph
rules (i) and (ii) from Def.~\ref{def:n-b} hold. ``When the same
assignments are carried over to the projectors in the pentagram
operator [$\mathcal A$]\dots at most two of them can be
assigned the value 1 [in our notation below
$HI_{cM}=2$; Def.~\ref{def:ch-i}]. In a
noncontextual reality an experimenter\dots will therefore
always find that \cite[p.~415, Eq.~(3)]{badz-cabel-larss-11}
\begin{eqnarray}
 \langle {\mathcal A}_c\rangle\le 2.{\textrm{''}}
\label{eq:pent-c}
\end{eqnarray}

In the quantum representation, the operators are
$|{\tt i}\rangle\langle{\tt i}|$ and the maximum of the mean
value for $\Psi=(0,0,1)$ is:
\begin{eqnarray}
  \langle {\cal A}_q\rangle_\Psi
  =\sum_{{\tt i}=1}^5 |\langle {\tt i}|\Psi\rangle|^2
  =\sqrt{5}\approx 2.236>2
\label{eq:pent-q}
\end{eqnarray}

Its minimum value $\frac{5-\sqrt{5}}{2}\approx 1.382$ we obtain
for $\Psi_m=(1,0,0)$ and these dependencies of mean values of the
measured observable on the chosen states render the pentagram
setup state-dependent in the operator approach.

But here we point out to two features of the pentagon.

First, we can generate vectors of the 10-5 in an automated way
(as in \cite{pavicic-entropy-19}) from simple vector components,
$\{0,\pm 1,2\}$, so as to obtain {{\tt 1}=(0,0,1),
  {\tt 2}=(0,1,0),
  {\tt 3}=(1,0,1)/$\sqrt{2}$, {\tt 4}=(1,1,-1)/$\sqrt{3}$,
  {\tt 5}=(1,-1,0)/$\sqrt{2}$, {\tt 6}=(1,0,0),
  {\tt 7}=(1,0,{-1})/$\sqrt{2}$, {\tt 8}=(-1,2,1)/$\sqrt{6}$,
  {\tt 9}=(1,1,2)/$\sqrt{6}$, {\tt A}=(1,1,0)/$\sqrt{2}$.
  Vectors {\tt 1,2,\dots,5} do not determine a plane. The
  pentagon is shown in Fig.$\,$\ref{fig:pentagon}(f,g).
  For vector $\psi=(3.15,-8.46,8.46)$ we obtain 
  $\langle{\cal A}_{qMax}\rangle_{_\psi}=2.23>2$.
  The full 10-5 MMP hypergraph is a binary one, i.e., a non-KS set.
  $A_1\cdot A_6\cdot A_2=\dots=A_5\cdot A_A\cdot A_1=I$, where
  $A_i=|{\tt i}\rangle\langle{\tt i}|$, gives:
\begin{eqnarray}
 \langle {\mathcal A}_c[10{\textrm-}5]\rangle=5\langle I\rangle=5
\label{eq:pent-c5}
\end{eqnarray}

Second, the 5-5 pentagon is a NBMMP hypergraph and we can
make it hyper\-graph-state-independent in the following sense.
It can be implemented via a generalised Stern-Gerlach experiment
which makes use of both magnetic and electric fields
\cite{anti-shimony}, or via photonic triplets
\cite{zu-wang-duan-12}, or via photon orbital angular
momentum \cite{d-ambrosio-cabello-13}. From each gate
represented by pentagon hyperedges, a particle or a photon will
exit through one of the three ports and will be detected by a
corresponding detector. We postprocess the data so as to keep the
records of the ``clicks'' triggered by {\tt 1,2,\dots,5} events and
discard those triggered by {\tt 6,7,\dots,A} events. After a
recalibration of data, the probability of obtaining a click
triggered by {\tt 1} or by {\tt 2} while measuring
the {\tt 162} hyperedge is 1/2. Additionally, the probability of
obtaining a click while measuring {\tt 1} within the {\tt 5A1}
hyperedge is also 1/2 as well as the probability of obtaining a click
while measuring {\tt 2} within the {\tt 273} hyperedge,
and so on for all other ports/vertices. Therefore, the sum of
probabilities of registering any of the {\tt 1,2,\dots,5} events in
pairs of hyperedges they belong to is 5. In notation of
Sec.~\ref{sec:structure} $HI_q=5$; Def.~\ref{def:qh-i}. Since we
can assign at most two classical 1s (satisfying the conditions (i)
and (ii)) from Def.~\ref{def:n-b} to pentagon vertices and since
each of them share two hyperedges we have
$HI^m_{cM}=2\times 2=4$  and we obtain the following
v-inequality (Def.~\ref{def:iin}):
\begin{eqnarray}
HI^m_c[5{\textrm-}5]=4<HI_q[5{\textrm-}5]=5
\label{eq:in-25}
\end{eqnarray}

Notice that the non-KS filled 10-5 pentagon violates it:
\begin{eqnarray}
  HI^m_c[10{\textrm-}5]=5=HI_q[10{\textrm-}5]=5
\label{eq:in-55}
\end{eqnarray}
The violation occurs because the sum of probabilities for
{\tt 1,2,\dots,5} is 10/3 and for {\tt 6,7,\dots,A} it is 5/3
which together make 5. The maximum number of classical 1s is 5
(each positioned in one of {\tt 6,7,\dots,A}). So, a pentagon
hypergraph inequality is hypergraph-state-independent in the
sense that it relies on the MMP hypergraph structure and not
on its coordinatization.  

Thus, there are three things we take from here.
\begin{enumerate}
\item While the operator-based representation of the
  pentagon is state dependent, the hypergraph one is not.
\item In the operator-based representation each state contributes
  just once in the measurements, i.e., via projections to $\Psi$,
  while in the MMP hypergraph representation each state/vertex
  contributes twice, once through a measurement of a port contained
  in a chosen hyperedge/gate and then through a measurement
  of the same port contained in the next hyperedge/gate it shares.
  In Sec.~\ref{sec:structure} we formalize the hypergraph notions
  we introduced here.
\item The non-binary MMP hypergraph 5-5 pentagon is not a
  subhypergraph of the binary MMP hypergraph pentagon 10-5.
  In order to obtain 5-5 from 10-5 we discard vertices that are
  contained in just one hyperedge while making use of the full
  coordinatization of 10-5 to assign vectors to remaining vertices.
  We denote such a subset as a $\overline{\rm subhypergraph}$. 
\end{enumerate}

\begin{definition}\label{def:over-sub}
  $\overline{\bf Subhypergraph}$ is a subset of an {\rm MMP}
  hypergraph ${\cal H}=(V,E)$ from which an arbitrary number of
  vertices contained in just one hyperedge are taken out so as to
  satisfy the conditions of Def.~\ref{def:MMP-string}, i.e., so
  as to be an {\rm MMP} hypergraph itself.  
\end{definition}

\subsection{\label{sec:structure}Hypergraph structure
        and inequalities}

In this section we elaborate on MMP hypergraph vertices and the
MMP hypergraph structure and properties based on them.

We start with the following definitions. 

\begin{definition}
  {\bf Vertex multiplicity} is the number of hyperedges vertex
  $i$ belongs to; we denote it by $m(i)$.
\label{def:qh-mi}
\end{definition}

\begin{definition}
{\bf MMP classical vertex index} $HI_c$ is the number of\/ $1$s
one can assign to vertices of an\/ {\rm MMP} hypergraph, non-binary
or binary, so as to satisfy the conditions {\rm(i)} and {\rm (ii)}
from Def.~\ref{def:n-b}. Maximal (minimal) $HI_c$ is denoted as
$HI_{cM}$ ($HI_{cm}$).
\label{def:ch-i}
\end{definition}

(Notice that in {\rm\cite{pavicic-entropy-19}} some values of
$HI_c$ are wrongly calculated due to an application problem of
our previous algorithm and program; program {\textsc{One}} used
in this paper is a substitute for that ones.)

\begin{definition}
  {\bf MMP classical multiplexed vertex index} $HI^m_c$ is the
  number one obtains when summing up all multiplicities of
  vertices of an {\rm MMP} hypergraph with all hyperedges
  contining $n$ vertices, non-binary or binary, to
  which one can assign $1$s so as to satisfy the conditions
  {\rm(i)} and {\rm (ii)} from Def.~\ref{def:n-b}. Maximal
  (minimal) $HI^m_c$ is denoted as $HI^m_{cM}$ ($HI^m_{cm}$).
\label{def:ch-mi}
\end{definition}

\smallskip
We obtain $HI_c$ and $HI^m_c$ by an algorithm and its program
{\textsc{One}} which assign 1s to vertices of an MMP hypergraph.
The algorithm randomly searches for a distribution of 1s satisfying
the conditions (i) and (ii) from Def.~\ref{def:n-b}. It starts with
a randomly chosen hyperedge whose one vertex is assigned 1 and the
others 0s and continues with connected hyperedges until all
permitted vertices are assigned 1. Multiplicities for found
1s accumulated in the process are taken into account. For contextual
non-binary MMPs that means until a contradiction is reached
(although not necessarily a KS contradiction), i.e., a point at
which no vertex from the remaining hyperedges can be assigned 1;
vertices within these hyperedges are all assigned 0s. The maximal
number of 1s ($HI_{cM}$, $HI^m_{cM}$) is obtained by (up to 50,000)
parallel runs with reshuffled vertices and hyperedges. Because 
we do not make use of backtracking search algorithm resolve
conflicts, the procedure does not exponentially increase the CPU
time with increasing number of vertices. KS sets with several
thousand vertices and hyperedges are processed within seconds
on each CPU of a cluster or a supercomputer.

The probability of not finding correct minimal or maximal $HI_c$
and $HI^m_c$ after so many runs is extremely small but
nevertheless that slight probability restrains our results
meaning that slightly bigger maximums and smaller minimums might
be found in the future computations for a chosen hypergraph.

\begin{definition}{\bf Classical hyperedge number $l_c$}
  is the number of hyperedges which contain vertices that build up
  $HI_{c}$ and the maximal and minimal number of
  such hyperedges are $l_{cM}$ and $l_{cm}$, respectively.
\label{def:lcMm}
\end{definition}

We stress that, in most cases, $l_{cM}$ hyperedges do not contain
$HI_{cM}$ vertices but a smaller number of them. Also, $l_{cm}$
hyperedges usually do not contain $HI_{cm}$ vertices but a
bigger number of them.

The classical vertex index $HI_{cM}$ of a hypergraph $\cal H$ is
related to the independence number of $\cal H$ introduced by
Gr{\"o}tschel, Lov{\'a}sz, and Schrijver (GLS)
{\rm \cite[p.~192]{gro-lovasz-schr-81}}. They introduced the
definition for graphs but it holds for hypergraphs as well,
with graph cliques transliterated into hyperedges. 

\begin{definition}{\bf GLS \boldmath${\alpha}$}. The independence
  number of $\cal H$ denoted by $\alpha({\cal H})$ is the maximum
  number of pairwise non-adjacent vertices.  
\label{def:alpha}
\end{definition}

The independence number $\alpha$ has been given several definitions
and names in the literature. For instance, ``$\alpha({\cal H})$
is the size of the largest set of vertices of $\cal H$ such that
no two elements of the set are adjacent'' \cite{magic-14}.
Such a set is called an {\em independent\/} or a {\em stable\/}
set \cite[Def.~2.13]{amaral-cunha-18},\cite[p.~272,428]{berge-73}
and $\alpha$ is also called a {\em stability number}
\cite[p.~272,428]{berge-73}. In such a set no two vertices are
connected by a hyperedge. Definitions of these notions given by
Voloshin differ, since his sets might include two or more vertices
from the same hyperedge \cite[p.~151]{voloshin-09}. 

\begin{lemma} $HI_{cM}({\cal H})=\alpha({\cal H})$
  \label{lemma:alpha}
\end{lemma}

\begin{proof} Via conditions {\rm(i)} and {\rm (ii)} from
  Def.~\ref{def:n-b} which Def.~\ref{def:ch-i} invokes,
  no two vertices to which one can assign `1' can belong to
  the same hyperedge. The maximum number of such vertices,
  i.e., $H_{cM}({\cal H})$, is therefore the maximum number of
  pairwise non-adjacent vertices, i.e., according to
  Def.~\ref{def:alpha}, just $\alpha({\cal H})$.
\end{proof}  

The reason for distinguishing the two terms $H_{cM}({\cal H})$
and $\alpha({\cal H})$ that are numerically equal is
methodological. Finding $\alpha({\cal H})$ is an NP complete, 
i.e., it is nondeterministic polynomial-time complete
procedure \cite[p.~195]{gro-lovasz-schr-81} applied to the
vertex structure of a hypergraph while our algorithm for
finding $H_{cM}({\cal H})$ relies on repeated (sequential)
non-exhaustive linear searches for 0s and 1s from given lists
so as to satisfy conditions from Def.~\ref{def:ch-i}. 
Hence, while the definition of $H_{cM}({\cal H})$ in
Def.~\ref{def:ch-i} is exact, the algorithm and program
(\textsc{One}) approximate it to an arbitrary precision. Each
run takes 10 ms or less. We obtained $H_{cM}({\cal H})$ for
over 1,000 MMP hypergraphs and verified (via other methods)
that $H_{cM}({\cal H})=\alpha({\cal H})$ for all small MMP
hypergraphs we considered. In this paper we present 43 MMP
hypergraphs for which $H_{cM}({\cal H})=\alpha({\cal H})$
holds and one which might not hold (we were not able to
independently verify whether $H_{cM}(192$-$118)=75$ is the
maximum). In the literature we found only three explicitly
calculated $\alpha({\cal H})$'s: two in
\cite{cabello-severini-winter-14} and one in \cite{magic-14}. 

To arrive at our noncontextuality inequality we introduce the
following definition and lemma.

\begin{definition}
{\bf MMP Quantum Hypergraph Index} $HI_q$ is the sum of weighted
probabilities of all vertices of an $n$-dim $k$-$l$ {\rm MMP}
hypergraph measured in all hyperedges/gates, i.e., repeatedly
whenever they share more than one hyperedge (multiplicity
being greater than 1). 
\label{def:qh-i}
\end{definition}

\begin{lemma}{\bf Vertex-Hyperedge Lemma.}
  For any $n$-dim $k$-$l$ {\rm MMP} hypergraph in which each
  hyperedge contains $n$ vertices the following holds
\begin{eqnarray}
HI_q=\sum^k_{i=1}\frac{m(i)}{n}=l.
\label{eq:theorem}
\end{eqnarray}
 In general, any $n$-dim $k$-$l$ {\rm MMP} hypergraph with
 $\kappa(j)$ considered vertices in $j$-th hyperedge,
 $j=1,\dots,l$,  the following holds
\begin{eqnarray}
  HI_q=\sum^l_{j=1}\sum^{\kappa(j)}_{\lambda=1}p(j,\lambda)=l,
\label{eq:theorem-g}
\end{eqnarray}
where $\kappa(j)$ is the number of vertices in a hyperedge $j$
and $p(j,\lambda)=\frac{1}{\kappa(j)}$ is the probability that
a state of a system corresponding to one of the vertices would
be detected when the hyperedge/gate $j$ is being measured.
\label{th:theorem}
\end{lemma}

\begin{proof} To prove Eq.~(\ref{eq:theorem}) we take a
  constructive approach of building non-isomorphic hypergraphs.
  For any loop of two or more hyperedges Eq.~(\ref{eq:theorem})
  obviously holds. E.g., a loop of 3-dim hyperedges (pentagon)
  contains 10 vertices 5 of which share one hyperedge and the
  other 5 two. Therefore $5\times 1+5\times 2=3\times 5$. When
  we add a hyperedge at two vertex connections the $m$ numbers
  of these vertices rise by one so that the total number of
  vertices increase by $n$ and Eq.~(\ref{eq:theorem}) holds.
  By weaving hyperedges so as to obtain the so-called
  $\delta$-feature \cite{pavicic-pra-17}, i.e., by making pairs
  of them to intersect each other twice (at two vertices) in a
  4-dim space, or up to $n-2$ times in an $n$-dim space
  (Def.~\ref{def:MMP-string}(4.)), the number of vertices
  lowers, but $m$ proportionally rises at the vertices at
  which the hyperedges intersect and Eq.~(\ref{eq:theorem})
  again holds. With this we exhaust constructive steps of
  generating MMP hypergraphs \cite{pmmm05a,pmmm05a-corr}.

  To prove Eq.~(\ref{eq:theorem-g}), we just note that
  $\sum^{\kappa(j)}_{\lambda=1}p(j,\lambda)=1$ for any $j$.
\end{proof}

Eq.~(\ref{eq:theorem}) is equivalent to a generalized Handshake
Lemma for Hypergraphs given as a Solution to Exercise 11.1.3.a
in \cite{melnikov-98}. No proof of the lemma is given in
\cite{melnikov-98}.

\begin{definition} {\bf v-inequality.}
An {\rm MMP} hypergraph vertex inequality or
simply {\rm v-inequality} is defined as
\begin{eqnarray}
HI_{cm}\le HI_{cM}\le HI^m_{cM}<HI_q=l.
\label{eq:i-ineq}
\end{eqnarray}
\label{def:iin}
\end{definition}

\begin{lemma}
  All $n$-dim {\rm NBMMP} hypergraphs satisfy the
  v-inequality, i.e., any v-inequality is a noncontextuality
  inequality (\ref{def:non-c-i-s}).
\label{lemma:v}
\end{lemma}
\begin{proof}
  In an NBMMP hypergraph a maximal number of hyperedges
  that contain `1' must be smaller than the total number of
  hyperedges $l$ and in a BMMP hypergraph every hyperedge
  must allow assignment of one `1', as follows from
  Defs.~\ref{def:n-b} and \ref{def:bin}.
\end{proof}

Measurements of a $k$-$l$ set are carried out on gates, i.e.,
hyperedges---hyperedge by hyperedge---and each hyperedge/gate yields
a single detection (click) corresponding to one of $n$ vertices
(vectors, states) contained in the hyperedges with a probability
of $\frac{1}{n}$. This means that for MMP hypergraphs whose all
hyperedges contain $n$ vertices, we can build the statistics of the
obtained data in two ways:

\begin{statistics}\label{hyp-stat-a} \phantom{i}
\begin{enumerate}
\item  {\em Raw data statistics} for {\rm MMP} hypergraphs with
  all hyperedges containing $n$ vertices, often adopted in the
  literature, e.g., {\rm \cite[Eqs.~(2)]{d-ambrosio-cabello-13},
  \cite[lines under Eq.~(2)]{yu-oh-12}}, etc., consists
  of assigning $\frac{1}{n}$ probability to each of $k$ 
  vertices contained in the hypergraph (see Def.~\ref{def:n-b}),
  independently of whether the vertices appear in just one
  hyperedge or in two or more of them. Such a statistics does not
  appear as a valid processing of measurement data.
\item {\em Postprocessed data statistics} takes into
  account that within an MMP hypergraph
\begin{enumerate}
\item vertex \textquoteleft$v$\textquoteright\ might share $m(v)$
  hyperedges;
\item measurements are performed on $\frac{1}{n}$ vertices
  $v(j)$ contained in hyperedges \textquoteleft$j$\textquoteright,
  hyperedge by hyperedge ($j=1,\dots,l$);
\item outcomes of measurements carried out on particular
  vertices $v(j)$ in particular hyperedges $j$ might be dropped
  out of consideration leaving us with $\kappa(j)$ vertices in
  hyperedges $j$.
\end{enumerate}
  Hence, we collect data from $\kappa(j)\le n$ vertices
  in each hyperedge $j$. The probability of getting measurement
  data for each vertex within the hyperedge, after discarding data
  for $n-\kappa(j)$ dropped vertices, is  $\frac{1}{\kappa(j)}$.
  The sum of all probabilities is, according, to
  Eq.~(\ref{eq:theorem-g}), equal to the size of the hypergraph,
  i.e., to the number of its hyperedges $l$.    
\end{enumerate}
\end{statistics}

As for hyperedges, several additional definitions are due for
a further analysis of the aforementioned structure in the next
sections.

\begin{definition} An {\rm MMP} hypergraph maximum
  hyper\/{\bf e}dge inequality or simply
  {\bf e$_{Max}$-inequality} is defined as
\begin{eqnarray}
  l_{cM}<l.
\label{eq:e-ineq}
\end{eqnarray}
\label{def:ein}
\end{definition}

As for $l_{cm}$, it satisfies the noncontextuality
e$_{min}$-inequality
\begin{definition} An {\rm MMP} hypergraph minimum
  hyper\/{\bf e}dge inequality or simply
  {\bf e$_{min}$-inequality} is defined as
\begin{eqnarray}
  l_{cm}<l.
\label{eq:e-ineq-m}
\end{eqnarray}
\label{def:einm}
\end{definition}

They are the noncontextuality inequalities simply because
$l_{cm}=l_{cM}=l$ for all binary MMP hypergraphs. The inequality
Eq.~(\ref{eq:e-ineq-m}) has a bigger span between the terms than
the inequality (\ref{def:ein}) (because $l_{cm}\le l_{cM}$) and
therefore it is more viable for an implementation. It seems to
us that $l_{cm}$ is the ``rank of contextuality'' Horodecki at
al.~\cite{horod-22} introduced as a quantifier of contextuality
for hypergraphs, although it is rather difficult to establish a
correspondence between their formalism and the MMP hypergraph
language, in particular because they keep using several
different names for vertices and hyperedges throughout their
paper.

Whenever we refer to both e$_{Max}$- and e$_{min}$-inequalities
we invoke them as e-inequalities. 

\begin{lemma}
  All $n$-dim non-binary {\rm MMP} hypergraphs satisfy the
  e-inequalities, i.e., any e-inequality is a noncontextuality
  inequality (\ref{def:non-c-i-s}).
\label{lemma:e}
\end{lemma}

\begin{proof}
  For KS MMP hypergraphs it follows directly from the KS theorem
  \ref{th:KS} since both a maximal and a minimal number of
  hyperedges that contain `1' must be smaller than the total
  number of hyperedges $l$. For non-KS NBMMP hypergraphs it
  follows from Def.~\ref{def:n-b} and its conditions (i) and (ii)
  in the same way. 
    \end{proof}

Here we stress that the raw data statistics cause a problem with
the application of the maximum of total probabilities to obtain
measurement outcomes that served some authors in the literature
to establish noncontextual inequalities which should single out
contextual sets. The maximum in question is derived from the
fractional independence number defined in the graph and
hypergraph theories by the following definition
\cite[p.~192]{gro-lovasz-schr-81}.

\begin{definition}{\bf Fractional independence number
    \boldmath{$\alpha^*({\cal H})$}} of an {\rm MMP} hypergraph
  ${\cal H}(V,E)={\cal H}(k$-$l)$ is the maximum value of
  $\sum_{v=1}^k x(v)$, where $v\in V$ and where $x(v)$ are
  non-negative real numbers such that $\sum_{v\in e} x(v)\le 1$
  for each hyperedge $e\in E$ of $\cal H$.
    \label{def:alpha-star}
\end{definition}

Since $\alpha^*({\cal H})$ is the optimum of a linear
programming (LP) problem, it can also be given the following
equivalent definition \cite{dudek-08}. 

\begin{definition}{\bf LP Fractional independence number
    \boldmath{$\alpha^*({\cal H})$}} of an {\rm MMP} hypergraph
  ${\cal H}(k$-$l)$ is the optimum value of the following linear 
  programming problem $LP=LP({\cal H})$

  \smallskip
  \qquad (LP) Maximize $\sum_{v\in V}x(v)$

\medskip
\qquad\phantom{(LP)\ }subject to $\sum_{v\in e}\le 1$,
$\forall e\in E$

\medskip
\qquad\phantom{(LP)subject to\ $\>$}$x(v)\in[0,1]$,
$\forall v\in V$
\label{def:alpha-star-LP}
\end{definition}

The fractional independence number $\alpha^*$ has recently
been renamed to the {\em fractional packing number} and used
for obtaining noncontextuality inequalities for measured
contextual quantum systems 
\cite{cabello-severini-winter-14,magic-14,amaral-cunha-18}.
However, the properties of probabilities of quantum contextual
measurements in these references have not been fully used in
applications of the fractional independence number to them, as
follows from the following postulate and theorem which dispense
with variable probabilities $x(v)$ used in
Defs.~\ref{def:alpha-star} and \ref{def:alpha-star-LP}.

\begin{resultttt}\label{postulate} A quantum system generated in
  an unknown (unprepared) pure state in an apparatus (e.g., a
  generalized Stern-Gerlach one), when exiting from it through
  one of the out-ports (channels) of its gate, has equal
  probability of being detected {\rm\cite[Sec.~5-1]{feynmanIII}}
  on its exit.
\end{resultttt}

That means that in an $n$-dim $k$-$l$ MMP hypergraph with $n$
vertices within each hyperedge the probability of detecting the
system at one of the ports is $p(v)=\frac{1}{n}$ for
any vertex $v$ within each hyperedge $E_j$, $j=1,\dots,l$ and
that the condition $\sum_{v\in {E_j}} p(v)\le 1$ for each
hyperedge $E_j\in E$, $j=1,\dots,l$ is satisfied.
Later on we might decide to drop particular vertices and
apply postprocessed data statistics \ref{hyp-stat-a}.

That also means that assuming that one can manipulate the
generated pure states before measuring them is not plausible.
For instance, an unprepared spin-1 system can be projected to
one of three subspaces with equal probabilities of 1/3 and this
inherent quantum randomness is what builds up the contextuality of
the whole set. Any filtering of the systems before measurements,
i.e., any other set of probabilities that do not amount to
1/3 each, ruins the contextuality since then the sum of
probabilities is less than 1 per a hyperedge and we loose data.

As an example, take Eq.~(\ref{eq:theorem}) which gives the sum of
probabilities of detecting vertex states over all multiple
appearances of vertices in hyperedges/gates obtained by
postprocessing of measurement data. The sum takes into account
multiple detections of systems corresponding to the same vertices
exiting through different hyperedges/gates the vertices/ports share.
For instance, take the 18-9 MMP
{\tt 1234,4567,789A,ABCD,DEFG,GHI1,\break 35CE,29BI,68FH.}
and carry out measurements on all 9 hyperedges of them. Then, the
probability of detecting a system determined by any of the vertices
in any hyperedge/gate is $\frac{1}{4}$. But every vertex appears
in two gates, so the sum of probabilities of the system being
detected in each such pair is $\frac{1}{2}$ and the overall sum
of probabilities is $18\frac{1}{2}=9$, i.e., $HI_q=l=9$.
This is a consequence of measuring each of the 9 hyperedges
separately and obtaining outcomes for each tetrads of vertices
per hyperedge with the probability of $\frac{1}{4}$---making a
total of 1 per hyperedge. Hence, the sum of probabilities
for all 9 hyperedges is equal to 9 in contrast to the collection
of the raw data in \cite[Eqs.~(2)]{d-ambrosio-cabello-13} where
the sum $18\frac{1}{4}=4.5$ is assumed what would mean the sum of
probabilities per a hyperedge of $\frac{1}{2}$ and a dismissal of
half of the data. 

\begin{Xeorem}Let variables $x(v)$ from Defs.~\ref{def:alpha-star}
  and \ref{def:alpha-star-LP} be the probabilities $p(v)$,
  $v\in V$ of detecting an event by {\rm YES-NO}
  measurements at one of the out-ports (vertices) contained within
  a hyperedge of an $n$-dim {\rm MMP} hypergraph
  ${\cal H}(V,E)={\cal H}(k$-$l)$. Each of $E_j\in E$,
  $j=1,\dots,l$ hyperedges (gates) contains $n$ vertices. The
  probabilities satisfy the condition:
  \begin{eqnarray} \sum_{v\in {E_j}}p(v)\le 1, j=1,\dots,l.
    \label{eq:sum-edge}
  \end{eqnarray}
  They also satisfy the following:

  (a) Under the {\rm raw data statistics \ref{hyp-stat-a}(a)}
  assumption, i.e., under the assumption that every vertex
  within an {\rm MMP} hypergraph has $\frac{1}{n}$ probability of
  being detected {\rm \cite{cabello-severini-winter-14},
    \cite{magic-14}, \cite{amaral-cunha-18}},
  the sum of all probabilities is:
\begin{eqnarray} \sum_{v=1}^kp(v)=\frac{k}{n}
   =\alpha_r^*(k{\text{-}}l)=\alpha_r^*({\cal H})
  \label{eq:alpha-star}
\end{eqnarray}
where $\alpha_r^*$ is the {\rm raw quantum fractional independence
  number}.

This implies that, in general, the
{\boldmath{$\alpha^*_r$}}{\bf-inequality}
(compare it with free probability GLS inequality
{\rm \cite[p.~192]{gro-lovasz-schr-81}}) 
\begin{eqnarray} 
 HI_{cM}=\alpha({\cal H})\le \alpha_r^*({\cal H})=\frac{k}{n}
    \label{eq:alpha-alpha}
\end{eqnarray}
  does not always hold for quantum mechanical measurements whose
  probabilities of detection within each hyperedge satisfy
  the condition given by Eq.~(\ref{eq:sum-edge}), i.e.,
  $p(v)=\frac{1}{n}$, $v\in V$. The inequality is violated by a
  significant portion of contextual non-binary MMP hypergraphs
  in any dimension; see Fig.~\ref{fig:alpha-star}(a-d), the 3rd
  figure in Sec.~\ref{subsec:3d}, figures (a,h) in
  Sec.~\ref{subsec:4dm}, 1st figure (d) in Sec.~\ref{subsec:3d},
  the 1st figure (b) in Sec.~\ref{subsec:4d}, the 3rd figure (g)
  in Sec.~\ref{subsec:magic}, the 1st figure (a,h) in Appendix
  \ref{app4}, and the 1st table in Sec.~\ref{subsec:3d}, the 1st
  table in Sec.~\ref{subsec:5d}, and the 2nd table in
  Sec.~\ref{subsec:6d}. It is, therefore, not a reliable
  discriminator of contextual sets. This is only to be expected
  since the raw data statistics is, as we pointed out above,
  not a consistent elaboration of measurement data. 

\begin{figure}[ht]
  \begin{center}
  \includegraphics[width=1\textwidth]{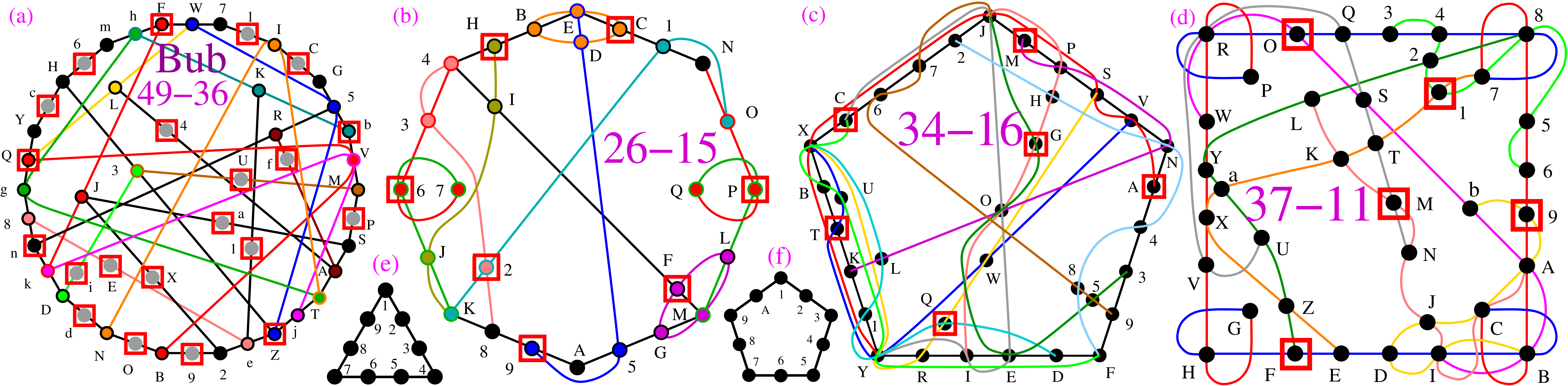}
\end{center}
\caption{3-,4-, 6-, and 8-dim KS MMP hypergraphs that violate the
  inequality $\alpha < \alpha_r^*$ given by
  Eq.~(\ref{eq:alpha-alpha}); the vertices that belong to the
  independent (stable) set and contribute to $\alpha=HI_{cM}$ are
  squared in red; 
  (a) 3-dim; $\alpha=21>\alpha_r^*=\frac{49}{3}\approx 16.\dot{3}$;
  (b) 4-dim; $\alpha=7>\alpha_r^*=\frac{26}{4}=6.5$;
  (c) 6-dim $\alpha=6>\alpha_r^*=\frac{34}{6}\approx 5.\dot{6}$;
  (d) 8-dim; $\alpha=5>\alpha_r^*=\frac{37}{8}=4.625$;
  (e) 4-dim; $\alpha=3>\alpha_r^*=\frac{9}{4}=2.25$;
  (f) 3-dim; $\alpha=5>\alpha_r^*=\frac{10}{3}=3.\dot{3}$;
  (a-d) MMP hypergraphs are KS non-binary critical contextual
  sets while (e) and (f) are 9-3 4-dim and 10-5 3-dim binary
  noncontextual MMP hypergraphs, respectively.}
\label{fig:alpha-star}
\end{figure}

(b) Under the {\rm postprocessed data statistics}
  \ref{hyp-stat-a}(b) assumption, i.e., under the assumption
  that every vertex $v$ within an {\rm MMP} hypergraph has 
  $\frac{m(v)}{n}$ probability of being detected, the sum of
  all probabilities is:
  \begin{eqnarray} \sum_{v=1}^kp(v)
    =\sum_{v=1}^k\frac{m(v)}{n}
    =l=\alpha_p^*(k{\text{-}}l)=\alpha_p^*({\cal H})
  \label{eq:alpha-star-b}
\end{eqnarray}
where $\alpha_p^*$ is called the {\rm postprocessed quantum
  fractional independence number}.

This implies that the {\boldmath{$\alpha^*_p$}}{\bf-inequality}
  {\rm \cite[p.~192]{gro-lovasz-schr-81}}, 
\begin{eqnarray} 
  HI_{cM}=\alpha({\cal H})< \alpha_p^*({\cal H})=l=HI_q,
    \label{eq:alpha-alpha-b}
\end{eqnarray}
which follows from the Vertex-Hyperedge Lemma \ref{eq:theorem-g},
is another form of the v-inequality (\ref{eq:i-ineq}) and is
therefore a noncontextuality inequality and a
reliable discriminator of contextual sets. 
\label{th:alpha-star}
\end{Xeorem}

\begin{proof}
  Quantum YES-NO measurements of states determined by MMP
  hypergraphs are carried out either by letting the quantum
  system through gates, e.g., Stern-Gerlach devices or via
  projecting their states on unit vectors. According to the
  quantum indeterminacy postulate \ref{postulate} that
  makes the probabilities of their detection constant.

  (a) Within the {\em raw data statistics} \ref{hyp-stat-a}(a)
  one assumes, according to the Postulate \ref{postulate}, that the
  probability of detecting a state that corresponds to a vertex
  $v\in V$ is equal to the probability of detecting that state
  within any of hyperedges the vertex might belong to, i.e.,
  $p(v)=\frac{1}{n}$ for any $v\in V$. That yields
  Eq.~(\ref{eq:alpha-star}). Examples of such an approach in the
  literature are: $\alpha_r^*(5,5)=\frac{5}{2}$ for the induced
  4-dim pentagon ($5\times\frac{1}{2}$)
  \cite[p.3, top]{cabello-severini-winter-14} and 
  $\alpha_r^*(18,9)=4.5$ for the 4-dim 18-9 MMP 
  ($18\times\frac{1}{4}$) \cite[Eq.~(2)]{d-ambrosio-cabello-13}.
  These examples do satisfy the inequality (\ref{eq:alpha-alpha}).
  The others that do not are given in the Theorem
  \ref{th:alpha-star}(a). 

  (b) Within the {\em postprocessed data statistics}
  \ref{hyp-stat-a}(b) every vertex $v\in V$ is taken into account
  $m(v)$ times, yielding the probability $p(v)=\frac{m(v)}{n}$
  (Cf.~Eq.~(\ref{eq:theorem})). This gives
  Eq.~(\ref{eq:alpha-star-b}) and the inequality 
  (\ref{eq:alpha-alpha-b}). Examples of such an approach are given
  for a pentagon in Sec.~\ref{sec:hyp-op}, Eq.~(\ref{eq:in-25}):
  $\alpha_p^*(5,5)=HI_q[5,5]=5$ and for the 18-9 MMP hypergraph
  in Sec.~\ref{sec:op-ineqal} below Eq.~(\ref{eq:a-cab2}) and in
  Sec.~\ref{sec:structure} below the Postulate \ref{postulate}:
  $\alpha_p^*(18,9)=HI_q[18,9]=9$.
\end{proof}

Notice that since the theorem asserts that a contextual non-binary
MMP hypergraph might or might not satisfy the raw quantum
fractional independence number inequality given by
Eq.~(\ref{eq:alpha-alpha}) and which is therefore not a
noncontextuality inequality, the only known unequivocal
noncontextuality inequalities that hold for every MMP hypergraph
are v- and e-inequalities (and hence also $\alpha_p$ inequalities).
Still, for a contextual $k$-$l$ MMP hypergraph the
$\alpha^*_r$-inequality has had a greater span (smaller $\alpha$)
than for a noncontextual $k$-$l$ MMP hypergraph for roughly 1,000
randomly chosen $k$-$l$ MMP hypergraphs. 

\smallskip
If v- and e-inequalities were satisfied, an MMP hypergraph would
be contextual. If not, it wouldn't. So, the v- and e-inequalities are
noncontextuality inequalities. On the other hand, as we stressed
above, $\alpha_r^*$-inequalities, are not such
direct measures of the quantum contextuality since many contextual
MMP hypergraphs do not satisfy them. All contextual MMP hypergraphs
satisfy the inequality $l_{cm}\le l_{cM}<l$, while the 
noncontextual MMP hypergraphs satisfy $l_{cm}=l_{cM}=l$. That means
that a noncontextual MMP hypergraph is structurally  different from
a contextual MMP hypergraph. 

Note that both $\alpha_r^*$- and $\alpha_p^*$-inequalities
given by Eqs.~(\ref{eq:alpha-alpha}) and (\ref{eq:alpha-alpha-b}),
respectively, assume the validity of the quantum indeterminacy
postulate \ref{postulate}. Notwithstanding the plausibility of
the postulate, some authors apply the original GLS inequality
\cite[Result 1]{cabello-severini-winter-14} 
\begin{eqnarray}\label{eq:alpha-cabelllo}
\alpha({\cal H})\le\alpha^*({\cal H}),   
\end{eqnarray}
(where $\alpha({\cal H})$ is defined by Def.~\ref{def:alpha}
and $\alpha^*({\cal H})$ by Defs.~\ref{def:alpha-star} and
\ref{def:alpha-star-LP}), to contextual non-binary MMP hypergraphs
and claim \cite[Results 1 \&\ 2]{cabello-severini-winter-14}
that the inequality (\ref{eq:alpha-cabelllo}) is a noncontextuality
inequality. In \cite{cabello-severini-winter-14} there is also the
weight of the probabilities at each hyperedge which, according to
the indeterminacy postulate \ref{postulate}, should be equal to 1. 

The discrepancy comes from the fact that the inequality 
Eq.~(\ref{eq:alpha-cabelllo}) is correct provided $p$ is not a
constant (as it would be under the assumption of the quantum
indeterminacy postulate) but a free variable which is determined
as a solution of the linear programming problem given in
Def.~\ref{def:alpha-star-LP}. In \cite{cabello-severini-winter-14}
it is even stated that finding $\alpha^*$ is NP hard, what is
correct for the GLS inequality.

For a pentagon, the raw data statistics and LP approaches give
the same result $\alpha^*=\frac{5}{2}$. 

A difference emerges already for a very simple MMP hypergraph
9-3 given in Fig.~\ref{fig:alpha-star}(e), though. For a free
$p$ we have:

LP[\{-1,-1,-1,-1,-1,-1,-1,-1,-1\},\{\{1,1,1,1,0,0,0,0,0\},\{0,0,0,1,1,1,1,0,0\},\{1,0,0,0,0,0,1,1,1\}\},\break \{\{1,-1\},\{1,-1\},\{1,-1\}\}]

Out:=\{0,1,0,0,1,0,0,1,0\}, i.e, $\alpha^*=3$. Since $\alpha=3$,
inequality (\ref{eq:alpha-cabelllo}) is satisfied.

However, for $p=\frac{1}{4}$ we get

LP[\{-1,-1,-1,-1,-1,-1,-1,-1,-1\},\{\{1,1,1,1,0,0,0,0,0\},\{0,0,0,1,1,1,1,0,0\},\{1,0,0,0,0,0,1,1,1\}\},\break \{\{1,-1\},\{1,-1\},\{1,-1\}\},\{$\!$\{$\frac{1}{4}$,1\},\{$\frac{1}{4}$,1\},\{$\frac{1}{4}$,1\},\{$\frac{1}{4}$,1\},\{$\frac{1}{4}$,1\},\{$\frac{1}{4}$,1\},\{$\frac{1}{4}$,1\},\{$\frac{1}{4}$,1\},\{$\frac{1}{4}$,1\}$\!$\}$\!$]

Out:=\{$\frac{1}{4}$,$\frac{1}{4}$,$\frac{1}{4}$,$\frac{1}{4}$,$\frac{1}{4}$,$\frac{1}{4}$,$\frac{1}{4}$,$\frac{1}{4}$,$\frac{1}{4}$\},
i.e., $\alpha^*=\frac{9}{4}=2.25$, which violates
inequality (\ref{eq:alpha-cabelllo}) as well as (\ref{eq:alpha-alpha})
($\alpha_r^*=\frac{9}{4}$).

Hence, $\alpha^*\ne\alpha_r^*$, meaning that $\alpha_r^*$ is a
special case of $\alpha^*$; the former $\alpha$ applies to variable
probabilities and the latter to fixed probabilities of YES-NO
quantum measurements implying that Eq.~(\ref{eq:alpha-cabelllo})
fails for arbitrary many quantum measurements and that the
probabilities must be equal and constant at all ports of a quantum
gate as a consequence of quantum indeterminacy postulate
\ref{postulate}, i.e., of a genuine quantum randomness.

\begin{discussion}\label{th:discussion} {\bf Non-maximal number
    of vertices within hyperedges and their probabilities.}
There are {\rm MMP} hypergraphs for which we should yet decide what
an optimal approach to form a proper statistics of their measurement
should be and those are the $n$-dim {\rm MMP} hypergraphs
whose hyperedges do not all have a maximal number of vertices,
i.e., $n$ vertices. Consider, for example, the 9-3 {\rm MMP} shown
in Fig.~\ref{fig:alpha-star}(e) from which the vertices {\tt 8} and
{\tt 9} are dropped from the consideration. We are left with 7-3
{\rm MMP}: {\tt 1234,4567,71.} When we detect particles at outgoing
ports of {\tt 1234,4567} hyperedges/gates the probability of
their detection is $\frac{1}{4}$. The same is with the {\tt 7891}
in the 9-3 {\rm MMP}, but in the 7-3 {\rm MMP} we discard two
outcomes---those at {\tt 8} and those {\tt 9}. After the dismissal
of {\tt 8} and {\tt 9} data, the {\tt 7}- and {\tt 1}-detection
have the probability of $\frac{1}{2}$ each. But the question
remains about the overall probability of detections at {\tt 7}
and {\tt 1} within the 7-3 {\rm MMP} hypergraph. In
{\rm\cite{pavicic-entropy-19}} we proposed that the probability
be the arithmetic mean of the probabilities the vertex has in all
hyperedges it shares. For instance, vertex {\tt 7} within
the {\tt 4567} hyperedge would have the probability $\frac{1}{4}$
and within the {\tt 71} hyperedge it would have the probability of
$\frac{1}{2}$. The overall probability for {\tt 7} to occur,
i.e., the arithmetic mean of these probabilities, would be
$(\frac{1}{4}+\frac{1}{2})/2=\frac{3}{8}$. (Notice also that it
would be plausible to assume twice as many measurements for the
{\tt 7891} than for the other two hyperedges, if we wanted to
drop data for {\tt 8} and {\tt 9} and claim the probability
$\frac{1}{2}$ for {\tt 1} and {\tt 7}.)
We provide some further examples and a discussion in
Secs.~{\rm\ref{subsec:3d}} and {\rm\ref{subsec:magic}}. 
\end{discussion}

\subsection{\label{sec:structure-e}Hypergraph structure
   exemplified}

To get a better insight into the introduced notions and features,
we consider two examples: a complex 4-dim 21-11 KS MMP hypergraph,
shown in Fig.~\ref{fig:21c}, that is not a subset of Peres' 24-24
KS set \cite{peres} unlike the real 4-dim KS 18-9 that is
\cite{cabello-08} and Yu-Oh's 3-dim 13-16 non-binary non-KS
MMP hypergraph, shown in Fig.~\ref{fig:yu-oh-mmp}, that is a
$\overline{subhypergraph}$ of a binary 25-16 which is
itself a subhypergraph of Peres' 57-40 KS MMP hypergraph. 

\begin{figure}[h]
\begin{center}
  \includegraphics[width=0.98\textwidth]{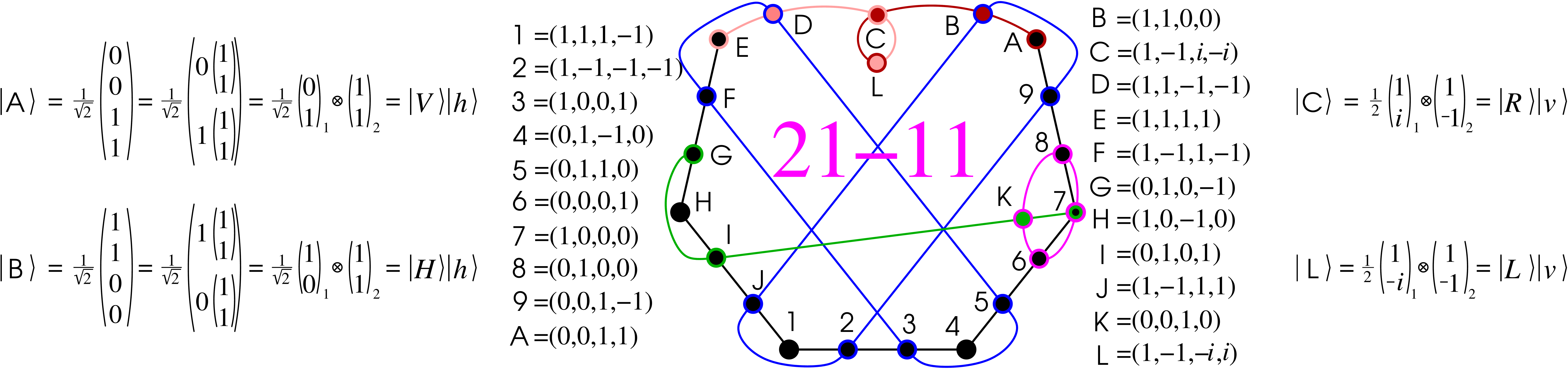}
\end{center}
\caption{21-11 KS set from the 60-105 KS class
  \cite[Fig.5]{pavicic-pra-17} with a coordinatization and
  2-qubit states (polarization + OAM) on single photons at the
  hyperedge {\tt ABCL}. See text. Notice the orthogonality: e.g.,
  {\tt C}$\cdot{\overline{\tt L}}=(1,-1,i,-i)\cdot\overline{(1,-1,-i,i)}=0$}
\label{fig:21c}
\end{figure}

We establish a relation between the hypergraph-based features
introduced in the previous section and the operator-based features
introduced in Sec.~\ref{sec:op-ineqal}, in particular with respect
to inequalities (\ref{eq:a-cab2}) and (\ref{eq:yu-oh}).

We can implement the 21-11 set by means of two qubits mounted
on single photons via spin and angular momentum
\cite{barreiro-kwiat-08,barreiro-wei-kwiat-10} states defined as
follows
\begin{eqnarray}
  |H,V\rangle\!=\!
\begin{pmatrix}
      1,0\\
      0,1\\
\end{pmatrix}_{\!1}\!\!, 
|D,A\rangle\!=\!\frac{1}{\sqrt{2}}\!\left(
     \begin{matrix}
      \pm 1\\
      1\\
    \end{matrix}
  \right)_{\!1}\!\!, 
  |R,L\rangle\!=\!\frac{1}{\sqrt{2}}\!\left(
     \begin{matrix}
      1\\
     \pm i\\
    \end{matrix}
  \right)_{\!1}\!\!, 
|\pm2\rangle\!=\!\left(
     \begin{matrix}
      1,0\\
      0,1\\
    \end{matrix}
  \right)_{2\!}\!\!, 
  |h,v\rangle\!=\!\frac{1}{\sqrt{2}}\!\left(
     \begin{matrix}
      1\\
     \pm 1\\
    \end{matrix}
  \right)_{\!2}\!\!,
\nonumber
\end{eqnarray}
where $H,V$ are horizontal, vertical, $D,A$ diagonal,
anti-diagonal, and $R,L$ right, left circular polarizations,
while $\pm 2$ are Laguerre-Gauss modes carrying
$\pm 2\hbar$ units of orbital angular momentum (OAM) and
$h,v$ are their $\pm$ superposition, respectively. Indices `1'
and `2' refer to the 1st and 2nd qubit mounted on the system,
respectively. Four states building the hyperedge $\tt ABCL$
are given in Fig.~\ref{fig:21c}. Other states have similar
expressions and they enable us to obtain the analogues
of Cabello's states defined by Eq.~(\ref{eq:a-cab}). Since our
vectors are complex, our bras are hermitian conjugates of our
kets: ${\cal O}_{\tt i}=2|{\tt i}\rangle\langle {\tt i}|^\dag-I$.
The matrix forms of the operators of our four states read:
\begin{eqnarray}
  {\cal O}_{\tt A}=
  \begin{pmatrix}
    \bar{1} & 0 & 0 & 0 \\
     0 &\bar{1} & 0 & 0\\
     0 & 0 & 0 & 1 \\
     0 & 0 & 1 & 0
  \end{pmatrix}\!,\, 
  {\cal O}_{\tt B}=
  \begin{pmatrix}
    0 & 1 & 0 & 0\\
    1 & 0 & 0 & 0\\
    0 &0 & \bar{1} & 0\\
    0 &0 & 0 & \bar{1}
  \end{pmatrix}\!,\,
  {\cal O}_{\tt C}=\frac{1}{2}\!
  \begin{pmatrix}
    \bar{1} & \bar{1} & \bar{i} & i\\
    \bar{1} & \bar{1} & i & \bar{i}\\
    i & \bar{i} & \bar{1} & \bar{1}\\
    \bar{i} & i & \bar{1} & \bar{1}
    \end{pmatrix}\!,\,
  {\cal O}_{\tt L}=\frac{1}{2}\!
  \begin{pmatrix}
    \bar{1} & \bar{1} & i & \bar{i}\\
    \bar{1} & \bar{1} & \bar{i} & i\\
    \bar{i} &  i & \bar{1} & \bar{1}\\
     i & \bar{i} & \bar{1} & \bar{1}\\
   \end{pmatrix}\!,
\nonumber
\end{eqnarray}
where $\bar{1}$ stands for $-1$ and $\bar{i}$ for $-i$. 

We can verify that any of
$|{\tt A}\rangle$,$|{\tt B}\rangle$,$|{\tt C}\rangle$,$|{\tt L}\rangle$
is an eigenvector of any of ${\cal O}_{\tt A,B,C,L}$ with eigenvalues
$\pm1$, and that
${\cal O}_{\tt A}{\cal O}_{\tt B}{\cal O}_{\tt C}{\cal O}_{\tt L}=-I$
holds. We can also verify that these relations hold for any
hyperedge. Actually, we conjecture that they hold for any hyperedge
of any critical KS MMP hypergraph in any dimension. That yields:
\begin{eqnarray}
  P_q[k,l]=\mp\sum^l_{e=1}\langle{\cal O}[e]\rangle=l;
 \qquad {\rm for\ }{\rm our\ Fig.~\ref{fig:21c}\ set}:\
  P_q[21,11]=-\sum^{11}_{e=1}\langle{\cal O}[e]\rangle=11,\qquad
\label{eq:O-I}
\end{eqnarray}
where $\mp$ signs are for even/odd dimensions, respectively,
and where ${\cal O}[e]$ stands for
${\cal O}_{1e}{\cal O}_{2e}$ $\cdots{\cal O}_{ne}$, where $je$ refers
to the $j$-th vertex on the hyperedge $e$. With respect
to the aforementioned eigenvalues we assume that classical
counterparts $O_{je}$ of quantum ${\cal O}_{je}$ have two possible
results $O_{je}=1$ and $O_{je}=-1$. Maximal values of
the classical analogues of Eq.~(\ref{eq:O-I}) is
given by  Eq.~(\ref{eq:O-Ic}).
\begin{eqnarray}
  P_c[k,l]=\mp\sum^l_{e=1}O[e]=l-2; \qquad
  {\rm for\ }{\rm our\ set:}\qquad
P_c[21,11]=-\sum^{11}_{e=1}O[e]=9,\qquad
\label{eq:O-Ic}
\end{eqnarray}
where $O[e]$ stands for $O_{1e}O_{2e}\cdots O_{ne}$. We confirmed
the special case 21-11 result by Mathematica.

Equations (\ref{eq:O-I}) and
(\ref{eq:O-Ic}) yield the noncontextuality
inequalities
\begin{eqnarray}
P_c[k,l]\le P_q[k,l]; \qquad
  {\rm for\ our\ }{\rm set:}\qquad
P_c[21,11]=9 < P_q[21,11]=11. \qquad
\label{eq:O-Iin}
\end{eqnarray}
  
These results correspond to Cabello's \cite[Eqs.~(1,2)]{cabello-08}
($P_c[18,9]=7 < P_q[18,9]=9$) referred to by Eqs.~(\ref{eq:a-cab})
and (\ref{eq:a-cab2}) above.

Now, let us establish the correspondence of these operator-based
results with our hyper\-graph-based approach.
$P_q[k,l]=l$ given by Eq.~(\ref{eq:O-I}) in the operator-based
approach is equivalent to $HI_q[k,l]=l$ given by
Eq.~(\ref{eq:theorem}) of the hypergraph-based approach. In
accordance with this, Cabello \cite[p.~2, top]{cabello-08}
obtains $P_c[18,9]=7$, $P_q[18,9]=9$, and the noncontextuality
inequality $7<9$. In other words, in \cite{cabello-08} he adopts
the postprocessed data statistics while 
in \cite[Eqs.~(2)]{d-ambrosio-cabello-13} the authors adopt
the raw data statistics and have $P_q[18,9]=4.5$ (with 
operators $|i><i|$, not $2|i><i|-I$, but the result should be
the same). In the former approach each vertex state shares two
hyperedges (has multiplicity $m=2$) and is therefore measured twice,
once within measurements carried out on the first hyperedge and the
second time within those carried out on the second hyperedge. Since
all vertices in the 18-9 MMP share two hyperedges one is tempted to
apply the raw data approach, but in the 21-11 MMP the vertex {\tt 7}
shares four hyperedges, i.e., its multiplicity is $m({\tt 7})=4$
and we should take into account that it is measured four times while
all the other vertices are measured only twice when we measure all
hyperedges in turn.

A correlated approach is given by Badzi{\'a}g, Bengtsson, Cabello,
and Pitowsky who obtain $\beta_{CM}(n,l)\le l(n-2)-2$ and
$\beta_{QM}(n,l)=l(n-2)$ \cite[Eqs.~(1,8)]{badz-cabel-09}
corresponding to our $2l_{cM}$ and $2l$ for $n=4$, respectively,
due to the way they define the operators
\cite[Eqs.~(3)]{badz-cabel-09}. It is, therefore, rather surprising
that they get puzzling results for simple cases. For instance,
they claim (in 2009) that the Peres' 24-24 MMP hypergraph
``generate[s] 96 (critical) 20-observable [20 vertices] and 16
(critical) 18-observable [18 vertices] proofs of the KS theorem,''
while it was proved (in 2005) that it contains only two
non-isomorphic critical MMP hypergraphs with 20 vertices (20-11)
and a single critical with 18 vertices (18-9)
\cite[Fig.~3, Figs.~4(b,c)]{pmmm05a,pmmm05a-corr}. Do they
refer to isomorphic instances of these MMP hypergraphs?
Because it was proved in \cite[Table 1]{pmm-2-09} that Peres'
24-24 MMP hypergraph contains only one MMP hypergraph with  
18 vertices and 7 (non-isomorphic) ones with 20 vertices
(including two 20-11 criticals). 

Yu-Oh's operator approach is different \cite{yu-oh-12}. They make
use of the inequality given by Eq.~(\ref{eq:yu-oh}) to prove the
operator contextuality, but the underlying MMP hypergraph is
itself contextual. See Fig.~\ref{fig:yu-oh-mmp}. 

\begin{figure}[h]
\begin{center}
  \includegraphics[width=0.97\textwidth]{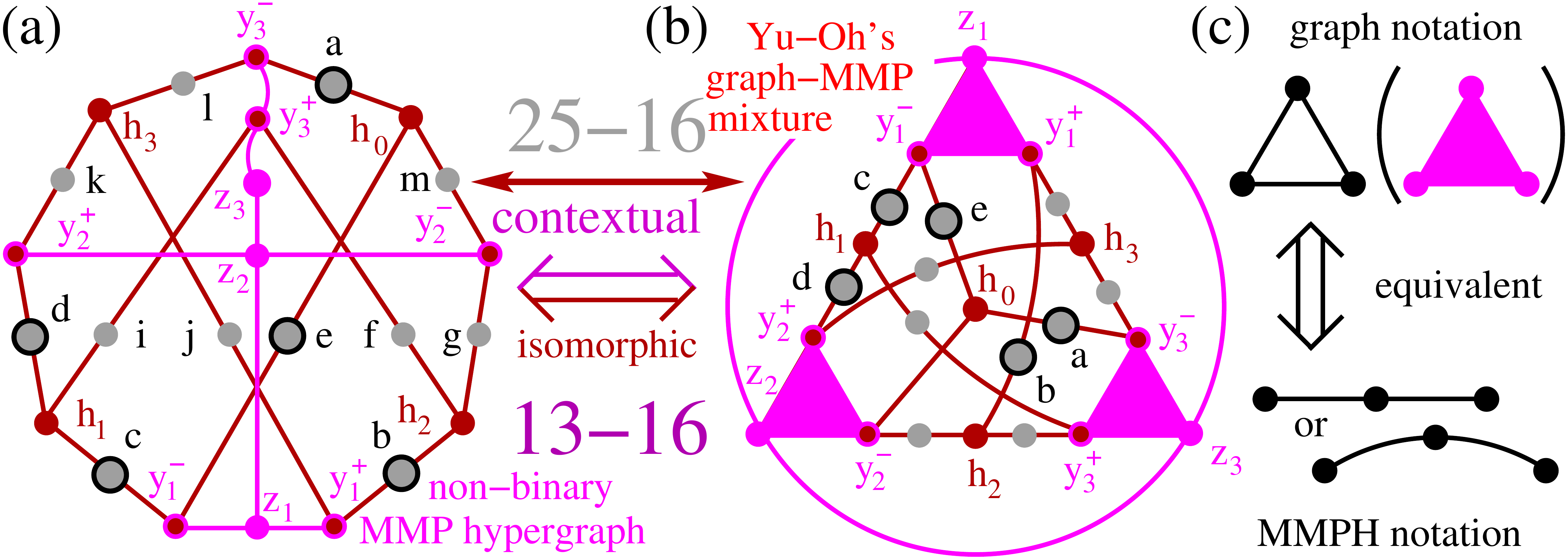}
\end{center}
\caption{(a) MMP hypergraph representation of Yu-Oh's
  graph-hypergraph mixture; vertices with $m=1$ are shown as gray
  dots; all vertices together build 25-16 binary MMP hypergraph;
  MMP hypergraph with the gray vertices dropped build a non-binary
  13-16 MMP hypergraph; (b) Yu-Oh's graph-MMP-hypergraph graphical
  presentation of their set; (c) graph clique vs.~MMP hyperedge.}
\label{fig:yu-oh-mmp}
\end{figure}

More specifically, they build their operators $A_i$ in
Eq.~(\ref{eq:yu-oh}) by means of vectors/states assigned to
vertices of their 13-16 MMP hypergraph; e.g.,
$A_1=A_{z_1}=I-2|z_1\rangle\langle z_1|$, where
$|z_1\rangle=(1,0,0)$ \cite[Eq.~(1) and Appendix]{yu-oh-12}.
All 13 vectors are eigenvectors of $A_i$, $i=1,\dots,13$.
Therefore ``the outcomes for observables $A_i$ are either $+1$
or $-1$, depending on whether there is a photon click (or no
click) in the corresponding photon
detector'' \cite[Supp.~Material]{zu-wang-duan-12}. The operators
violate inequality (\ref{eq:yu-oh}) for any state, i.e., the
violation is state independent.

Before we proceed with a further analysis of Yu-Oh's set we would
like to point out that there is the following problem with the
violation of inequality (\ref{eq:yu-oh}). In
\cite{pavicic-entropy-19} we tested it on 50 different non-binary
MMP hypergraphs and found no violation. That means that the
inequality is unsuitable for application on an arbitrary MMP
hypergraph. 

On the other hand, Yu-Oh's MMP hypergraph does satisfy the
v- ($HI_{cM}=\alpha=5<HI^m_{cM}=14<l=16$), e- ($l_{cM}=14<l=16$),
and even $\alpha_r^*$-inequality ($5=\alpha<\alpha_r^*=5.\dot{6}$).
(Note that $l_{cM}\ne l-1$ because 13-16 is not critical;
see the 3rd figure from Sec.~\ref{subsec:3d}.)
Thus, in addition to its operator implementation, Yu-Oh's set,
as any non-KS non-binary MMP hypergraph, can be straightforwardly
and instantaneously identified as such via our programs and
implemented with the help of YES-NO measurements of vertex states
exiting the hyperedge gates. 

Yu and Oh arrived at their 13-16 non-binary MMP hypergraph by
removing $m=1$ vertices from the 25-16 binary MMP hypergraph
(shown as gray dots in Fig.~\ref{fig:yu-oh-mmp}(a)) which is
itself a subhypergraph of Peres' 57-40 MMP hypergraph (figure (c)
in Appendix \ref{app3d}). Actually, just five gray vertices
({\tt a,\dots,e}) suffice to turn the 13-16 into a binary 18-16
MMP hypergraph. Its parameters are: $HI^m_{cM}=l_{cM}=l=16$ and
v- and e-inequalities are violated. If we then remove any of
the five gray vertices ({\tt a,\dots,e}), the MMP hypergraph
becomes a contextual non-binary 17-16 one and
$HI^m_{cM}=l_{cM}=14<l=16$, i.e., v- and e-inequalities are
satisfied. It should be stressed here that in a 3-dim space
all vector triples should be implementable, i.e., that 25-16
should have a coordinatization. The $\{0,\pm1\}$ components
suffice for only 13 vectors of the ``magic cube'' shown in
\cite[Fig.~1]{yu-oh-12}. For a proper implementation of the
13-16 one should make use of, e.g., $\{0,\pm1,2\}$ components
to allow an implementation of the 25-16 as well. In doing so
one also changes the assignments of vectors to the original
13 vertices. This is analogous to the pentagon (e)-resolution
of the (c)-attempt in Fig.~\ref{fig:pentagon}. Vectors
{\tt a,\dots,m} in Fig.~\ref{fig:yu-oh-mmp}(a) contain the
component `2', while the vectors assigned to the original 13
vertices do not.

Taken together, operator-based measurements of contextual states
differ from hypergraph-based ones in the following way. To measure
the mean values of observables/operators we have to first measure
correlations between observables/operators defined by
vertices/vectors of an MMP hypergraph. To prove the state
independence we have to carry out measurements with different
input states. Thus, the number of measurements grows
exponentially with the size and with the dimension of the set.

In the hypergraph-based approach the input states are the states
from the coordinatization of an MMP, as in Fig.~\ref{fig:21c}, and
we verify them by detecting output states at the ports of each
gate/hyperedge. The number of measurements grows linearly with the
size of MMP hypergraphs and with their dimension.

Of course, each approach has its own application. When an MMP
hypergraph is a part of a quantum network which requires projections
to specified states, then we use the operator-based approach, and
when it is a part of a quantum computation which has to
distinguish contextual loops from noncontextual ones, then we
use a hypergraph approach.

\bigskip\bigskip

\section{\label{sec:exampl}Analysis of MMP
  hypergraph features in diversified dimensions}

\medskip
\subsection{\label{subsec:4dm}MMP hypergraph
  multiplicity}

So far we have seen that the multiplicity of vertices plays
significant roles in determining the features of MMP hypergraphs.
Here we consider two such features shown in Fig.~\ref{fig:for17}
and Table \ref{T:masters} (odd number of hyperedges) and in the
Appendix \ref{app4} (even number of hyperedges).

\begin{figure}[H]
\includegraphics[width=1\textwidth]{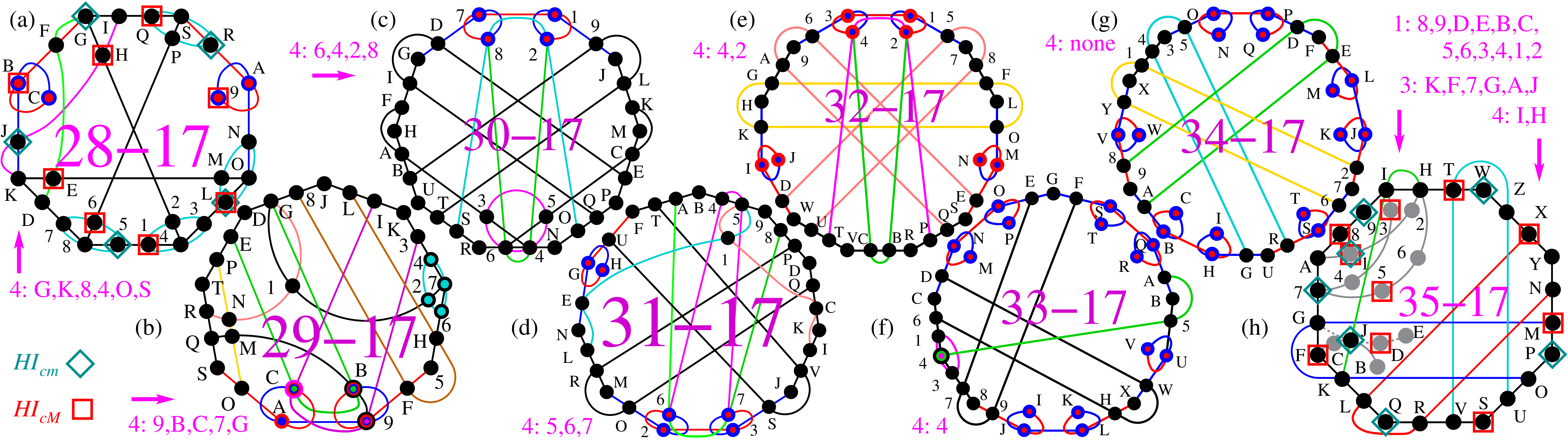}
\caption{4-dim KS criticals with 17 hyperedges from the 156-249
  class. $m\ne 2$ are stated for each set. (a)-(g) have only $m=2$
  and $m=4$. Distributions (for (a) and (h)) of the maximal and
  minimal numbers of ``classical 1s'' are given by
  squares and diamonds, respectively; (a)-(g) have parity proofs;
  (a) $\alpha=8>\alpha_r^*=7$; (h) $\alpha=10>\alpha_r^*=8.75$.}
\label{fig:for17}
\end{figure}

    \renewcommand{\arraystretch}{1.4}
\begin{table}[!htbp]
\caption{Multiplicities $m$ of master KS sets. The 3-dim 81-52 KS
  master is vector-generated from vector components
  $\{0,\pm 1,\pm\sqrt 2,3\}$ which build vectors of Peres' 57-40 sets
  \cite{peres}. Master 81-52 has only one critical set---Peres'
  57-40.}
  \centering
\setlength{\tabcolsep}{4pt}
\begin{tabular}{|c|c|cccc|ccc|} 
  \hline
$n$&3-dim&\multicolumn{4}{c}{4-dim}&\multicolumn{3}{|c|}{6-dim}
  \\
  \hline
\multirow{2}{*}{master}\ \ &\multirow{2}{*}{81-52}&24-24&60-75&60-105&148-265&81-162&216-153&834-1609\\
& &\cite{peres,pmm-2-10}&\cite{waeg-aravind-megill-pavicic-11}&\cite{waeg-aravind-jpa-11,pavicic-pra-17,pwma-19}&\cite{pavicic-pra-17,pwma-19}&\cite{pwma-19}&\cite{pm-entropy18,pm-paris19}&\cite{pwma-19}
  \\
\multirow{2}{*}{$m$}
&8($\times 3$),
&\multirow{2}{*}{4 ($\times 24$)}
&\multirow{2}{*}{5 ($\times 60$)}
&\multirow{2}{*}{7 ($\times 60$)}
&13($\times 4$),
&\multirow{2}{*}{12($\times 54$)}
&33($\times 6$), 
&193($\times 6$),
\\
&3, 2, 1
&
&
&
&7($\times 144$)
&
&4, 3
&12, 4
\\
\hline
\end{tabular}
\label{T:masters}
\end{table}
\renewcommand{\arraystretch}{1}

First, for thousands of 4-dim MMP hypergraphs we checked, it turns
out that those with odd number of hyperedges predominantly have
vertices with even multiplicities. The program {\textsc{One}} gives
vertex multiplicities $m$. For smaller sets, they can be verified
by hand (see, e.g., figures in \cite{pavicic-pra-17}), but for
the bigger ones, it would be a really demanding endeavour. So, as
an example we consider a subclass with 17 hyperedges from the 4-dim
class 156-249 [27] shown in Fig.~\ref{fig:for17}. We also contrast
it with a subclass with 18 hyperedges from the same class shown
in the Appendix \ref{app4} which exhibits a prevalent number of odd
multiplicities, once $m=2$ (dominant in all MMPs) is excluded.
Notice that any KS MMP with a parity proof must have an odd number
of edges.

Second, multiplicities of vertices uniquely characterize master MMP
hypergraphs we use to generate all known MMP hypergraphs classes
from. Master sets that are generated from symmetric geometry or
from symmetric polytopes or from symmetric vector-generated MMPs
exhibit large and unique multiplicities $m$, while with
asymmetric vector-generated ones we have $n$ (=dimension)
bigger $m$'s followed by multiple occurrences of one or two smaller
$m$'s, as shown in Table \ref{T:masters}. We can see that 4-dim
master 24-24 consist of 24 vertices all of which have multiplicity
$m=4$, 60-75 of 60 vertices with $m=5$, 60-105 of 60 vertices with
$m=7$, etc. The bigger the asymmetric vector-generated MMP
hypergraphs are, the more $m$'s they contain. E.g., 4-dim the KS MMP
hypergraph master 1132-2460 (not shown in Table \ref{T:masters})
contains $m=79$ four times, and then 42, 36, etc, down to 1
(altogether 16 different $m$'s) in multiple occurrences. The KS MMP
hypergraph master 1132-2460 contains the 60-75 master.

\bigskip
\subsection{\label{subsec:3d}3-dim MMP hypergraphs}

In \cite{pavicic-pra-17} we gave figures and strings of 3-dim
Bub \cite{bub}, Conway-Cohen \cite{peres-book}, Peres
\cite{peres}, and original Kochen-Specker \cite{koch-speck}
critical MMPs: 49-36, 51-37, 57-40, and 192-118, respectively.
Renewed figures are given in the figure in Appendix \ref{app3d}.
New 3-dim MMP hypergraphs, mostly obtained in 
\cite{pavicic-pra-22}, are given in Fig.~\ref{fig:3dup1}. 
Their properties are in Table \ref{T:3d}.
As for the $\omega$ components in Table \ref{T:3d},
$\omega=e^{2\pi i/3}=(-1+i\sqrt{3})/2$.
Note that proving ortogonalities between vertices containing
complex vectors require complex conjugate dot products.

\begin{figure}[h]
  \flushleft
  \includegraphics[width=1\textwidth]{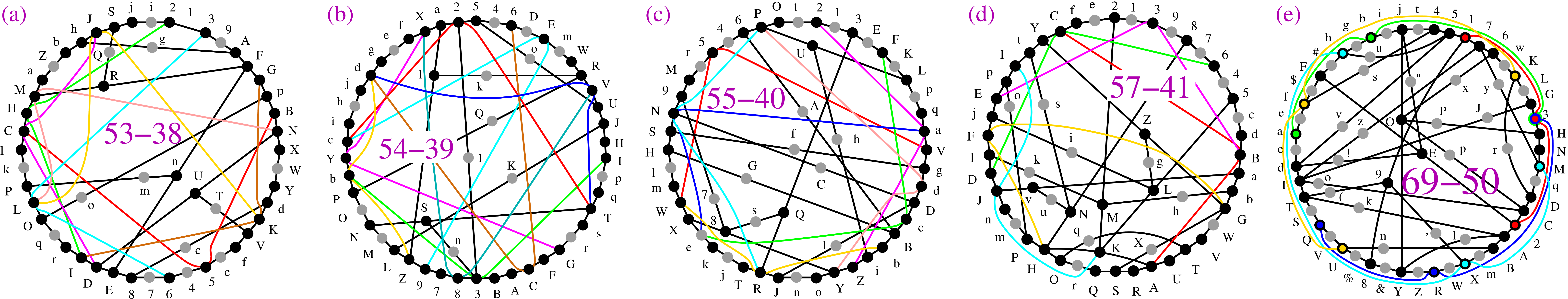}
\caption{(a-c) Critical 3D MMPHs generated by the components
  $\{0,\pm 1,\pm 2,5\}$; (a) the only 53-38; 22-gon; (b) one of the
  eight 54-39s; 23-gon; (c) the only 55-40; 22-gon; (d) the only
  57-41---the smallest MMPH generated by
  $\{0,\pm1,2,\pm5,\pm\omega,2\omega\}$; 21-gon; (e) that smallest
  MMP hypergraph 69-50 generated by
  $\{0,\pm\omega,2\omega,\pm\omega^2,2\omega^2\}$; 24-gon.}
\label{fig:3dup1}
\end{figure}

As explained in
\cite{larsson,pmmm05a,pmmm05a-corr,held-09,pavicic-pra-17,pavicic-entropy-19},
in order to be KS sets the aforementioned original MMP hypergraphs
must have 49, 51, 57, and 192 vertices/vectors, respectively, not 33,
31, 33, and 117 as often stated in the literature and even in the
original papers. The latter versions of the sets are those with
$m=1$ vertices dropped. They are not KS sets but are contextual
non-binary MMP hypergraphs. The same holds for all the other MMP
hypergraphs, e.g., for those in Fig.~\ref{fig:3dup1}.

\renewcommand{\arraystretch}{1.2}
\begin{table}[h]
\caption{Terms for the inequalities of 3-dim
  KS sets: $\alpha_r^*$-inequality: $\alpha\le\alpha_r^*$,
  v-inequality: $HI_{cM}<l$ and e-inequality: $l_{cM}<l$;
  $m_M$ is the maximal $m$. Notice that the
  $\alpha_r^*$-inequality is violated for all MMP hypergraphs.}
  \center
\setlength{\tabcolsep}{2.7pt}
\begin{tabular}{ccccccccccc} 
\hline
\multirow{2}{*}{dim}&\multirow{2}{*}{KS hypergraphs}&$HI_{cM}$&\multirow{2}{*}{$HI_{cm}$}&\multirow{2}{*}{$l_{cM}$}&\multirow{2}{*}{$l_{cm}$}&\multirow{2}{*}{$l$}&\multirow{2}{*}{$m_M$}&\multirow{2}{*}{$\alpha_r^*$}&\multirow{2}{*}{crit.}&vector
\\
&&$\leftrightarrow\alpha$&&&&&&&&components
\\
\hline
\multirow{11}{*}{\rotatebox{90}{3-dim MMP hypergraphs}}
&Bub's 49-36 \cite{bub} 
&21
&11
&35
&24
&36
&4
&$16.\dot{3}$
&yes
&$\{0,\pm 1,\pm 2,5\}$
\\                    
&Conway-Kochen's 
&\multirow{2}{*}{22}
&\multirow{2}{*}{13}
&\multirow{2}{*}{36}
&\multirow{2}{*}{26}
&\multirow{2}{*}{37}
&\multirow{2}{*}{4}
&\multirow{2}{*}{17}
&\multirow{2}{*}{yes}
&\multirow{2}{*}{$\{0,\pm 1,\pm 2,5\}$}
  \\
&51-37 \cite{peres-book}  
&
&
&
&
&
&
&
&
  \\
& 53-38\cite{pavicic-pra-22} 
&21
&12
&37
&27
&38
&4
&$17.\dot{6}$                                                             &yes
&$\{0,\pm 1,\pm 2,5\}$
\\
& 54-39\cite{pavicic-pra-22} 
&23
&13
&38
&27
&39
&4
&18
&yes
&$\{0,\pm 1,\pm 2,5\}$
\\
   & 55-40\cite{pavicic-pra-22} 
&23
&13
&39
&27
&40
&4
&$18.\dot{3}$
&yes
&$\{0,\pm 1,\pm 2,5\}$
\\
&Peres' 57-40 \cite{peres}
&27
&15
&39
&31
&40
&4
&19
&yes
&$\{0,\pm 1,\pm\sqrt 2,3\}$ 
\\
& 57-41\cite{pavicic-pra-22} 
&24
&13
&40
&29
&41
&5
&19
&yes
&$\{0,\pm 1,2,\pm 5,\pm\omega,2\omega\}$
\\
& 69-50 
&36
&21
&49
&40
&50
&4
&23
&yes
&$\{0,\pm\omega,2\omega,\pm\omega^2,2\omega^2\}$
\\
&Kochen-Specker's
&\multirow{2}{*}{75}
&\multirow{2}{*}{63}
&\multirow{2}{*}{116}
&\multirow{2}{*}{99}
&\multirow{2}{*}{118}
&\multirow{2}{*}{9}
&\multirow{2}{*}{64}
&\multirow{2}{*}{yes}
&24 components 
\\
&192-118 \cite{koch-speck}
&
&
&
&
&
&
&
&
&$\to$ Ref.~\cite{pavicic-pra-22}
\\
\hline
\end{tabular}
\label{T:3d}
\end{table}
\renewcommand{\arraystretch}{1}

Their $HI_{cM}$, $HI_{cm}$, $l_{cm}$, $l_{cM}$, and $m_M$ are
given in Table \ref{T:3d}. The e$_{Max}$-inequalities are
trivial for the critical MMP hypergraphs for which we have
$l_{cM}=l-1$.  For the KS 192-118, in 100,000 runs on a
supercomputer, we obtained $l_{cM}=l-2=116$. But our program
{\textsc{One}} for finding $l_{c}$ is probabilistic
and an exhaustive search would not allow parallel computation what
means too lengthy a computation. Their e$_{min}$-inequalities read
$24<36$, $26<37$, \dots, $99<118$. They would allow for a
more robust implementation. Cf.~7-dim case at the end of
Sec.~\ref{subsec:78d}.

In the figure (d) from Appendix \ref{app3d} we give the 192-118
KS MMP hypergraph. Notice that the original figure of Kochen and
Specker \cite[p.~69]{koch-speck} is neither a graph nor a hypergraph.
Its points {\tt a} and {\tt p}$_0$, {\tt b} and {\tt q}$_0$, {\tt c}
and {\tt r}$_0$ \cite[p.~69]{koch-speck} are actually single
vertices, respectively, and lines between them are not edges but
only indications of merged dots what makes their figure
together with comments in its caption just a set of instructions
on how to design a proper hypergraph, what we did in
\cite[Fig.~6]{pmmm05a,pmmm05a-corr} and \cite[Fig.~19]{pavicic-pra-17}
and here. 

Surprisingly, Budroni, Cabello, G{\"u}hne, Kleinmann and Larsson
\cite[Fig.~1]{budroni-cabello-rmp-22} copied the main part of
the figure from \cite[Fig.~7.8]{svozil-book-ql}, or
\cite[Fig.~6]{pmmm05a,pmmm05a-corr}, or \cite[Fig.~19]{pavicic-pra-17} 
(without citing the sources) and cut off parts of twelve of its
hyperedges thus making their Kochen-Specker figure inconsistent---it
is, like the original Kochen-Specker's figure, neither a graph nor a
hypergraph. In the caption of their figure, they call it a graph.
However, in the figure itself they substituted the MMP hypergraph
version of $\Gamma_0$ from \cite[Figs.~7.5,7.8]{svozil-book-ql} for
a graph version from \cite[Fig.~on p.~68]{koch-speck} shown in the
figure (d) from Appendix \ref{app3d} as $\Gamma_0'$ and
$\Gamma_0''$. So, \cite[Fig.~1]{budroni-cabello-rmp-22} shown
here in Fig.~\ref{fig:3d-bud-ks}(a) should be an MMP hypergraph,
but it is not. To see this, let us look at two red hyperedges in
Fig.~\ref{fig:3d-bud-ks}(b)) {\tt 2-10-9} and {\tt 2-12-13}.
The caption of \cite[Fig.~1]{budroni-cabello-rmp-22}
(here: Fig.~\ref{fig:3d-bud-ks}(a)), in effect, reads:
``node 2 is orthogonal to all nodes connected to the red
{\bf edge}s. Similarly for the green and [blue] nodes.''
{\em Nodes\/} are hypergraph vertices, but that what the
{\em nodes\/} (e.g., the {\em node\/} {\tt 2}) are ``connected'' to
(e.g., {\tt 9-10} or {\tt 12-13} in Fig.~\ref{fig:3d-bud-ks}(a))
are neither graph ``edges'' nor hypergraph ``hyperedges.''
They are just lines connecting dot {\tt 9} with dot {\tt 10}, etc.
All that confuses the reader who, after more than 50 years of the
first appearance of the iconic KS set, deserves references to its
unambiguous hypergraph presentation as given in
\cite{svozil-book-ql,pmmm05a,pmmm05a-corr,pavicic-pra-17} and here in
Fig.~\ref{fig:3d-bud-ks}(b).

\begin{figure}[h]
  \begin{center}
  \includegraphics[width=0.9\textwidth]{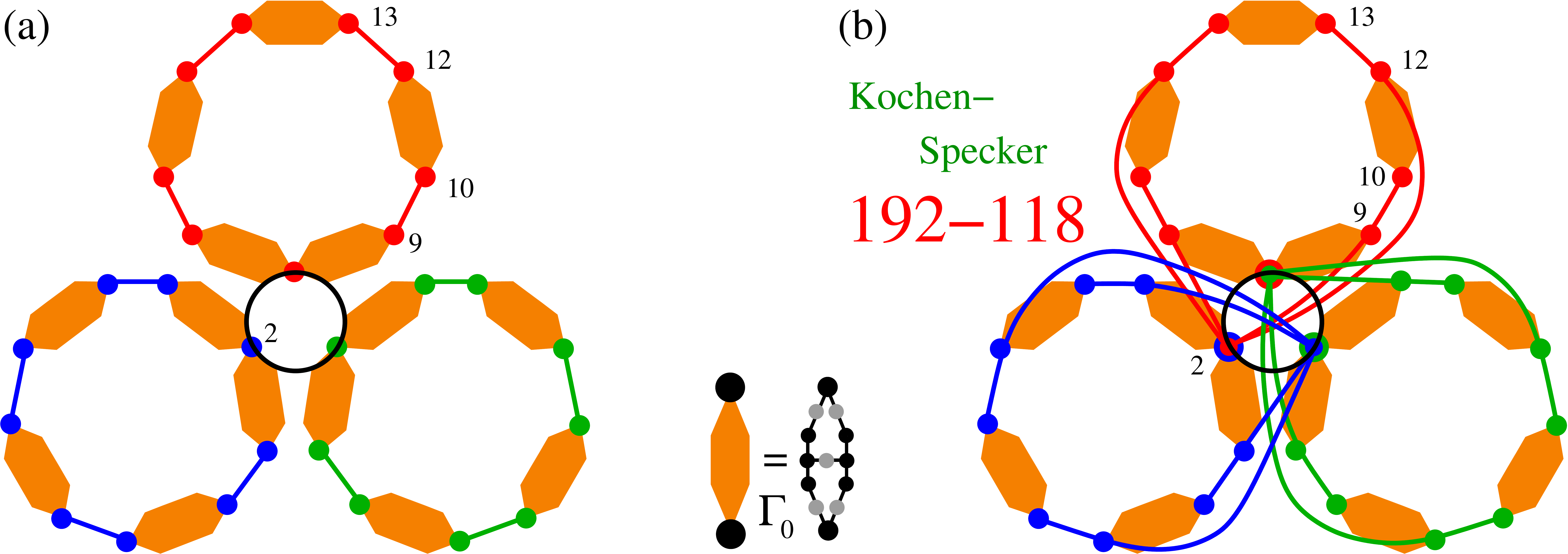}
\end{center}
\caption{(a) Graphics according to
  \cite[Fig.~1]{budroni-cabello-rmp-22}; lines {\tt 9-10}
  and {\tt 12-13} represent neither edges nor hyperedges; they
  are just pieces of hyperedges {\tt 2-9-10} and {\tt 2-12-13},
  respectively; (b) KS MMP hypergraph according to
  \cite[Fig.~7.8]{svozil-book-ql}, or \cite[Fig.~6]{pmmm05a,pmmm05a-corr}, or
  \cite[Fig.~19]{pavicic-pra-17}, or Fig.~\ref{fig:3dup}(d).}
\label{fig:3d-bud-ks}
\end{figure}

In the figure (a,b,c) from Appendix \ref{app3d}, vertices with $m=1$,
are shown as gray dots. E.g., Conway-Kochen's 51-37 hypergraph has 20
such vertices, and this is why Conway-Kochen's 51-37 is often called
a KS set with 31 vertices ($51-20=31$) in the literature.

\medskip
Peres' 57-40 KS is characterised by its coordinatization derived
from the $\{0,\pm 1,\pm\sqrt 2,3\}$ vector components.
By means of vector components $\{0,\pm 1,\pm\sqrt 2\}$, used by
Peres \cite{peres}, in a 3-dim space we can only build 49 vectors,
while in the 57-40 KS MMP hypergraph  there are 57 of them, meaning
that eight vertices cannot have a vector representation at all and
that Escher's ``impossible Waterfall''
\cite{mermin93,penrose-02,cass-gal-05} geometry (mapping of Peres'
set onto the configuration of three interpenetrating cubes) cannot
represent it. To build a KS set, all three vertices in every
hyperedge/triple must be realisable via 3-dim mutually orthogonal
vectors, irrespective of whether we make use of all three of them
(while postprocessing measurement data) or not. They do live in a
3-dim space and must be there, virtual or actual. 

\medskip
We also stress here that the caption of 
\cite[Table 1]{budroni-cabello-rmp-22} is incorrect and misleading
in the following sense. It reads ``[in] KS proofs [3-dim Bub,
Conway-Kochen, and Peres'] the\dots numbers inside parenthesis
(33,31,33) are the numbers used in the contradiction, numbers
outside (49,51,57) counts all vectors when completing the bases.''
But as we show in
\cite{pmmm05a,pmmm05a-corr,pavicic-pra-17,pavicic-entropy-19} not
the 33-36, 31-37, and 33-40, but the 49-36, 51-37, and 57-40 MMP
hypergraphs are critical KS MMP hypergraphs which are therefore
primarily ``used in the contradiction'' of the KS theorem since
they are the KS sets while the former ones are not. We show
in \cite{pavicic-entropy-19} that many MMP hypergraphs one obtains
from 49-36, 51-37, and 57-40 by removing chosen $m=1$ vertices,
down to 33-36, 31-37, and 33-40 MMPs, are non-binary contextual
MMP hypergraphs. However, they are {\em not\/} KS MMP hypergraphs
by definition and therefore they are {\em not\/} ``KS proofs.''

\medskip
Still, excluding the $m=1$ vertices in a postprocessing of data
generated by measurements provide us with an important method of
obtaining arbitrary many smaller contextual non-binary MMP
hypergraphs from both non-binary and binary MMP hypergraphs.
This is due to an important structural difference between the
MMP hypergraphs with hyperedges containing the maximal number
of vertices per hyperedge and those with less vertices in some
hyperedges. If the former MMPs are critical (as, e.g., all MMPs
from the figure from Appendix \ref{app3d}, then no stripping of
their hyperedges would lead to another non-binary MMP. However,
stripping of their $m=1$ vertices may yield non-critical MMPs
which may generate smaller non-binary critical MMPs which may be
stripped again and may yield even smaller criticals. Of course,
because of the stripping, none of the obtained smaller MMPs is a
proper subhypergraph of an MMP we start with. They are all
$\overline{\rm subhypergraphs}$. 

\medskip
In \cite{pavicic-entropy-19} we generated thousands of smaller
non-binary MMP critical hypergraphs from all four bigger MMPs
given in the figure in Appendix \ref{app3d}, the smallest of which
are shown in Fig.~\ref{fig:alpha-3d}. As a rule, all small critical
non-binary MMP hypergraphs do satisfy the $\alpha_r^*$-inequality.
Notice that the 14-12 MMP hypergraph which does not satisfy it, is
not critical and that the critical 13-11 which it contains does
satisfy the inequality. 

\medskip

\renewcommand{\arraystretch}{1}
\begin{table}[ht]
  \caption{Terms for the inequalities of 3-dim
  contextual non-binary $\overline{\rm subhypergraphs}$ from
  Fig.~\ref{fig:alpha-3d}: $\alpha_r^*$-inequality:
  $\alpha\le\alpha_r^*$, v-inequality: $HI^m_{cM}<l$, and
  e$_{Max}$-inequality: $l_{cM}<l$; $m_M$ is the maximal $m$.}
  \center
\setlength{\tabcolsep}{6pt}
\begin{tabular}{ccccccccccc} 
\hline
\multirow{2}{*}{dim}&\multirow{2}{*}{KS MMPs}&$HI_{cM}$&\multirow{2}{*}{$HI_{cm}$}&\multirow{2}{*}{$l_{cM}$}&\multirow{2}{*}{$l_{cm}$}&\multirow{2}{*}{$l$}&\multirow{2}{*}{$m_M$}&\multirow{2}{*}{$\alpha_r^*$}&\multirow{2}{*}{crit.}&vector
\\
&&$\leftrightarrow\alpha$&&&&&&&&components\\
\hline
\multirow{2}{*}{\rotatebox{90}{3D MMPs\ }}
&8-7
&3
&3
&6
&6
&7
&2
&3.5
&yes
&$\{0,\pm 1\}$
\\
&14-11
&5
&5
&10
&10
&11
&2
&5.5
&yes
&$\{0,\pm 1,2\}$
\\
&11-10
&4
&4
&9
&9
&10
&3
&4.75
&yes
&$\{0,\pm 1,2\}$
\\
&14-12
&6
&5
&11
&11
&12
&3
&$4.91\dot{6}$
&no
&$\{0,\pm 1,2\}$ \\
\hline
\end{tabular}
\label{T:alpha-3d}
\end{table}
\renewcommand{\arraystretch}{1}

\begin{figure}[h]
\begin{center}
  \includegraphics[width=\textwidth]{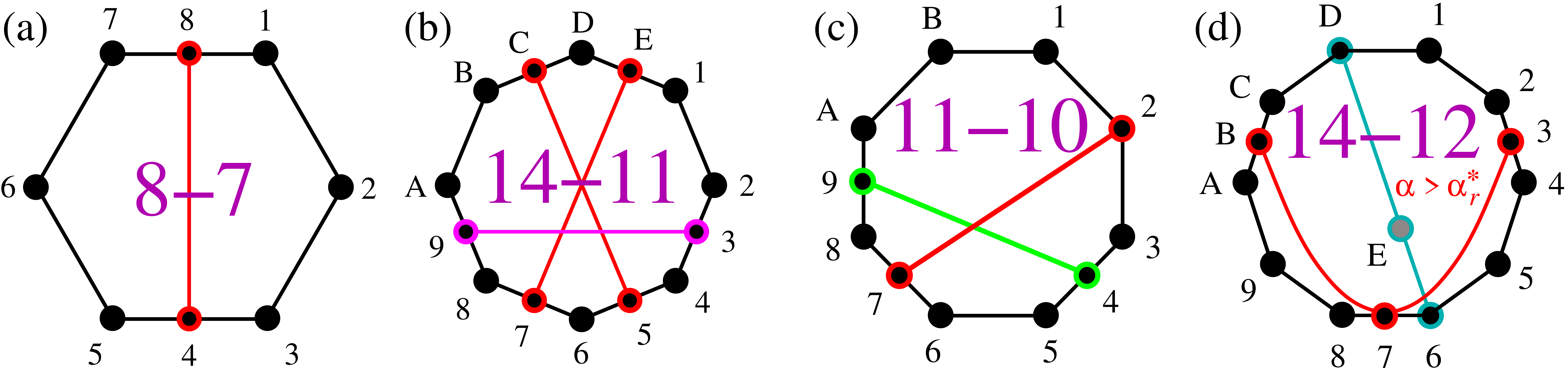}
\end{center}
\caption{(a) 8-7 MMP hypergraph ($\Gamma_0$ from
  Fig.~\ref{fig:3d-bud-ks}) is a $\overline{\rm subhypergraph}$ of
  Bub's 49-36 and Yu-Oh's 13-16 (Fig.~\ref{fig:yu-oh-mmp});
  note that Yu-Oh's 13-16 \cite{yu-oh-12} is not critical and
  that its filled version 25-16 is a subgraph of Peres's 57-40
  \cite{pavicic-entropy-19};
  (b) $\overline{\rm subhypergraph}$ of Bub's 49-36;
  (c) $\overline{\rm subhypergraph}$ of both Bub's 49-36 and
  Conway-Kochen's 49-36; (d) $\overline{\rm subhypergraph}$ of
  Conway-Kochen's 49-36; in contrast to the previous MMPs it
  violates the $\alpha_r^*$-inequality: $6>4.91\dot{6}$; its
  filled MMP can have a coordinatization from the $\{0,\pm 1,2\}$
  component set; (a,b) do have a parity proof, while (c,d) do not;
  (a,b,c) are critical, while (d) is not; 14-12 without the cyan
  hyperedge is a 13-11 critical MMP with a parity proof.}
\label{fig:alpha-3d}
\end{figure}

\subsection{\label{subsec:4d}Small 4-dim MMP hypergraphs
    and the smallest MMP hypergraph that exists}

In Sec.~\ref{subsec:3d} we obtained small 3-dim critical
non-binary MMP hypergraphs from big critical non-binary MMP
hypergraphs. In this section we consider small 4-dim critical
non-binary MMP hypergraphs we generate from big non-binary MMP
hypergraphs by the same method we used in Sec.~\ref{subsec:3d}. 

In Table \ref{T:4d1} we present $HI_{cM}$, $HI_{cm}$, $l_{cM}$ and
$l_{cm}$ values for chosen MMP subhypergraphs of the KS master MMP
hypergraph 636-1657 \cite{pwma-19}. Among billions of them that
we generated in an automated fashion from the 636-1657, we have
chosen a number of MMP hypergraphs some of which were also
previously obtained in the literature via other methods.

None of them contain vertices with multiplicity $m=1$, i.e., they
are structurally {\em dense}. Since one can easily assign ASCII
characters to the vertices we do not show them in Fig.~\ref{fig:22}.
An MMP hypergraph is characterized by its structure, not by a
specification of characters assigned to its vertices.

\begin{figure*}[ht]
\begin{center}
  \includegraphics[width=\textwidth]{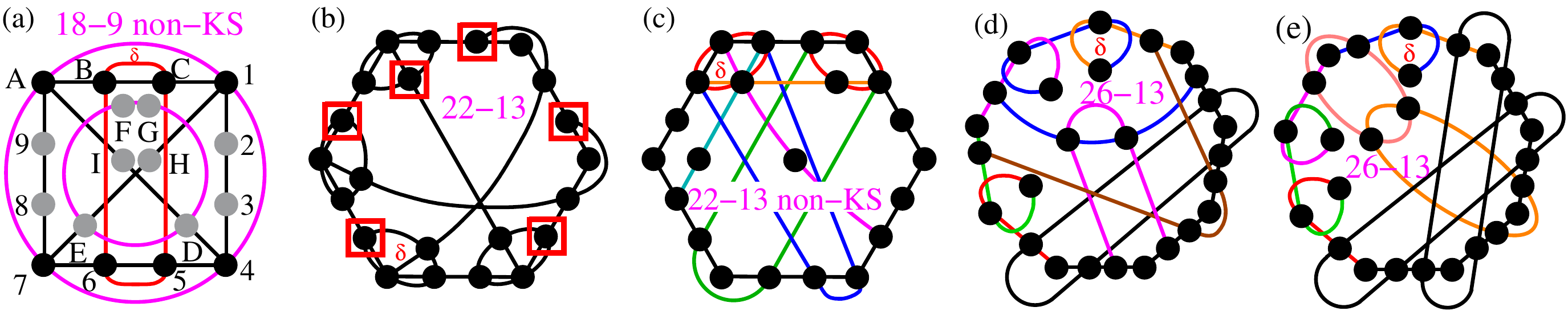}
\end{center}
\caption{(a-d) Figures of MMP hypergraphs from Table \ref{T:4d1};
  (a) 18-9 non-KS---a subhypergraph of Peres' 24-24; (b) critical
  KS 22-13 from the 636-1657 class; vertices that contribute to
  $\alpha=HI_{cM}=6$ are squared in red; note that
  $\alpha_r^*=\frac{22}{4}=5.5<\alpha$; (c) 22-13 non-KS from
  the 60-105 class which is a subclass of the 636-1657 class;
  (d) critical KS 26-13 which belongs to both classes;
  (e) critical KS 26-13 which is from the 636-1657 class but which
  does not belong to the 60-105 class.}
\label{fig:22}
\end{figure*}

All MMP hypergraphs shown in Fig.~\ref{fig:22} exhibit the maximal
level of the so-called $\delta$ feature (pairs of them to intersect
each other twice, at two vertices) which characterizes most of
the KS MMP hypergraphs from the 636-1657 class. Notice that the
$\delta$ feature characterizes all MMP hypergraphs---due to their
definition (Def.~\ref{def:MMP-string}(4.))---and not just KS ones
(cf.~Fig.~\ref{fig:22}(a,c)).

The 636-1657 class, whose critical KS sets overwhelmingly exhibits the
$\delta$ feature, is completely disjoint from two other 4-dim classes
of KS criticals 300-675 and 148-265, which do not exhibit the $\delta$
feature at all, and which are in turn completely disjoint from each
other. Moreover, the non-KS sets and the non-critical KS sets from
the 636-1657 class also possess the $\delta$ feature.

\begin{table}[ht]
\caption{Parameters of the considered 4-dim MMP hypergraphs.
  KS ones are from the 636-1657 class, apart from the 60-75 master
  which is from the 300-675 class. Non-KS ones are from the 24-24 and
  60-105 classes, respectively.}
  \center
\setlength{\tabcolsep}{8pt}
\begin{tabular}{@{}l*{7}{c}@{}} 
\hline
MMP hypergraphs&$HI_{cM}$&$HI_{cm}$&$l_{cM}$&$l_{cm}$&crit.&vec.~compon.\\
\hline
18-9 \cite{cabell-est-96a} 
&4
&3
&8
&6
&yes
&(0,$\pm 1$)
\\
18-9 {\color{red}non-KS} [{\color{Blue}here}]
&6
&4
&9
&7
&no
&(0,$\pm 1$)
\\
20-11 \cite{pmmm05a,pmmm05a-corr}  
&5
&3
&10
&8
&yes
&(0,$\pm 1$)
\\
21-11 \cite{pavicic-book-13,pm-entropy18}  
&5
&3
&10
&8
&yes
&(0,$\pm 1$,$i$)
\\
22-13 \cite{pmmm05a,pmmm05a-corr}
&6
&3
&12
&$8$
&yes
&(0,$\pm 1$)
\\
22-13 {\color{red}non-KS} [{\color{Blue}here}]
&8
&3
&13
&9
&no
&(0,$\pm 1$)
\\
26-13 [{\color{Blue}here}]
&6
&5
&12
&10
&yes
&(0,$\pm 1$,$\pm i$,2)
\\
28-17 [{\color{Blue}here}]
&8
&5
&16
&12
&yes
&(0,$\pm 1$,$\pm i$)
\\
29\dots 34-17 [{\color{Blue}here}]
&8
&5
&16
&12
&yes
&(0,$\pm 1$,$\pm i$)
\\
35-17 [{\color{Blue}here}]
&10
&7
&16
&14
&yes
&(0,$\pm 1$,$\pm i$)
\\
35-17 {\color{red}non-KS} [{\color{Blue}here}]
&11
&7
&17
&13
&yes
&(0,$\pm 1$,$\pm i$)
\\
30-18 [{\color{Blue}here}]
&8
&5
&17
&11
&yes
&(0,$\pm 1$,$\pm i$)
\\
31\dots 36-18 [{\color{Blue}here}]
&11
&5
&17
&12
&yes
&(0,$\pm 1$,$\pm i$)
\\
37-18 [{\color{Blue}here}]
&11
&7
&17
&15
&yes
&(0,$\pm 1$,$\pm i$)
\\
24-24 \cite{peres} (master)
&5
&3
&20
&12
&no
&(0,$\pm 1$)
\\
60-105 \cite{waeg-aravind-jpa-11} (master)
&12
&7
&84
&70
&no
&(0,$\pm 1$,$i$)
\\
60-75 \cite{mp-nm-pka-mw-11} (master)
&13
&9
&65
&45
&no
&(0,$\pm 1$,$\pm\frac{\sqrt{5}\pm 1}{2}$)
\\
\hline
\end{tabular}
\label{T:4d1}
\end{table}

That means that the $\delta$ feature characterises classes of
hypergraphs although it does not determine the contextuality---both,
classes that possess it as well as those that do not are
contextual. The $\delta$ feature determines MMP classes
through their structure and coordinatization, though. For instance:
\begin{enumerate}[label=(\roman*)]
\item it is absent in 3-dim MMP hypergraphs due to their definition; 
  3-dim MMP hypergraphs are equivalent to Greechie diagrams
  \cite{pavicic-entropy-19}, but $n$-dim, $n\ge 4$ MMP hypergraphs
  are not, exactly due to the $\delta$-feature which is not permitted
  to any Greechie diagram due to its definition; also the smallest
  loops in any Greechie diagram in dimension are pentagons by its
  definition \cite{mp-7oa}; the smallest loops in 3-dim MMP
  hypergraphs are pentagons due to their geometry and that is
  the reason why they are equivalent only in the 3-dim space
  \cite{pmmm05a,pmmm05a-corr}; 
\item it allows the smallest hypergraphs in the 636-1657 MMP
  hypergraph class to be smaller than the smallest ones in the
  300-675 and 148-265 MMP hypergraph classes that do not exhibit
  the $\delta$ feature;
\item it is present in the MMP hypergraph which represents the
  exclusivity graph \cite{magic-14} and plays an essential 
  role in the quantum computation theory (See
  Sec.~\ref{subsec:magic});
\item it characterizes all higher dimensional MMP hypergraphs;
\item in the 4-dim Hilbert space it resides in the complex spaces,
  while it is absent in the real ones \cite{pavicic-pra-17,pwma-19}.
\end{enumerate}

Apart from these characteristics, parameters obtained for the MMPs
from the 300-675 and 148-265 classes do not fundamentally differ
from those obtained for the 636-1657 class and therefore we do not
give equivalent set of examples for the former classes.

However, there is another feature of all non-KS MMP
hypergraphs like the 18-9 shown in Fig.~\ref{fig:22}(a). Let us
first analyze the very non-KS 18-9. Its coordinatization is
generated by the $\{0,1,-1\}$ vector components:
{\tt 1}=(0,0,0,1), {\tt 2}=(1,-1,0,0), {\tt 3}=(1,1,0,0), {\tt 4}=(0,0,1,0),\break
{\tt 5}=(1,0,0,1), {\tt 6}=(1,0,0,-1), {\tt 7}=(0,1,0,0), {\tt 8}=(0,0,1,-1), {\tt 9}=(0,0,1,1), {\tt A}=(1,0,0,0), {\tt B}=(0,1,1,0),\break {\tt C}=(0,1,-1,0), {\tt D}=(0,1,0,1), {\tt E}=(1,0,1,0), {\tt F}=(1,1,-1,-1), {\tt G}=(1,-1,-1,1), {\tt H}=(0,1,0,-1), {\tt I}=(1,0,-1,0).

Its $\overline{\rm subhypergraph}$ with all $m=1$ vertices removed
is shown in Fig.~\ref{fig:18-9-non}(a). It is a contextual non-binary
MMP hypergraph (Def.~\ref{def:n-b}) which contains four critical
MMP subhypergraphs shown in Figs.~\ref{fig:18-9-non}(b,f,g,h).
A $\overline{\rm subhypergraph}$ of Fig.~\ref{fig:18-9-non}(b) with
all $m=1$ vertices removed is shown in Fig.~\ref{fig:18-9-non}(c).
It contains a critical MMP subhypergraph 3-3 shown in
Fig.~\ref{fig:18-9-non}(d). It is the smallest contextual non-binary
MMP hypergraph that exists.

\begin{figure}[ht]
\begin{center}
  \includegraphics[width=0.99\textwidth]{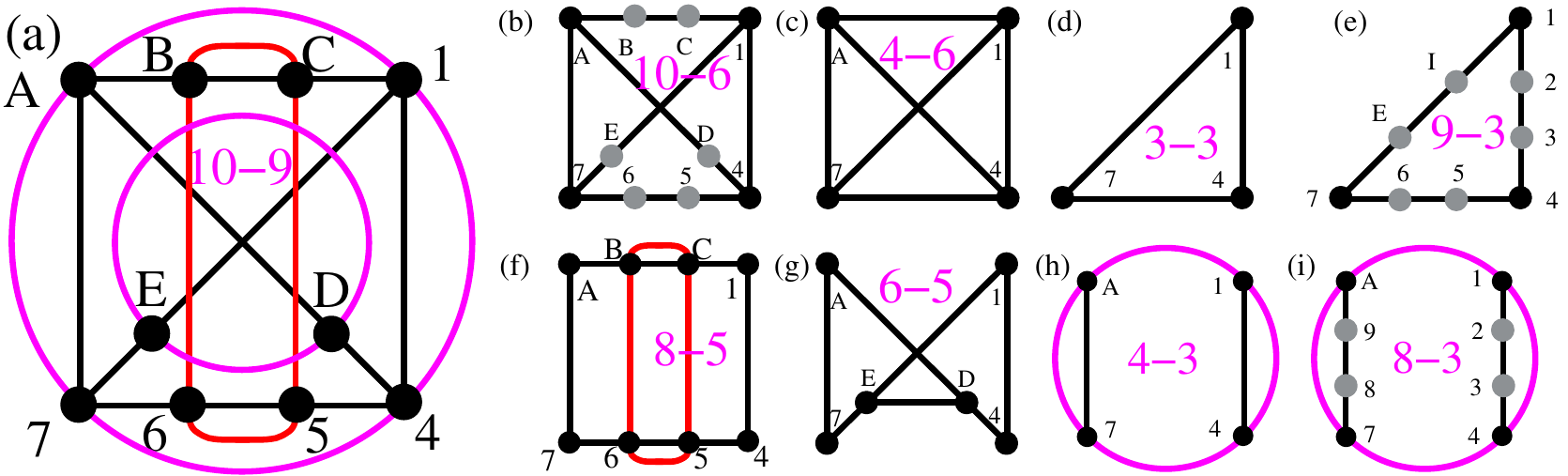}
\end{center}
\caption{(a) $\overline{\rm subhypergraph}$ of 18-9 non-KS MMP
  hypergraph shown in Fig.~\ref{fig:22}(a); (b,f,g,h) critical
  sub-hypergraphs of 10-9; (c) non-critical
  $\overline{\rm subhypergraph}$ of 10-6; (d) a critical
  subhypergraph of 4-6---the smallest MMP hypergraph that
  exists: 3-3; (e) filled 3-3; (i) filled 4-3.}
\label{fig:18-9-non}
\end{figure}
\begin{result}\label{smallest}{
  The critical contextual non-binary 4-dim MMP hypergraph {\rm 3-3}
  with 3 vertices and 3 hyperedges with coordinatization shown in
  Fig.~\ref{fig:18-9-non}(d) is the smallest existing
  contextual hypergraph with a coordinatization in any dimension
  because the 3-dim {\rm 3-3} non-binary {\rm MMP} does not have
  a coordinatization {\rm\cite{pavicic-entropy-19}}}.
\end{result}

\begin{resultt}\label{sub-non}
  $\overline{Subhypergraphs}$ of
    noncontextual binary {\rm MMP} hypergraphs as well as of their
    subhypergraphs and $\overline{subhypergraphs}$ are
    overwhelmingly contextual, i.e., they are mostly non-binary
    {\rm MMP} hypergraphs.
\end{resultt}

We confirmed this feature on thousands of binary MMP hypergraphs.
It enables us to obtain a much greater varieties of contextual sets
than via the KS or the operator generation, including obtaining
a plethora of small sets in any dimension. For the time being,
we have carried out a massive generation neither of binary MMP
hypergraphs nor of non-binary MMP hypergraphs that would follow
from the binary ones via the \ref{sub-non} feature
(noncontextual $\longrightarrow$ contextual).
We leave that for a future project.

\phantom{We leave that for a future project.}

\subsection{\label{subsec:magic}Graphs,
  hypergraphs, contextuality,  experiments, and
  computation}

In this section we review the usage of the graph formalism and the
GLS inequality via the following examples: ``Experimental
Implementation of a Kochen-Specker Set of Quantum Tests'' by
{D'A}mbrosio, Herbauts, Amselem, Nagali, Bourennane, Sciarrino,
and Cabello \cite{d-ambrosio-cabello-13}, ``Contextuality Supplies
the `Magic' for Quantum Computation}'' by Howard, Wallman, Veitech,
and Emerson \cite{magic-14}, and ``Graph-Theoretic Approach to
Quantum Correlations'' by Cabello, Severini, and Winter
\cite{cabello-severini-winter-14}.

In Fig.~\ref{fig:gr-hyp} we see that the graph representation (c)
of the 18-9 KS set has three times as many edges as its MMP
hypergraph representation (b). In higher dimensions and
for more vertices and edges the graph representation gets more and
more complicated and graphically unintelligible. Matrix graph
representation also becomes hardly manageable in comparison with
the MMP string representation. That is why Cabello
\cite[Fig.~1]{cabello-08} first adopted the general hypergraph
representation \cite[Fig.~3(a)]{pmmm05a,pmmm05a-corr} for the 18-9 KS. However,
in \cite[Fig.~1(a)]{d-ambrosio-cabello-13} the authors,
surprisingly, abandoned the hypergraph language and adopted the
graph representation shown in Fig.~\ref{fig:graph}(e) that caused
the following inconsistencies.  

In \cite[Fig.~1(a)]{d-ambrosio-cabello-13} 9 edges were added to
the 18-9 KS graph from Fig.~\ref{fig:graph}(d)---in
Fig.~\ref{fig:graph}(e) they are denoted as 3 green and 6 red ones.
This turns the 18-9 set into an 18-18 set whose MMP hypergraph
representation is shown in Fig.~\ref{fig:graph}(g). The measurements
for vertices on these additional 9 edges were carried out and
provided in \cite[Supp.~Material, Table III]{d-ambrosio-cabello-13}.
For instance, for vertices {\tt 1,2,5,10,15,18}, i.e., edges
{\tt 1-10,\ 2-15} and {\tt 5-18}, the probabilities $p_{1,10}$,
$p_{2,15}$, and $p_{5,18}$ were obtained. However, as shown in
\cite{larsson,pmmm05a,pmmm05a-corr,held-09,pavicic-pra-17,pavicic-entropy-19,budroni-cabello-rmp-22}
not just two but four vertices should be measured and should
have a coordinatization in each green and red hyperedge even when we
do not take all of them into account while postprocessing the data.

\begin{figure}[h]
  \includegraphics[width=1\textwidth]{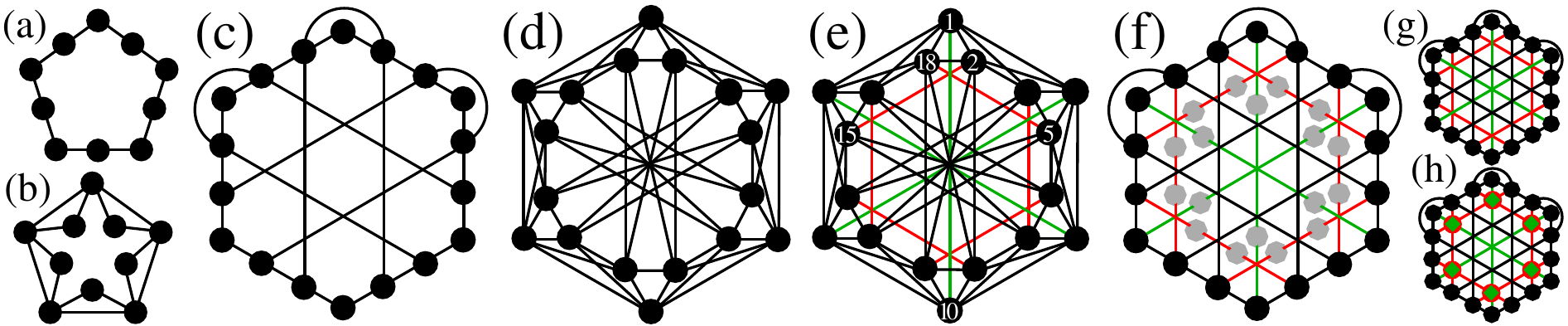}
\caption{(a) MMP representation of the 10-5 pentagon; (b) the
  same set in the graph representation \cite{berge-73}: 10-15;
  (c) the smallest 4-dim KS set in the MMP hypergraph
  representation: 18-9 \cite{pmmm05a,pmmm05a-corr}; (d) the same set in the
  graph representation: 18-54; (e) contextual non-KS set
  implemented in \cite{d-ambrosio-cabello-13} where it is
  misnamed as a KS set---graph representation from
  \cite[Fig.~1(a)]{d-ambrosio-cabello-13}: 18-63; (f) the same
  set, with $m=1$ (gray) vertices added, in the MMP hypergraph
  representation: 36-18; (g) the same set with $m=1$ (gray)
  vertices dropped---non-binary MMP hypergraph 18-18---equivalent
  to 18-63 (e)-graph; (h) the same set with gray $m=1$ vertex
  triples merged into $m=3$ vertices---KS 24-18 MMP
  hypergraph---a subhypergraph of Peres' 24-24 MMP hypergraph.}
\label{fig:graph}
\end{figure}

That means that the 18-vertex set from 
\cite[Fig.~1(a)]{d-ambrosio-cabello-13} is not a ``Kochen-Specker
set'' as claimed in the title of the paper because all edges in a 
4-dim KS set should have 4 vertices and green and red edges have
only two vertices. The missing vertices should be added. In a graph
representation it would be a real mess of lines and therefore
we show them in the MMP hypergraph representation as gray dots in
Fig.~\ref{fig:graph}(f). It should have a coordinatization, but
apparently this filled 36-18 MMP hypergraph (36 vertices and 18
hyperedges) has {\em no\/} coordinatization. We verified, that
MMP 36-18 is not a subhypergraph of any known 4-dim MMP hypergraph
class \cite{pavicic-pra-17,pavicic-entropy-19,pwma-19},
i.e., that there are no known \cite[Tables\ 1\ \&\ 2]{pm-entropy18}
vector components for any available coordinatization. Hence, not
only that the considered set is not a KS set, but the measurement
data themselves in \cite{d-ambrosio-cabello-13} are inconsistent.

One way out of these inconsistencies is to merge triples of
gray vertices at the intersections of hyperedges as shown in
Fig.~\ref{fig:graph}(h), i.e., new measurements should be carried
out for the additional 6 vertices of the new 24-18 MMP hypergraph
which is one of 1233 KS MMP hypergraphs \cite{pmm-2-09}
contained in Peres' 24-24 master set.

Another way out would be to abandon green and red hyperedges and
reduce the implementation to the 18-9 KS MMP hypergraph.

Our second example is the one of Howard, Wallman, Veitech, and
Emerson \cite{magic-14}. They have shown that stabilizer operations
with quantum bits initialized as magic states, i.e., superposition
of states, can be used to purify quantum gates provided they exhibit
contextuality. As a proof that considered sets are contextual
the authors make use of the GLS inequality 
\cite[p.~192]{gro-lovasz-schr-81}.

However, ``There are some subtleties that limit what these results
can say about [qubits] as opposed to larger quantum systems. The
limitation could simply be a vagary of the proof technique used
by the authors''
\cite{bartlett-nature-14}.

Their proof is a kind of a proof by induction and we shall focus
on its first step which elaborates on a two-qubit system, a qubit
being a 2-dim system ($p=2$). The graph $\Gamma$ they make
use of for the purpose is shown in \cite[Fig.~2]{magic-14} and
its MMP hypergraph presentation in Fig.~\ref{fig:magic}(a). Details
on how is $\Gamma$ obtained from a set of entangled projectors
\cite[Eq.~(14)]{magic-14} which are in turn obtained from the
set of stabilizers states \cite[Eq.~(14)]{magic-14} are not
provided. A reference to \cite{cabello-severini-winter-14} is
given instead and we shall come back to it below. 

\begin{figure}[h]
  \includegraphics[width=1\textwidth]{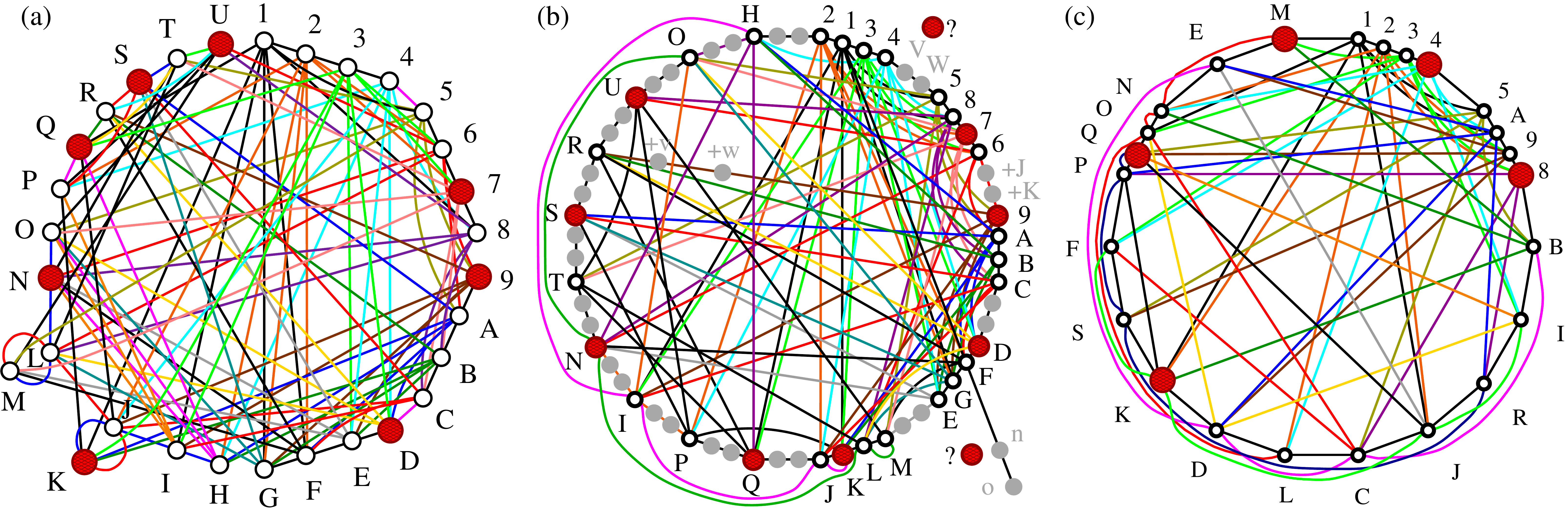}
\caption{(a) $\Gamma$ in MMP notation: 30-108 MMP; (b) filled
  $\Gamma$ with all vertices that belong to only one hyperedge
  ($m=1$); it is a 232-108 KS MMP hypergraph;  vertices {\tt 7,9,D}
  have $m=8$ and {\tt K,N,Q,S,U} have $m=7$; if vertices with $m=1$
  were shown (as, e.g., {\tt +v} and {\tt +w}), that would overcrowd
  the figure. One can avoid such a clutter by extending the hyperedges
  and positioning the vertices outside the loop as, e.g., {\tt n} and
  {\tt o}; (c) the only found critical contained in 232-108 is 152-71
  and when its vertices with $m=1$ are dropped it becomes the 24-71
  MMP shown here; in all three figures, independent vertices are dark
  patterned, red, and enlarged; the ASCII strings for all three
  figures are given in Appendix \ref{app2}.}
\label{fig:magic}
\end{figure}
One can verify that $\alpha(\Gamma)=p^3=8$ holds (e.g., via
\textsc{One}). How $\alpha=8<\alpha^*=9$ (the 2nd line of the proof
of Theorem 1 in \cite{magic-14}) is obtained is not explained in
detail but the approach the authors seem to have applied apparently
runs as follows. $\Gamma$ is a 30-108 non-binary MMP hypergraph
whose string is given in Appendix \ref{app2} and whose graphical
representation is given in Fig.~\ref{fig:magic}(b). The MMP
notation is substituted for the original clique representation of
mutually orthogonal vertices in seven hyperedges which each
contain four vertices ({\tt 1234,5678,9ABC,DEFG,HIJK,JKLM,MLNO}),
altogether 24 vertices, each of which within each of the 7
hyperedges has the probability $p=\frac{1}{4}$ of being detected,
so that their sum of probabilities amounts to 6. The remaining 6
vertices ({\tt P,Q,R,S,T,U}) are apparently assumed to have the
probabilities $p=\frac{1}{2}$ and it is apparently also assumed
that the sum of their $\frac{1}{2}$ probabilities amounts to 3.
That yields the total sum of probabilities equal to
$\alpha^*(\Gamma)=9$ \cite[Proof of Theorem 1]{magic-14}. But
hyperedges that connect two vertices that do not both belong to
the aforementioned 6 vertices, e.g., {\tt 15} or {\tt 1Q}, are
not taken into account in this calculation at all. Let us see
how this can be amended.

The string of the 30-108 MMP given in Appendix \ref{app2} offers us
the following probabilities. Vertex {\tt 1} is in the hyperedge
{\tt 1234} and has the probability of $\frac{1}{4}$ of being
detected. But it is also in the following 7 hyperedges {\tt 15},
{\tt 18}, {\tt 1F}, {\tt 1G}, {\tt 1I}, {\tt 1K}, and {\tt 1Q}
within each of which it has the probability $\frac{1}{2}$ of being
detected. The arithmetic mean of these probabilities is
$(7\frac{1}{2}+\frac{1}{4})/8=\frac{15}{32}$. There are 16
such blocks. The rest are organized in 7-hyperedges blocks:
four $(5\frac{1}{2}+2\frac{1}{4})/7=\frac{3}{7}$,
four $(6\frac{1}{2}+\frac{1}{4})/7=\frac{13}{28}$, and six 
$7\frac{1}{2}/7=\frac{1}{2}$. The total sum of probabilities 
probabilities is $\frac{197}{14}=14.07=\alpha_r^*$, which
differs from $\alpha^*(\Gamma)=9$
\cite[Proof of Th.~1]{magic-14}
  \begin{eqnarray}\label{eq:alpha-alpha-gamma}
    8=\alpha(\Gamma)< 14.07=\alpha_r^*(\Gamma)\ne
    \alpha^*(\Gamma)=9.   
  \end{eqnarray}
The question arises whether $\alpha^*(\Gamma)=\alpha_r^*(\Gamma)$
under another approach. We discus this below. In any case it does
not seem correct to assign the probabilities $\frac{1}{2}$ and
$\frac{1}{4}$ to the aforementioned 6 and 24 vertices based on
their containment in the 2-vertex- and 4-vertex-hyperedges,
respectively, and ignore their containment in 86 hyperedges that
connect vertices in the 2-vertex-hyperedges with those in the
4-vertex-hyperedges or two vertices in different 4-vertex-hyperedges
via 2-vertex-hyperedges. For the time being, let us elaborate on
discarding two states from a tensor product of states of two qubits.

Any two mutually orthogonal vertices from two-vertex-hyperedges
in the 30-108 MMP belong to an edge and therefore to two qubits.
The two-vertex-hyperedge only means that two of four states are
discarded. Thus 202 of 232 vertices are excluded. Each qubit has
a coordinatization and a complete measurement of each edge must
involve all four vertex states, i.e., only complete states (both
vertices from each hyperedges) can build up tensor products of
two qubits. The quantum systems must pass out-ports no matter
whether we take them into account in a later elaboration
of our data or not. Related to that, the claim that $\alpha=8$
should be clarified, because, e.g., in Fig.~\ref{fig:magic} we see
that the vertices {\tt 7,9,D,K,Q,N,S,U,V,n}, etc., are independent
and yield the independence number $\alpha\ge101$ that violates the
inequalities $\alpha<\alpha^*$ and $\alpha<\alpha_r^*$. It remains
to be explained how come that vertices with $m=1$ do not contribute
to the independence set when they actually build up a KS MMP
hypergraph 232-108. If it were a result of the construction of
$\Gamma$, then the role of $\alpha=p^3=8$ should also be explained,
because there are contextual critical non-binary MMP hypergraphs
with higher $\alpha$ which violate Eq.~(\ref{eq:alpha-alpha})
(see Table \ref{T:4d1}).

When we take into account all vertices of all qubits we get 
$\Gamma$(filled) 232-108 KS MMP, whose string is given in
Appendix \ref{app2}. Its parameters are given in Table \ref{T:gamma}.
It is not critical, and the only critical set contained in it, that
we obtained, is the 152-71 critical KS MMP shown in
Fig.~\ref{fig:magic}(d) (where we dropped $m=1$ vertices) whose
hypergraph string is also given in the Appendix \ref{app2}. It
satisfies $64\le\alpha(152$-$71)>\alpha_r^*(152$-$71)=38$.
The string of 152-71 with dropped $m=1$ vertices---24-71 MMP---is
given in the Appendix \ref{app2}. It has 5 independent vertices:
{\tt 4,8,K} ($m=7$) and {\tt Q,M} ($m=6$). Their parameters are
also given in Table \ref{T:gamma}. We obtain
$5=\alpha(24$-$71)<\alpha_r^*(24$-$71)=7.12$.

\renewcommand{\arraystretch}{1.2}
\begin{table}[ht]
\caption{Terms for the  inequalities of 4-dim
  contextual non-binary $\Gamma$ MMP hypergraphs from 
  Fig.~\ref{fig:magic}: $\alpha_r^*$-inequality:
  $\alpha\le\alpha_r^*$ (violated for 232-108 and 152-71),
  v-inequality: $HI^m_{cM}<l$, and e$_{Max}$-inequality: $l_{cM}<l$;
  $m_M$ is the maximal $m$; computer search for vectors formed
  from simple components
  ($0,\pm 1,\pm i,\pm\omega,\pm 2,\pm\sqrt{2},\pm 3,\pm 5$) failed;
  finding of $HI_{cm},l_{cM},l_{cm}$ for 30-180 and 24-71 requires
  tweaking of the program {\textsc{One}} so as to provide us with
  the parameters when some hyperedges (from an $n$-dim space)
  contain less than $n$ vertices, what we have not done as of yet;
  note that non-critical 30-108, 232-108, and 24-71 MMPs generate
  thousands of smaller non-binary $\overline{\rm subhypergraphs}$.}
\center
\setlength{\tabcolsep}{4.5pt}
\begin{tabular}{ccccccccccc} 
\hline
\multirow{2}{*}{dim}&\multirow{2}{*}{KS hypergraphs}&$HI_{cM}$&\multirow{2}{*}{$HI_{cm}$}&\multirow{2}{*}{$l_{cM}$}&\multirow{2}{*}{$l_{cm}$}&\multirow{2}{*}{$l$}&\multirow{2}{*}{$m_M$}&\multirow{2}{*}{$\alpha_r^*$}&\multirow{2}{*}{crit.}&vector
\\
&&$\leftrightarrow\alpha$&&&&&&&&components
\\
\hline
\multirow{2}{*}{\rotatebox{90}{$\Gamma$ MMPs\ \ $\,$}}
&30-108
&8
&-
&-
&-
&108
&8
&14.07
&no
&?
\\
&232-108
&101
&59
&107
&101
&108
&8
&58
&no
&?
\\
&24-71
&5
&-
&-
&-
&70
&7
&7.12
&no
&?
\\
&152-71
&64
&41
&70
&64
&71
&7
&38
&yes
&? \\
\hline
\end{tabular}
\label{T:gamma}
\end{table}
\renewcommand{\arraystretch}{1}

While deriving their fractional independence number inequality
(Eq.~(\ref{eq:alpha-alpha})) for their $\Gamma$ graphs,
Howard, Wallman, Veitech, and Emerson \cite{magic-14} refer
to the paper of Cabello, Severini and Winter
\cite{cabello-severini-winter-14} who claim that the GLS
inequality is a noncontextuality inequality. In contrast,
Theorem \ref{th:alpha-star} shows that for quantum YES-NO
measurements carried out on MMP hypergraphs for which the raw data
statistics \ref{hyp-stat-a}(a) is formed, the $\alpha^*$-inequality
(\ref{eq:alpha-cabelllo}) should reduce to the
$\alpha^*_r$-inequality (\ref{eq:alpha-alpha}) and therefore
cannot be considered a noncontextuality inequality since arbitrarily
many contextual and noncontextual MMP hypergraphs violate it, as
exemplified in Figs.~\ref{fig:alpha-star}(a-e), \ref{fig:for17}(a,h),
\ref{fig:alpha-3d}(d), and \ref{fig:22}(b).

On the other hand, in order to deal with the 30-108 MMP hypergraph
we first have to implement 232-108, measure all 232 vertices within
their hyperedges and only then postselect 30 vertices to form
30-108 MMP hypergraph and prove the contextuality. But the 232-108
KS MMP has far too intricate a coordinatization for an
implementation. We tried to generate it from simple vector
components but did not get anything within months of running our
programs on a supercomputer.

Besides, there are practically arbitrary many simpler non-binary
contextual MMP hypergraphs that can be automatically generated
and whose contextuality can be automatically verified via existing
algorithms and programs which then satisfy the inequalities
(\ref{eq:i-ineq}) or (\ref{eq:e-ineq}). Here, one should only
answer the question of what such an inequality for an MMP
hypergraph offers to a quantum computer once it has already been
verified that it is contextual.

\subsection{\label{subsec:mermin}Peres-Mermin
  non-binary MMP hypergraphs  and the smallest MMP
  hypergraph that exists revisited}

In Sec.~\ref{sec:op-ineqal}, Eq.~(\ref{eq:a-cab2}), we
referred to an operator-based inequality for the 4-dim
KS 18-9 MMP hypergraph and in Sec.~\ref{sec:structure},
Eqs.~(\ref{eq:O-I})-(\ref{eq:O-Ic}), we consider an analogous
operator-based inequality for a general critical MMP hypergraph
and for the 21-11 MMP, in particular. In both cases the operators
are defined via vectors/states/vertices of a given MMP hypergraph.
In contrast, the so-called Peres-Mermin square is defined via
operators alone, i.e., without a vector-defined set underlying
the operator set. The operator set is defined by means of the
following nine operators \cite{mermin90}:
\begin{eqnarray}\label{eq:mermin}
\Sigma_{1}=\sigma_z^{(1)}\!\otimes\! I^{(2)},\
\Sigma_{2}=I^{(1)}\!\otimes\!\sigma_z^{(2)},\
\Sigma_{3}=\sigma_z^{(1)}\!\otimes\!\sigma_z^{(2)},&  \notag\\
\Sigma_{4}=I^{(1)}\!\otimes\!\sigma_x^{(2)},\
\Sigma_{5}=\sigma_x^{(1)}\!\otimes\! I^{(2)},\
\Sigma_{6}=\sigma_x^{(1)}\!\otimes\!\sigma_x^{(2)},& \notag\\
\Sigma_{7}=\sigma_z^{(1)}\!\otimes\!\sigma_x^{(2)},\
\Sigma_{8}=\sigma_x^{(1)}\!\otimes\!\sigma_z^{(2)},\
\Sigma_{9}=\sigma_y^{(1)}\!\otimes\!\sigma_y^{(2)}.&\qquad
\label{eq:sigma-m}
\end{eqnarray}

The Peres-Mermin square schematic is shown in the
figure (a) below. The square has 9 dots and 6 lines
and it is claimed that the Peres-Mermin square which is
``convert[ible] \dots\ to KS vectors''
\cite[p.~8]{budroni-cabello-rmp-22}
``exhibits SIC [state independent contextuality]''
\cite{asadian-15}.

Why, then, do Cabello, Kleinmann and Portillo claim
that ``according to quantum theory, no SIC set with less
than 13 rays exists'' \cite{cabello-klein-port-prl-14}? 

Do they have some particular vector-based set which one can
derive from the operator-based Peres-Mermin square in mind?

Because, it might be argued that the Peres-Mermin square is not a
SIC set and even not a consistent contextual set in the following
sense. Eqs.~(\ref{eq:square-S}) and (\ref{eq:sigma-m}) show that
there are three operators in each row and/or column which multiply
so as to give $\pm I$. But there is no common eigenstate or a
combination of eigenstates $|\psi\rangle$ of operators $\Sigma_j$
which would counterfactually enable
$\Sigma_j|\psi\rangle=\pm|\psi\rangle$. Hence, their classical
counterparts $S_j$ (Eq.~(\ref{eq:square-s})) cannot be assigned
values $\pm 1$, either. In other words, since such counterfactually
assumed clicks of nondestructive measurements carried out via
each of operators $\Sigma_j$ cannot occur, assignments of $\pm 1$
to classical counterparts of these assumed measurements is
ungrounded.

This statement is at odds with the overwhelming acceptance and
acclaim of the Peres-Mermin square as a contextual set and a KS
proof in the literature. Let us dissect it.

\bigskip

\begin{resulttttt}\label{p-m-contr}\  
\begin{enumerate}[label=(\roman*)] 
\item We have 
  \begin{eqnarray}\label{eq:p-m-c}
    \Sigma_1\Sigma_2\Sigma_3=\Sigma_4\Sigma_5\Sigma_6
    =\Sigma_7\Sigma_8\Sigma_9=\Sigma_1\Sigma_4\Sigma_7
    =\Sigma_2\Sigma_5\Sigma_8=I,
       \qquad\Sigma_3\Sigma_6\Sigma_9=-I; 
   \end{eqnarray}
\item all operator/projector-based or hypergraph-based contextual
  or {\rm KS} sets assume {\rm YES-NO} measurements (counterfactual
  or actual) carried out on systems emerging from prepared
  gates; upon leaving a gate determined by operators/projectors
  or vectors/ver\-tices the systems are projected to or detected
  in a particular state or not; so, a quantum system in state
  $|\Psi\rangle$ which enters sequences of three gates, whose
  actions are described by operators $\Sigma_j$ from
  Eq.~(\ref{eq:sigma-m}), should either counterfactually or actually
  emerge from each of the gates either in the state $|\Psi\rangle$
  or in the state $-|\Psi\rangle$;
\item  a classical set of states of a classical system which
  would be a counterpart of the quantum system described in (ii)
  should experience predetermined actions of classical gates
  described by observables $S_j$ which would assign either
  `$1$' or `$-1$' to each state;
\item there is no state $|\Psi\rangle$ for which we would
       have
  \begin{eqnarray}\label{eq:p-m-cc}
    \Sigma_i|\Psi\rangle=|\Psi\rangle\qquad
    {\rm and} \qquad  \Sigma_j|\Psi\rangle=-|\Psi\rangle\qquad
  \end{eqnarray}
for $i,j\in \{1,\dots,9\};\ i\ne j$;  
\item statements (ii) and (iv) contradict each other, so the
  statement (iii) cannot hold either. 
\end{enumerate}
\end{resulttttt}

To be more specific, let us consider the following case.
It is generally assumed that the Peres-Mermin contradiction
is state independent and Eq.~(\ref{eq:p-m-c}) is offered
as a support for the claim. Our point is that a quantum
system in state $|\Psi\rangle$ has to pass three gates in a
succession; say first through $\Sigma_1$, then through
$\Sigma_2$, and finally through $\Sigma_3$ so as to emerge
in states $\pm|\Psi\rangle$. But to which counterfactual (not
performed) quantum measurements (``clicks'') assumed valuations
of classical counterparts $S_i$ might correspond? According to
point $(iv)$ eigenvalues of $\Sigma_i$ cannot play such a role.
For an arbitrary state, say the triplet
$|\Psi^+\rangle=|\uparrow\downarrow\rangle+
|\downarrow\uparrow\rangle$, we obtain:
$\Sigma_1\Psi^+\rangle
=|\uparrow\downarrow\rangle-|\downarrow\uparrow\rangle$ and
$\Sigma_2|\Psi^+\rangle
=|\uparrow\uparrow\rangle+|\downarrow\downarrow\rangle$
The observables change the triplet state into
other states (singlet and another triplet). So, the operators
do not act on the same state and this is the meaning of
point $(ii)$ above. 

Our conclusion is that it is inconsistent to assume the existence
of a classical observable $S_j$ which would assign $\pm 1$ to the
states of a system because there is no quantum state $|\Psi\rangle$
of a system which $\Sigma_j$ would project to states
$\pm|\Psi\rangle$. Since the noncontextuality cannot be formulated
we cannot talk about Peres-Mermin square contextuality either.

On the other hand, Budroni, Cabello, G{\"u}hne, Kleinmann and
Larsson claim \cite[p.~8]{budroni-cabello-rmp-22}: 
``The [Peres-Mermin] magic square can be converted into a standard
proof of the KS theorem with vectors [from] \cite{peres}.''
This is, however, incorrect. The Peres-Mermin square cannot be
converted to the vectors given in \cite[Table 2]{peres}. When
properly organized in hyperedges (not explicitly carried out in
\cite{peres}) those hyperedges yield one of 1,233 KS MMP hypergraphs
\cite{pmm-2-09} contained in the master MMP hypergraph 24-24.
The latter set is obtained by adding further hyperedges
(not provided in \cite{peres}) to the former ones
\cite{mpglasgow04-arXiv-0} ``although [Peres] has most probably
never tried to identify all 24 tetrads for his 24 vectors.''
But there exists no conversion of the Peres-Mermin square given
on p.~L179 of \cite{peres} to or from any of 1,233 KS MMP
hypergraphs generated by vectors from Table 2 on p.~L177 of
\cite{peres}. They are simply independently given in the same
paper.

Let us therefore see whether we can arrive at vectors for which
some plausible linkage with the Peres-Mermin operators would be
possible. The idea is to establish a correspondence with the Pauli
operators $\Sigma$ structure, as a contextual ``square of
orthogonalities'' which would support a postselection of 9
measurements in a $3\times 3$ arrangement of YES-NO measurements
shown in Fig.~\ref{fig:p-m}:

\begin{figure}[ht]
  \begin{center}
  \includegraphics[width=1\textwidth]{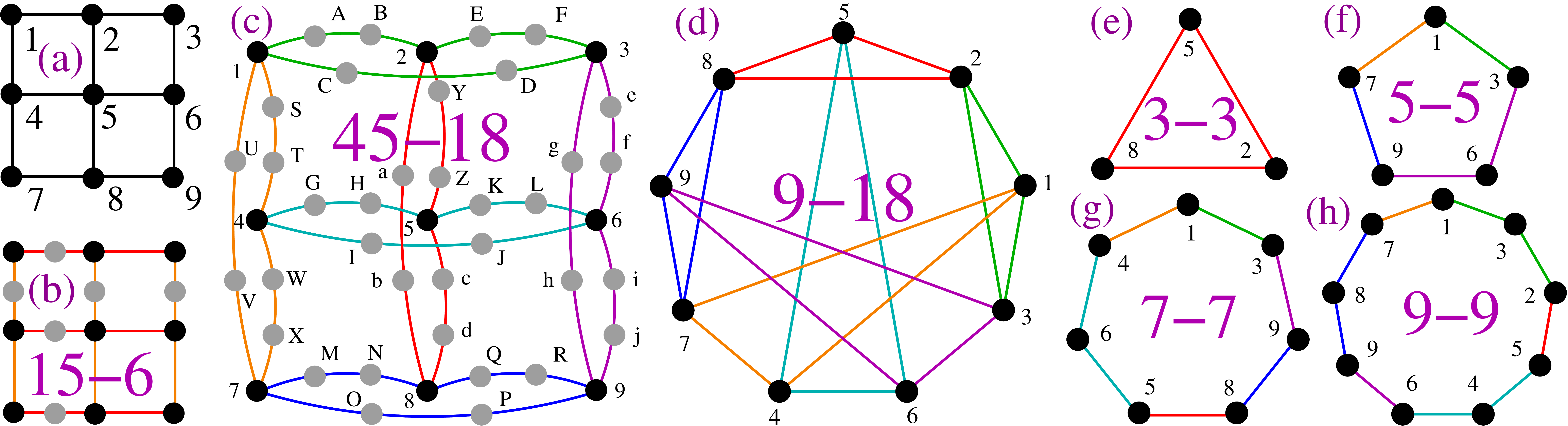}
\end{center}
\caption{(a) Peres-Mermin operator schematics;
  (b,c) filled MMPs (noncontextual);
  (d) filled MMP with extended orthogonalities and gray $m=1$
  vertices (noncontextual); (e) MMP with extended orthogonalities
  and $m=1$ vertices dropped---contextual, but not a KS set
  and not a critical set; (f)-(i) critical subsets of (e).}
\label{fig:p-m}
\end{figure}

\begin{itemize}
\item (a) a direct translation
  {\tt 12,23,45,56,78,89,14,47,25,58,36,69.} does not work
  (hyperedges connect vertices pairwise and consecutively);
  e.g., {\tt 1} is not orthogonal to {\tt 3} but it should be
  because $\Sigma_i$, $i=1,2,3$ mutually commute; 
\medskip
\item (b) 15-6: {\tt 1A23,4B56,7C89,1D47,2E58,3F69.}~is
  noncontextual; hyperedges go through all three vertices;
  e.g., {\tt 1} is orthogonal to {\tt 2} and {\tt 3}
  in the same hyperedge; 9-6 (15-6 with gray vertices
  {\tt A,B,\dots,F} dropped {\tt 123,456,789,147,258,369.}) is
  also noncontextual;
\medskip
\item (c) has 45 vertices/vectors and 18 hyperedges and its
  string is given in Appendix \ref{app3}:
  hyperedges connect vertices pairwise but exhaustively; e.g.,
  {\tt 1} is orthogonal to {\tt 2} via one hyperedge and to
  {\tt 3} via another); it is not contextual; some of its
  properties are:  $HI_{cm}=9$ and $HI_{cM}=18$ and they violate
  the v-inequality:
  \begin{eqnarray}
  HI_{cM}=18=HI_q=18.
  \label{eq:mer-ieq}
  \end{eqnarray}
  $l_{cm}=l_{cM}=18$ and they violate the e-inequalities:
  \begin{eqnarray}
  l_{cM}=18=l=18
  \label{eq:mer-eeq}
  \end{eqnarray}
\medskip
\item (d) dropping all gray vertices with $m=1$ from (c) yields a
  contextual 9-18 non-binary MMP hypergraph
{\tt 12,23,13,45,56,46,78,89,79,14,47,17,25,58,28,36,69,39.}\break
  All the remaining vertices have $m=4$; $HI_{cM}=3$.
  The $\alpha_r^*$-inequality reads 
  \begin{eqnarray}
  3=HI_{cM}=\alpha<\alpha_r^*=\frac{9}{2}=4.5.
  \label{eq:mer-ineq-1}
  \end{eqnarray}
  The e$_{Max}$-inequality reads:
  \begin{eqnarray}
  l_{cM}=12<l=18,
  \label{eq:mer-eneq}
  \end{eqnarray}
  and its span shows us that it is not critical.
  It might be implemented by port-detections at each hyperedge/gate,
  but not via letting systems fly through triple consecutive gates
  as in Peres-Mermin square measurements in the literature.
\end{itemize}

\bigskip

The non-binary 9-18 MMP hypergraph contains the critical
subhypergraphs 3-3, 5-5, 7-7, and 9-9 for whose we have
$l_{cM}=l-1$. These criticals are shown in Fig.~\ref{fig:p-m}(e-h).
They all have $l_{cM}=(l-1)/2$ and satisfy the e-inequalities.
Their filled versions 9-3, 15-5, 21-7, and 27-9 all have
$l_{cM}=l$ and therefore they violate the e-inequalities. 

We can get the coordinatization of these critical sets by
reading it off from Appendix \ref{app3} for the corresponding
vertices. 

The obtained smallest quantum contextual MMP hypergraph 3-3,
which differs from the one obtained in Sec.~\ref{subsec:4d}
(Def.~\ref{smallest}) only by its coordinatization, is not a
SIC set in the standard sense of the word, because it is not an
operator-based set. It is state independent in the sense 
that its contextuality is based on its hypergraph structure,
meaning that it holds for any set of states that can support it,
i.e., build its coordinatization. 

\bigskip\bigskip\bigskip
\subsection{\label{subsec:5d}5-dim KS MMP hypergraphs/sets}

The 5-dim hypergraph spaces are defined by spin-2 systems and cannot
include qubits. So far, to our knowledge, only several MMP
hypergraphs were obtained in
\cite{cabell-est-05},\cite[Supp.~Material]{waeg-aravind-pra-17}
``by hand.'' By our automated generation we obtained up to 30
millions critical non-isomorphic MMP hypergraphs from the vector
components $\{0,\pm 1\}$. Several smallest ones are shown in
Fig.~\ref{fig:5dim}.

\begin{figure}[h]
  \flushleft
  \includegraphics[width=1\textwidth]{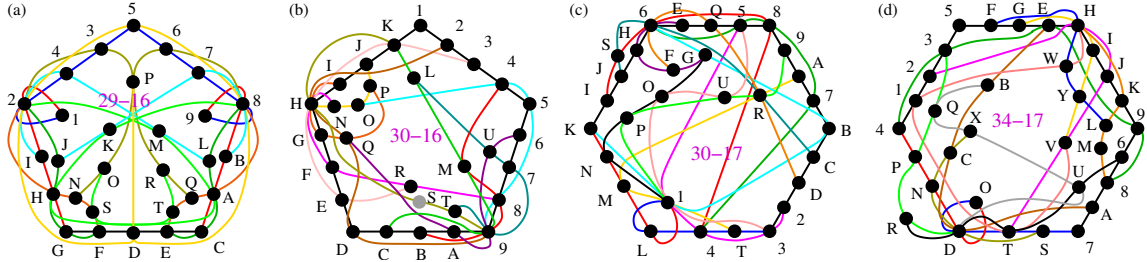}
  \caption{5-dim MMP hypergraphs; (a) the smallest 29-16 exhibits a
    left-right symmetry; its biggest loop is a pentagon; (b) one of
    several 30-16 MMPs; (c) one of several 30-17 MMPs; half of them
    have hexagons as their biggest loops; (d) one of several 34-17
    MMPs with hexagon loops.}
\label{fig:5dim}
\end{figure}

Their structural property parameters are shown in Table
\ref{T:5da}. 

\renewcommand{\arraystretch}{1.2}
\begin{table}[ht]
\caption{Structural properties of KS subhypergraphs of the 5-dim KS
  master set 105-136 generated from $\{0,\pm 1\}$ components.  All
  MMPs violate the $\alpha^*_r$-inequality.}
\center
\setlength{\tabcolsep}{5pt}
\begin{tabular}{@{}l*{7}{c}@{}} 
\hline
Dim&KS MMPs&$HI_{cM}$&$HI_{cm}$&$l_{cM}$&$l_{cm}$&crit&vec.$\,$comp.
\\
\hline
\multirow{2}{*}{\rotatebox{90}{\minibox{Real 5-dim MMP \ \ \ \ \ \  \\hypergraphs}}}
 &29-16 
&7
&3
&15
&10
&yes
&$\{0,\pm 1\}$
\\
&30-16 
&8
&3
&15
&11
&yes
&$\{0,\pm 1\}$
\\
&30-17 
&8
&3
&16
&10
&yes
&$\{0,\pm 1\}$
\\
&34-17 
&8
&3
&16
&11
&yes
&$\{0,\pm 1\}$
\\
&58-40 
&15
&7
&39
&22
&yes
&$\{0,\pm 1\}$
\\
&65-40 
&17
&8
&39
&25
&yes
&$\{0,\pm 1\}$\\
&\multirow{2}{6em}{\hfil 105-136 \\ \quad(master)}
&\multirow{2}{2em}{\hfil 23}
&\multirow{2}{2em}{\hfil 13}
&\multirow{2}{2em}{\hfil 122}
&\multirow{2}{2em}{\hfil 80}
&\multirow{2}{2em}{\hfil no}
&\multirow{2}{4em}{\hfil $\{0,\pm 1\}$}
\\ \\
\hline
\end{tabular}
\label{T:5da}
\end{table}
\renewcommand{\arraystretch}{1}

In Appendix \ref{app5string} we give the ASCII strings and
coordinatization for all MMP hypergraphs from Table \ref{T:5da}
which include the four hypergraphs given in Fig.~\ref{fig:5dim}.

\subsection{\label{subsec:6d}6-dim KS MMP hypergraphs/sets}

In 2014 a star-like 6D 21-7 KS set was found \cite{lisonek-14}
and implemented \cite{canas-cabello-14}. The coordinatization
used was defined by the vector components from the set
$\{0,1,\omega,\omega^2\}$. In \cite{pavicic-pra-17} it
was shown that the set can be given a simpler coordinatization
based only on the components from $\{0,1,\omega\}$ components and
that the star-like graphical representation is isomorphic to the
triangular representation given in \cite[Fig.~11]{pavicic-pra-17}.
In the same reference, a polytope-based class 236-1216 of 6-dim KS
hypergraph was generated but it did not contain the 21-7
star/triangle-like set; its vectors had components from the
following set:
$\{0,\pm\frac{1}{2},\pm\frac{1}{\sqrt{3}},\pm\frac{1}{\sqrt{2}},1\}$.
Also, based on many other failed attempts to generate real
coordinatization of 21-7, we conjecture that it might have only
complex ones.

\begin{figure}[ht]
  \begin{center}
  \includegraphics[width=1\textwidth]{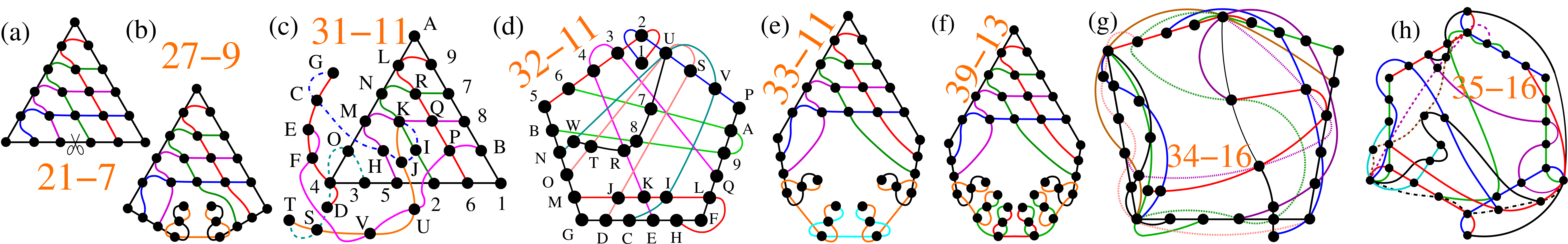}
\end{center}
\caption{(a-f) the smallest critical 6-dim KS MMP hypergraphs
  obtained from $\{0,1,\omega,\omega^2\}$ components; see text;
  (g,h) smallest critical 6-dim KS MMP hypergraphs obtained
  from $\{0,\pm1\}$ components; see text.}
\label{fig:6d-triangl}
\end{figure}

\renewcommand{\arraystretch}{1}
\begin{table}[h]
\caption{Structural properties of the KS subhypergraphs of the
  6-dim KS master set/hypergraph 81-162. Only
  20-5 and 31-11 violate the $\alpha_r^*$-inequality;
  $5>3.\dot3$ and $6>5.17$, respectively.}
\center
\setlength{\tabcolsep}{7pt}
\begin{tabular}{@{}l*{7}{c}@{}} 
\hline
Dim&KS MMPs&$HI_{cM}$&$HI_{cm}$&$l_{cM}$&$l_{cm}$&crit.&vec.$\,$comp.
\\
\hline
\multirow{2}{*}{\rotatebox{90}{\minibox{Complex 6-dim MMP\\hypergraphs}\quad}}
&20-5 Fig.~\ref{fig:pent-cab-6d}(d)  
&5
&4
&5
&5
&$\!\!\!\!\!$non-KS\!$\!\!$
&$\{0,\pm 1,2\}$
\\
&21-7 \cite{lisonek-14} 
&3
&3
&6
&6
&yes
&(0,1,$\omega$)
\\
&27-9 \cite{pm-entropy18,pwma-19}
&4
&4
&8
&8
&yes
&(0,1,$\omega$)
\\
&31-11 \cite{pm-entropy18,pwma-19}
&6
&4
&10
&9
&yes
&(0,1,$\omega$,$\omega^2$)
\\
&32-11 \cite{pwma-19}
&5
&4
&10
&8
&yes
&(0,1,$\omega$,$\omega^2$)
\\
&33-11 \cite{pwma-19}
&5
&3
&10
&8
&yes
&(0,1,$\omega$)
\\
&36-13 \cite{pwma-19}
&6
&4
&12
&10
&yes
&(0,1,$\omega$,$\omega^2$)
\\
&39-13 \cite{pwma-19}
&6
&4
&12
&8
&yes
&(0,1,$\omega$,$\omega^2$)
\\
&\multirow{2}{6em}{\hfil 81-162 \cite{pwma-19} (sub-master)}
&\multirow{2}{2em}{\hfil 11\hfil}
&\multirow{2}{2em}{\hfil 7\hfil}
&\multirow{2}{2em}{\hfil 132\hfil}
&\multirow{2}{2em}{\hfil 84\hfil}
&\multirow{2}{2em}{\hfil no\hfil}
&\multirow{2}{4.5em}{(0,1,$\omega$,$\omega^2$)}
\\ \\
\hline
\end{tabular}
\label{T:6dc}
\end{table}
\renewcommand{\arraystretch}{1}

\renewcommand{\arraystretch}{1}
\begin{table}[h]
\caption{Structural properties of KS subhypergraphs of the 6-dim KS
  master set 236-1216.  All the MMPs violate the
  $\alpha^*_r$-inequality.}
\center
\setlength{\tabcolsep}{7pt}
\begin{tabular}{@{}l*{7}{c}@{}} 
\hline
Dim&KS MMPs&$HI_{cM}$&$HI_{cm}$&$l_{cM}$&$l_{cm}$&crit&vec.$\,$comp.
\\
\hline
\multirow{2}{*}{\rotatebox{90}{\minibox{Real 6-dim MMP\ \ \ \\hypergraphs}}}
 &34-16 \cite{pavicic-pra-17}
&7
&3
&15
&10
&yes
&$\{0,\pm 1\}$
\\
&35-16 \cite{pavicic-pra-17}
&7
&3
&15
&10
&yes
&$\{0,\pm 1\}$
\\
&37-16 \cite{pavicic-pra-17}
&7
&3
&15
&10
&yes
&$\{0,\pm 1\}$
\\
&37-17 \cite{pavicic-pra-17}
&8
&3
&16
&11
&yes
&$\{0,\pm 1\}$
\\
&37-18 \cite{pavicic-pra-17}
&8
&3
&17
&11
&yes
&$\{0,\pm 1\}$
\\
&38-18 \cite{pavicic-pra-17}
&8
&4
&17
&10
&yes
&$\{0,\pm 1\}$
\\
&\multirow{2}{6em}{\hfil 236-1216 \cite{pavicic-pra-17} (sub-master)}
&\multirow{2}{2em}{\hfil 66}
&\multirow{2}{2em}{\hfil 36}
&\multirow{2}{2em}{\hfil 105}
&\multirow{2}{2em}{\hfil 78}
&\multirow{2}{2em}{\hfil no}
&\multirow{2}{4em}{\hfil $\{0,\pm 1\}$}
\\ \\
\hline
\end{tabular}
\label{T:6da}
\end{table}
\renewcommand{\arraystretch}{1}

In \cite{pm-entropy18} we obtain a master KS set 216-153 from
$\{0,1,\omega\}$ components. It contains just three critical sets
21-7, 27-9, and 33-11 (the last one has 8
non-isomorphic instances); Fig.~\ref{fig:6d-triangl}(a,b,e).
The 21-7 and 33-11 strings are given in
\cite[Supp.~Material]{pwma-19}. In \cite{pwma-19} we also generate
two master sets 591-1123 and 81-162 and the corresponding classes
(25 million non-isomorphic critical KS sets) from
$\{0,1,\omega,\omega^2\}$ components. The 31-11 string is given in
\cite[Supp.~Material]{pwma-19}. See Table \ref{T:6dc}.

In this paper, we generate the 332-1408 master from $\{0,\pm1\}$
components. It contains two unconnected sub-masters: a non-KS 96-192
one and a 236-1216 KS one, which turns out to be isomorphic with
the aforementioned polytope based 236-1216 one. 

The structure of the smaller 6-dim MMP hypergraphs from the
81-162 class is different from the 4-dim ones because the
$\delta$ feature allows the 6-dim hypergraphs to be more
interwoven than the 4-dim ones, as presented in
\cite[Fig.~A2]{pm-entropy18} and \cite[Fig.~4]{pwma-19} and
because of the coordinatization with complex vectors. As a
consequence, fewer vertices need to be assigned value 1 to
satisfy the KS conditions (i) and (ii) of the KS theorem.

The KS MMP hypergraphs from the 236-1216 class \cite{pavicic-pra-17}
obtained by means of real vector components have much larger
smallest hypergraphs, as shown in Table \ref{T:6da}. The
$\{0,\pm 1\}$ components yield the 332-1408 MMP master
which consists of two unconnected sub-masters: the KS 236-1216
and a non-KS (noncontextual) 96-192.

Both MMP classes, 81-162 and 236-1216 exhibit the $\delta$ feature
and the distinguishers which determine the sizes of minimal MMPs
are complex vs.~real vectors.

\subsection{\label{subsec:78d}7- and 8-dim KS MMP
  hypergraphs/sets}

Here we present a few examples from the 7- and 8-dim spaces.
The distribution of the 7-dim KS MMP class is provided
in \cite{pavicic-pra-22} and of the 8-dim one in
\cite{pavicic-pra-17}.

\begin{figure}[ht]
  \begin{center}
  \includegraphics[width=1\textwidth]{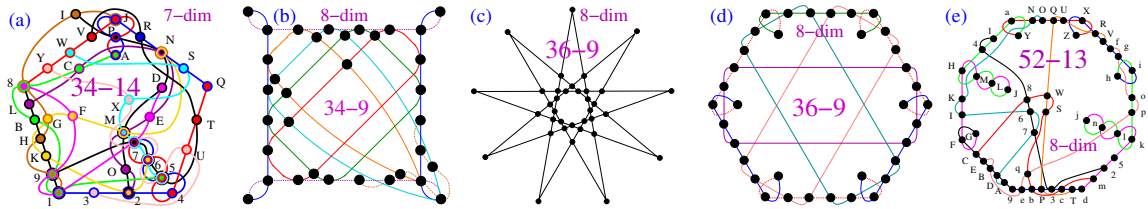}
\end{center}
\caption{(a) The smallest critical 7-dim KS MMP hypergraphs we
  obtained from the 805-9936 master generated by
  $\{0,\pm 1\}$ components in \cite[Fig.~3]{pavicic-pra-22};
  (b-d) smallest critical 8-dim KS MMP hypergraphs obtained we
  obtained from the 3280-1361376 master (more precisely its
  2768-1346016 sub-master) generated by $\{0,\pm1\}$
  components; (c) and (d) are isomorphic---they also have a
  triangular representation (given in \cite[Fig.~14]{pavicic-pra-17})
  as the 6-dim 21-7 in Fig.~\ref{fig:6d-triangl} does; as for (d),
  cf.~4-dim Fig.~\ref{fig:graph}(c); (e) the smallest MMP with
  52 vertices; cf. 52-16 in \cite[Fig.~14]{pavicic-pra-17}.}
  \label{fig:78d}
\end{figure}

We give the structural properties of the examples in Table
\ref{T:78d} and their ASCII strings and coordinatizations in
Appendix \ref{app78string}. The 8-dim MMP hypergraphs given
in Fig.~\ref{fig:78d} are highly symmetrical, so the reader
might easily assign ASCII symbols to vertices in the figure.

\begin{table}[ht]
  \caption{Structural properties of KS MMP subhypergraphs
  of the 7-dim KS master set 805-9936 and 8-dim KS (sub-)master 
  2768-1346016, both obtained from $\{0,\pm 1\}$ components. The
  7-dim 34-14 and 202-97 and 8-dim 37-11 and 52-14 violate the
  $\alpha^*_r$-inequality, while 8-dim 34-9 and 36-9 satisfy it.}
  \center
\setlength{\tabcolsep}{8pt}
\begin{tabular}{|c|ccccccc|}
    \hline
    & 
  &                       
  & 
  & 
  & 
  &
  & \multirow{2}{*}{Vector} \\
dim
    &  KS MMPHs
    & $HI_{cM}$
  & $HI_{cm}$
  & $l_{cM}$
  & $l_{cm}$
   & critical
   &  \multirow{2}*{components} \\
 & &  &  &  &  &  &  \\
  \hline
  \multirow{2}{*}{7-dim} & 34-14 & 7 & 3 & 13 & 8 & yes & $\{0,\pm 1\}$ \\
  & 202-97 & 44 & 21 & 94 & 68 & yes & $\{0,\pm 1\}$ \\
  \hline
\multirow{6}{*}{8-dim}
    & 34-9  & 4 & 3  &  8 &  8 & yes & $\{0,\pm 1\}$\\
    & 36-9  & 4 & 4  &  8 &  8 & yes & $\{0,\pm 1\}$\\
    & 37-11 & 5 & 3  & 11 &  8 & yes & $\{0,\pm 1\}$ \\
    & 52-13 & 8 & 6  & 12 & 10 & yes & $\{0,\pm 1\}$ \\
    & \multirow{2}{6em}{\hfil 120-2024 \ (sub-master)}
  & \multirow{2}{2em}{\hfil 8}
  & \multirow{2}{2em}{\hfil 3}
  & \multirow{2}{2em}{\hfil 1080}
  & \multirow{2}{2em}{\hfil 404}
  & \multirow{2}{2em}{\hfil no}
  & \multirow{2}{3.6em}{\hfil $\{0,\pm 1\}$} \\
 & &  &  &  &  &  &  \\
  \hline
  \end{tabular}
\label{T:78d}
\end{table}

In \cite{pavicic-pra-17} we obtained the 36-8(c) by hand and
34-9 and 36-9(d) from the master 120-2024 generated by
the Lie algebra E8. Coordinatizations of 34-9 and 36-9(d) in
\cite{pavicic-pra-17} therefore differ from the one obtained
here via our program {\textsc{Vecfind}}, i.e., from $\{0,\pm 1\}$
components. Since any 8-dim MMP hypergraphs with the latter
vector components is a subhypergraph of the 8-dim KS master
MMP hypergraph 3280-1361376 generated from $\{0,\pm 1\}$ components,
the 34-9 and 36-9 are subhypergraphs of both the 120-2024 and 
the 3280-1361376 master sets and have coordinatizations from
both $\{0,\pm 1\}$ and Lie E8 components. Actually, the 120-2024
itself, as proven in \cite{pm-entropy18}, can have both
coordinatizations and is therefore a sub-hypergraph/sub-master of
the 3280-1361376 master. The latter master consists of two larger
unconnected sub-masters: a KS 2768-1346016 one and a binary
(noncontextual) 512-15360 one. This means that all 6,925,540
MMP critical hypergraphs \cite[Fig.~12]{pavicic-pra-17} obtained
from the Lie E8 components of the 120-2024 master are also
subhypergraphs of the 2768-1346016 master with $\{0,\pm 1\}$-based
coordinatizations.

In Table \ref{T:78d} we obtain $l_{cM}=94$ instead of expected
96 (for a critical MMP hypergraph). This might be due to the
algorithm and program \textsc{One} imperfections: for such highly
interwoven hyperedges some assignments of 1s to vertices might
have failed. Recall a similar outcome for the original
Kochen-Specker MMP hypergraph discussed below Table \ref{T:3d}.
An independent algorithm and program which would check on these
discrepancies is needed.

\section{\label{sec:impl}Implementation}

The experiments cited in Sec.~\ref{sec:intro} were all focused
on proving that contextual sets really are contextual. In
particular, they mostly carried out repeated measurements with
operators acting on different states to prove their
state-independence (SIC). However, we entered the realm of
generation of arbitrary many contextual sets of any structure
in any dimension via automated algorithms and programs and
the next stage of their application should be a direct
implementation of the sets from the data base we obtained.
Instead of proving that contextual sets really are contextual,
we should simply accept that they are and start using them in
quantum computation and quantum communication.

For instance, generalised Stern-Gerlach experiment
which makes use of both magnetic and electric fields
\cite{anti-shimony} can generate any quantum state in any
dimension and therefore provide us with implementation of any
hypergraph which would yield these states.

On the other hand, there is a universal two-qubit quantum gate
by means of which one can build any quantum network/gates in a
$2^n$-dim space that can be realized by QED, nuclear spins,
quantum dots, trapped ions, or photon-photon coupling in an
all-optical realization \cite{weinf95}.

In a projector formulation, all KS MMP hypergraph with a
coordinatization are state-independent and for non-KS
non-binary MMP hypergraphs one still has to find a general
automated approach. However, in addition to numerous already
known small non-KS non-binary MMP hypergraphs
\cite{pavicic-entropy-19} we do have an abundance of KS MMP
hypergraphs of any size and structure in odd
\cite{pavicic-pra-22} and even
\cite{pavicic-pra-17,pm-entropy18,pwma-19} dimensional spaces.

Taken together, future research in the field should be focused
on finding general automated algorithms and programs for
implementation of arbitrary contextual sets in quantum gate
networks. 

\section{\label{sec:disc}Discussion}

In this paper we elaborate on particular approaches to features of
quantum contextual sets that determine their generation, usage,
applications, implementations, and perspectives of future
research.

\subsection{\label{subsec:disc1}Operator-based
  vs.~MMP-hypergraph-based contextuality}

In Secs.~\ref{sec:h-lang}, \ref{sec:op-ineqal}, and
\ref{sec:structure} we compare operator-based and
MMP-hypergraph-based approaches to contextual sets and show that the
former one relies on the latter. In the literature only a handful
of smallest operator-based contextual set have been analyzed while
there are billions of contextual MMP hypergraphs
\cite{bdm-ndm-mp-fresl-jmp-10,pmm-2-09,mp-7oa,pavicic-pra-17,pm-entropy18,pavicic-entropy-19,pwma-19}.

MMP hypergraph language is introduced in Sec.~\ref{sec:h-lang}
and contrasted with obsolete and/or inappropriate graph and
general hypergraph language throughout the paper. The latter
approaches are vividly graphically presented in
Figs.~\ref{fig:st-hyp}, \ref{fig:pentagon}, \ref{fig:graph},
\ref{fig:magic}, and \ref{fig:p-m} and their
disadvantages discussed in the text surrounding them. 

In Sec.~\ref{sec:h-ext} we consider several extensions of contextual
Kochen-Specker (KS) vector sets. The most important is a non-binary
contextual MMP hypergraph extension given by Def.~\ref{def:n-b} in
which we dispense with vectors (coordinatization), i.e., states, and
rely only on the very structure of MMP hypergraphs.
In this sense they are state-independent. To make use of
states/vectors or to define operators we attach a compatible
coordinatization to them using simple vector components.

In Sec.~\ref{sec:op-ineqal} we consider three approaches to
obtain operator-based contextual sets, two of which
(hyperedge (i) and vertex (ii) ones) generate operators
directly from MMP hypergraphs $k$-$l$, mainly via projectors
$P=|v_i\rangle\langle v_i|$, $i=1,\dots,k$, where $k$ is the number
of vertices and $l$ the number of hyperedges. We conjecture that
the following rule universally holds. 

\begin{resulttt}\label{r:o-m-r} Every {\rm MMP} hypergraph which
  might serve for a construction of an operator-based contextual
  set via its states/vectors is itself a non-binary contextual
  {\rm MMP} hypergraph.   
\end{resulttt}

We give a number of examples to this rule in
Secs.~\ref{sec:op-ineqal}, \ref{sec:hyp-op}, and
\ref{sec:structure}. We would like to single out Yu-Oh's 13-16 set
shown in Fig.~\ref{fig:yu-oh-mmp}; see the text above it.

An important notion we introduce in order to generate MMP
hypergraphs is the multiplicity of vertices, $m$, given by
Def.~\ref{def:qh-mi} that tells us how many hyperedges each vertex
shares. In relation to it, we can induce/generate contextual MMP
hypergraphs from both non-binary (contextual) and binary
(noncontextual) MMP hypergraphs in two ways:
\begin{itemize}
\item by dropping data obtained by measuring states of systems
  related to vertices with $m=1$ as well as vertices themselves
  from the MMP hypergraphs; this works for non-binary MMPs
  hypergraphs (e.g., for 3-dim ones shown in fig.~\ref{fig:3dup}
  and extensively elaborated on in \cite{pavicic-entropy-19}) as
  well as for the binary ones (e.g., Peres-Mermin's 45-18 MMP
  shown in Fig.~\ref{fig:p-m}); thousands of such
  ${\rm \overline{subhypergraphs}}$ which serve the purpose
  are generated in \cite[Sec.~II.D]{pavicic-entropy-19};
\item by finding smaller contextual non-binary MMP
  ${\rm \overline{subhypergraphs}}$ (Def.~\ref{def:over-sub})
  contained in binary MMP hypergraphs with $m\ne 1$, as carried
  out on two examples in \cite{cabello-severini-winter-14}.  
\end{itemize}

\subsection{\label{subsec:disc2}Raw- and postselected data
  statistics and their inequalities;
  the Gr{\"o}tschel-Lov{\'a}sz-Schrijver (GLS)
  {\boldmath{$\alpha^*$}}-inequality
  is not a noncontextuality inequality}

The role of the multiplicity $m$ (Def.~\ref{def:qh-mi}) with which
vertices share hyperedges in every hypergraph is characterized by
Eq.~(\ref{eq:theorem}) and Lemma \ref{th:theorem}. As we 
explain in Secs.~\ref{sec:hyp-op} and \ref{sec:structure} it
enables us to distinguish the standard Hypergraph Statistics
based on raw measurement data (\ref{hyp-stat-a}(a)) from the one
based on postprocessed measurement data (\ref{hyp-stat-a}(b)).
It also determines the structure of the hypergraphs as shown in
Sec.~\ref{subsec:4dm}, Fig.~\ref{fig:for17}, and
Table \ref{T:masters}. 

A standard tool for discriminating contextual from noncontextual
sets has lately been claimed to be noncontextuality inequalities
(Def.~\ref{def:non-c-i-v}), in particular the
operator/projector-based ones
\cite{klyachko-08,cabello-08,badz-cabel-09,beng-blan-cab-pla12,kurz-kasz-12,kurz-cabello-14,ram-hor-14,cabell-klein-budr-prl-14,cabello-severini-winter-14,yu-tong-14,xu-chen-su-pla15,yu-tong-15}. We review them
in Secs.~\ref{sec:op-ineqal} and \ref{sec:hyp-op}. They are mostly
defined by states/vectors of contextual MMP hypergraphs what means
that the MMP hypergraph structure together with its coordinatization
serves us to build operator/projector structure.

Vertex multiplicity enables us to introduce a new kind of
hypergraph-based vertex v-inequalities (Def.~\ref{def:iin}) and
relate the operator-based inequalities with the hypergraph-based
(hyper){\bf e}dge e-inequalities
(Defs.~\ref{def:ein}, \ref{def:einm}). 
We also consider the $\alpha^*$-  (Eq.~(\ref{eq:alpha-cabelllo})),
and $\alpha_r^*$- (Eq.~(\ref{eq:alpha-alpha}), and
$\alpha_p^*$- (Eq.~(\ref{eq:alpha-alpha-b})) inequalities. 
The $\alpha$ in them is the maximum number of pairwise
non-adjacent vertices (Def.~\ref{def:alpha}), i.e., the maximum
number of vertices to which one can assign `1'
(Lemma \ref{lemma:alpha}); it is called the independence
number and also the stability number. Table \ref{T:list} provides
us with a list of inequalities.

\begin{itemize}
\item The v-inequality, relates the maximal classical multiplexed
  vertex indices $HI_{cm}$, $HI_{cM}$, $HI^m_{cM}$, Def.~\ref{def:ch-mi}
  (the total number of 1s we can assign to vertices each
  multiplied by its multiplicity or not) to quantum hypergraph
  index $HI_q$, Def.~\ref{def:qh-i} (the sum of the
  probabilities of getting quantum measurement clicks within
  hyperedges:\break $HI_{cm}\le HI_{cM}\le HI^m_{cM}<HI_q$,
  Def.~\ref{def:iin}); see Tables
  \ref{T:3d}, \ref{T:4d1}, \ref{T:6dc}, and \ref{T:6da}.
\item The e$_{Max}$- and e$_{min}$-inequalities quantifies the KS
  theorem generalisation (Def.~\ref{def:n-b}), according to which
  we cannot assign 1 to all hyperedges of a non-binary MMP
  hypergraph, i.e., they simply relate the maximum and minimum
  number, respectively, of hyperedges which can contain 1
  ($l_{cM},l_{cm}$; Def.~\ref{def:lcMm}) with the actual number of
  hyperedges of a considered MMP hypergraph $k$-$l$ containing $k$
  vertices and $l$ hyperedges. The e$_{Max}$-inequality
  (\ref{eq:e-ineq}) corresponds to Badzi{\'a}g, Bengtsson, Cabello,
  and Pitowsky's $\beta$-inequality \cite{badz-cabel-09}, given in
  Eq.~(\ref{eq:b-cab-c}), and $l_{cM}$ corresponds to
  $\beta$; e$_{Max}$-inequality has a trivial form $l_{cM}+1=l$ for
  critical KS MMP hypergraphs (except perhaps for the original KS
  one---see Sec.~\ref{subsec:3d}). They become relevant ($l_{cM}$
  becomes significantly smaller than $l$) for non-critical and
  master KS MMP hypergraphs as shown in Tables \ref{T:4d1},
  \ref{T:6dc}, and \ref{T:6da}. The e$_{min}$-inequality
  (\ref{eq:e-ineq-m}) is more suitable for an implementation
  because of its greater span between its terms. We conjecture that
  $l_{cm}$ is the ``rank of contextuality'' recently introduced by
  Horodecki at al.~\cite{horod-22}.
\item The $\alpha_r^*$- (\ref{eq:alpha-alpha}) and
  $\alpha_p^*$- (\ref{eq:alpha-alpha-b}) inequalities, with
  constant/fixed probabilities of detecting (within quantum YES-NO
  measurements) a quantum particle in a particular state 
  (assigned to vertices within each hyperedge), are special cases
  of the GLS  $\alpha^*$-inequality which, in contrast to the
  former ones, rely to variable/free probabilities
  (Defs.~\ref{def:alpha-star}, \ref{def:alpha-star-LP}, and
  Eq.~(\ref{eq:alpha-cabelllo})): $\alpha\le\alpha^*$
  (\ref{eq:alpha-cabelllo}). The latter
  inequality is valid for any graph or hypergraph, contextual or
  not, with variable/free probabilities assigned to vertices
  within each hyperedge. Linear programming or
  any other algorithms for solving linear optimization problems
  then determines which values must the probabilities have within
  each hyperedge (where their sum must be $\le 1$).
  The $\alpha_r^*$-inequality is based on the Raw data statistics
  \ref{hyp-stat-a}(1) and it is not a noncontextuality inequality. 
  The $\alpha_p^*$-inequality is based on the Postprocessed data
  statistics \ref{hyp-stat-a}(2); it is just another another form of
  the v-inequality (\ref{eq:i-ineq}) and is therefore a
  noncontextuality inequality. For example, a spin-1 particle passing
  through a Stern-Gerlach gate/hyperedge has the probability of
  $\frac{1}{3}$ to exit from any of its ports along any of its 3
  vertices/ports. But then, for arbitrary many contextual non-binary
  MMP hypergraphs, the $\alpha_r^*$-inequality fails
  (Cf.~Figs.~\ref{fig:alpha-star}(a-d), \ref{fig:alpha-3d}(d),
  \ref{fig:22}(b)).
\end{itemize}

\renewcommand{\arraystretch}{1.2}
\begin{table}[ht]
\caption{List of the inequalities elaborated on in the paper.}
  \center
\setlength{\tabcolsep}{7pt}
\begin{tabular}{@{}l*{4}{c}@{}} 
\hline
  \multicolumn{2}{c}{\multirow{2}{9em}{\quad Inequality\hfill}}&\multirow{2}{4em}{\quad\ Eq.}&\multirow{2}{7.7em}{Noncontextuality inequality
(\ref{def:non-c-i-s})}
  \\
\\
  \hline
  v
&$HI_{cm}\le HI_{cM}\le HI^m_{cM}<l$
&(\ref{eq:i-ineq})
  &yes
\\  
  e$_{Max}$
&$l_{cM}<l$
&(\ref{eq:e-ineq})
  &yes
\\
e$_{min}$
&$l_{cm}<l$
&(\ref{eq:e-ineq-m})
&yes
\\
$\alpha^*_r$
&$\alpha\le\alpha^*_r=\frac{k}{n}$
&(\ref{eq:alpha-alpha})
&no
\\
$\alpha^*_p$
&$\alpha <\alpha^*_p=l$
&(\ref{eq:alpha-alpha-b})
&yes
\\  
GLS
&$\alpha\le\alpha^*$
&(\ref{eq:alpha-cabelllo})
&no
\\
\hline
\end{tabular}
\label{T:list}
\end{table}
\renewcommand{\arraystretch}{1}

Hypergraph v-, e-, $\alpha_r^*$, and $\alpha_p^*$-inequalities
can be generated via automated procedures directly from
non-binary MMP hypergraphs. The hypergraphs themselves can be
generated in automated ways from simple vector components
such as $\{0,\pm 1,i\}$ (Tables \ref{T:4d1},\ref{T:6da}) or
$\{0,1,\pm\omega\}$ (Table \ref{T:6dc}), etc. The operator
approach is suitable for forming quantum gates which can be
applied to arbitrary states to generate inequalities for
evaluating the contextual measurements, while the
latter automated hypergraph approach is suitable for testing a
level of contextuality of hypergraph states by postprocessing
measurements carried out at out-ports of gates determined by
hyperedges as well as for verifying contextual properties of a
chosen hypergraph.

The v-, e-, and $\alpha_p$ inequalities are the only genuine
noncontextuality inequalities (Def.~\ref{def:non-c-i}).

\subsection{\label{subsec:disc3}Structure and features
  of particular MMP hypergraphs}

The MMP hypergraph language applied to several well-known
contextual sets yields the following results.

\subsubsection{\label{subsub:mult}{\rm MMP} vertex multiplicity}

Throughout the paper we show that the multiplicity of vertices
plays significant roles in determining the features of $n$-dim
MMP hypergraphs $k$-$l$. In particular, by the Vertex-Hyperedge
Lemma \ref{th:theorem} we show that the sum of multiplicities is
equal to $nl$; in Sec.~\ref{subsec:4dm} we show that MMP
hypergraphs with odd number of hyperedges predominantly have
vertices with even multiplicities (see Fig.~\ref{fig:for17}
and the figure in Appendix \ref{app4}) and in Table \ref{T:masters}
that the multiplicities of vertices uniquely characterize master MMP
hypergraphs we use to generate all known MMP hypergraphs classes
from. 

\subsubsection{\label{subsub:3d}3-dim {\rm MMP} hypergraphs;
    Graph vs.~{\rm MMP} hypergraph representations}

In Sec.~\ref{subsec:3d} we discuss (see Fig.~\ref{fig:3dup}) the
four previously known 3-dim KS MMP hypergraphs: Bub's 49-36,
Conway-Kochen's 51-37, Peres' 57-40, and Kochen-Specker's 192-118
and point out that they are critical, i.e., that none of them
contains any smaller KS sets. By removing vertices with $m=1$ from
these KS MMP hypergraphs, we obtain, via the method presented in
Sec.~\ref{subsec:disc1}, the non-binary contextual 33-36, 31-37,
33-40, and 117-118 sets, respectively, but they are not KS sets,
contrary to what \cite[Table 1, p.~21]{budroni-cabello-rmp-22}
might mislead the reader in.

In Fig.~\ref{fig:3dup1} we give five new 3-dim critical MMP
hypergraphs obtained among thousands of such new 3-dim ones
in \cite{pavicic-pra-22}. 

In Sec.~\ref{subsec:3d} we show that Kochen and Specker's
original presentation of their 192-118 set or Budroni, Cabello,
G{\"u}hne, Kleinmann and Larsson's
\cite[Fig.~1]{budroni-cabello-rmp-22} ``simplification'' of
that set is neither a graph, nor a general hypergraph, nor an
MMP hypergraph. 

Note that no definite vectors for the KS set 192-118 have
been given prior to the ones provided in \cite{pavicic-pra-22}. 

None of the 3-dim KS sets satisfies the $\alpha_r^*$-inequality,
but most of their smaller $\overline{\rm subhyper}$-
$\overline{\rm graphs}$ do, in
particular the critical ones \cite{pavicic-entropy-19}. Some are
shown in Fig.~\ref{fig:alpha-3d} and Table \ref{T:alpha-3d}.

\subsubsection{\label{subsub:4d}The smallest non-binary
  {\rm MMP} hypergraph that exists and other small
  4-dim {\rm MMP}s}

In Sec.~\ref{subsec:4d} we review several chosen small 4-dim MMP
hypergraphs and give their parameters in Table \ref{T:4d1} and
their figures in Fig.~\ref{fig:22}. Binary 18-9, non-binary critical
22-13, and binary 22-13 shown in Figs.~\ref{fig:22}(a,b,c),
respectively, violate the $\alpha_r^*$-inequality. When the $m=1$
vertices are dropped from the binary
18-9 one obtains a non-binary MMP $\overline{\rm subhypergraph}$
10-9 shown in Fig.~\ref{fig:18-9-non}(a). One of its critical
$\overline{\rm subhypergraphs}$ is the non-binary 3-3 MMP
hypergraph shown in Fig.~\ref{fig:18-9-non}(d)---the smallest
4-dim MMP hypergraph with coordinatization that exists. Since
the 3-dim 3-3 MMP hypergraph does not have a coordinatization,
the obtained 4-dim 3-3 is the smallest contextual set that
exists in any dimension. This is the Result \ref{smallest}.

Another result obtained in that section is the Result
\ref{sub-non}.

\subsubsection{\label{subsub:4dgh}4-dim:
    Graph vs.~{\rm MMP} hypergraph case}

In Sec.~\ref{subsec:magic} we analyze a recent experiment
\cite{d-ambrosio-cabello-13} and show how and why graph
representation of contextual sets lead to wrong experimental and
theoretical results.  

In particular they claim that all 18 vectors they consider contribute
with an equal weight and that therefore their implementation is a
proper KS set. We show their graph in Fig.~\ref{fig:graph}(e). Their
red and green edges contain only two vertices, though. For example,
vertices {\tt 1,2,5,10,15,18}, i.e., edges {\tt 1-10,\ 2-15} and
{\tt 5-18}, the probabilities $p_{1,10}$, $p_{2,15}$, and $p_{5,18}$
from measurement data were obtained. We show their set in the MMP
hypergraph representation in Fig.~\ref{fig:graph}(g). It is an MMP
18-18, but it should have a coordinatization, i.e., $m=1$ vertices
should be added while performing the experiment and they are shown
in Fig.~\ref{fig:graph}(f) as gray dots---we end up with an MMP with
36 vertices and 18 hyperedges---and that leads us into a
contradiction: 36-18 MMP has {\em no\/} coordinatization. So,
not only that the set is not a KS set, but the measurement data
themselves are inconsistent. 

There are two possible remedies for the contradiction:
\begin{itemize}
\item merge triples of gray vertices at the intersections of red
  and green hyperedges as shown in Fig.~\ref{fig:graph}(h); new
  measurements should be carried out for the additional six
  vertices of the new 24-18 MMP hypergraph; it is one of 1233 KS
  MMP hypergraphs \cite{pmm-2-09} contained in Peres' 24-24
  master set;
\item abandon green and red hyperedges altogether and reduce the
  implementation to the 18-9 KS MMP hypergraph
  (Fig.~\ref{fig:graph}(c));
\end{itemize}

\subsubsection{\label{subsub:alpha}$\alpha^*$-inequality
    vs.~quantum computation and quantum indeterminacy}

As presented in details in Sec.~\ref{subsec:magic} Howard, Wallman,
Veitech, and Emerson have shown that stabilizer operations with
quantum bits initialized superposition of states (``magic states''),
can be used to purify quantum gates provided they exhibit the
contextuality \cite{magic-14}. As a proof that considered sets are
contextual the authors make use of the GLS inequality
\cite[p.~192]{gro-lovasz-schr-81} by invoking
Ref.~\cite{cabello-severini-winter-14}.

In the latter reference, two simple examples are given for
inducing small contextual non-binary MMP hypergraphs from bigger
noncontextual binary ones. This is
essentially the second procedure we referred to at the end of
Sec.~\ref{subsec:disc1}. A generalization of the procedure is
offered, which would consist in a recognition of the GLS inequality
\cite{gro-lovasz-schr-81} as a noncontextuality inequality
(Def.~\ref{def:non-c-i}). A similar approach is carried out by
Cabello \cite[Supp.~Material, Sec.~IV]{cabello-21}. However,
as we show in Theorem \ref{th:alpha-star}, the assumption of
variable probabilities of detecting outputs from hyperedges would
clash with the postulate of quantum indeterminacy \ref{postulate}
and therefore the $\alpha_r^*$-inequality given by
Eq.~(\ref{eq:alpha-alpha}) should be used, instead. The theorem
then states that $\alpha_r^*$-inequality is not a noncontextuality
inequality since arbitrary many contextual and noncontextual MMP
hypergraphs violate it. 

In Ref.~\cite{magic-14} the same problem emerges. The exclusivity
graph---``a source of quantum computer’s power''
\cite{bartlett-nature-14,magic-14}---is a non-binary 30-108 MMP
hypergraph and a $\overline{\rm subhypergraph}$ of a non-critical
KS 232-108 MMP hypergraph as analyzed in Sec.~\ref{subsec:magic}.
In Fig.~\ref{fig:magic}(b) we see that we can extend the original
30-108 MMP hypergraph to the KS 232-108 one by adding $m=1$
vertices to the former one. These added vertices enable
us to identify additional independent vertices in addition to the
original 8 thus exceeding the upper bound. We obtain
$101\le\alpha>\alpha_r^*=58$ (Table \ref{T:gamma}), and it is an
open problem to prove that that is not relevant for the proof that
the appropriate inequality is $\alpha=2^3=8<\alpha^*=2^3+1=9$ as
given in Ref.~\cite{bartlett-nature-14}. On the other hand, for
the 30-108 MMP hypergraph itself, a calculation which takes into
account all its edges and their 30 vertices yields
$\alpha_r^*=14.07$---see Eq.~(\ref{eq:alpha-alpha-gamma}).

Taken together, if the only reason for invoking the GLS inequality
was to prove that exclusivity MMP hypergraphs suitable for quantum
computation are contextual, then the more efficient approach would
be to directly check the obtained measurement data on contextuality,
e.g., via e-inequalities. 

\subsubsection{\label{subsub:peres-mermin}Peres-Mermin
             square: operators vs.~{\rm MMP} hypergraphs}

The Peres-Mermin square has received a great deal of attention
both theoretically and experimentally
\cite{peres90,mermin90,mermin93,cabello-08,k-cabello-blatt-09,pavicic-book-13,hofer-20}. So far it has been formulated only via operators
(the claim that it ``can be converted into a standard
proof of the KS theorem with vectors''
\cite[p.~8]{budroni-cabello-rmp-22} is incorrect as explained
in Sec.~\ref{subsec:mermin}), and in this paper we arrive at an MMP
hypergraph representation of the basic features of the Peres-Mermin
operators so as to examine possible candidates for such a
representation. We find that a non-KS 45-18 MMP hypergraph,
shown in Fig.~\ref{fig:p-m}(d), serves the purpose. This is because
it contains contextual 9-18 non-binary MMP hypergraph
(Fig.~\ref{fig:p-m}(e)) which we obtain by dropping all vertices
with multiplicity $m=1$ from the 45-18 MMP (Fig.~\ref{fig:p-m}(d)).
Notice that there are no states which would satisfy the conditions
given by Eqs.~(\ref{eq:p-m-c}) and (\ref{eq:p-m-cc}), so there is
no MMP contradiction which would correspond to the operator
Peres-Mermin contradiction. Notice also that vertices belonging
to rows and columns of the original operator-based schematics
(Fig.~\ref{fig:p-m}(d)) are mutually orthogonal in correspondence
to mutual commutations of operators in the same rows and columns.

The 9-18 MMP hypergraph is not critical. It contains 3-3, 5-5, 7-7,
and 9-9 criticals as shown in Figs.~\ref{fig:p-m}(e-h). Their
possible experimental implementations might use Peres-Mermin
operators but not in the square arrangements in which systems run
through triple gates as depicted in Fig.~\ref{fig:p-m}(a). Instead,
one of the Fig.~\ref{fig:p-m}(c-h) arrangements might be used.

As for the critical MMP subhypergraph 5-5, it is isomorphic to the
pentagon from Ref.~\cite{klyachko-08} but they are not equivalent
since they live in two different spaces, 4-dim and 3-dim,
respectively. That is why the coordinatization of complete, filled,
MMP hypergraphs are so essential. For example, while 5-5 can be
represented in both 3-dim and 4-dim spaces, 3-3 or 4-4 cannot,
because coordinatizations for a filled 3-3 (e.g.,
{\tt 1A2,2B3,3C1.}) or a filled 4-4 (e.g.,
{\tt 1A2,2B3,3C4,4D1.}) do not exist in a 3-dim space.

To sum up, the original operator formulation of the Peres-Mermin
square is inconsistent because classical observables $S_j$ which
would assign $\pm 1$ to the states have no quantum counterparts
since there is no quantum state $|\Psi\rangle$ of a system which
$\Sigma_j$ might project to states $\pm|\Psi\rangle$ (eigenstates
with $\pm 1$ eigenvalues). The fact that a correlated
noncontextuality cannot be formulated and that therefore the
Peres-Mermin square contextuality is void of its meaning we
presented as the Peres-Mermin contradiction \ref{p-m-contr}. 

\subsubsection{\label{subsub:pentagon}The pentagon case}

In Sec.~\ref{sec:hyp-op} we make use of different
coordinatizations to compare hypergraph inequalities with the
operator ones, on the example of Klyachko, Can, Binicio{\u g}lu,
and Shumovsky's 3-dim pentagon. They consider particular
coordinatization with vectors/states which an operator projects
to a chosen state $\Psi$ so as to give the maximum quantum mean
value $\sqrt 5$ as in Eq.~(\ref{eq:pent-q}) and the operator
inequality $2<\sqrt 5$, i.e., the inequality is state-dependent.
The hypergraph approach, on the other hand, gives the quantum
hypergraph index $HI_q=5$, the e- and v-inequalities $4<5$
(Eq.~(\ref{eq:in-25})), and the $\alpha_r^*$ inequality
$2=\alpha<\alpha_r^*=\frac{5}{2}$. 
The inequalities arise from the structure of the MMP hypergraph
alone independently of the states that build its coordinatization
and that makes the MMP hypergraph state independent in the sense
that its contextuality holds for any set of states that can build
its coordinatization and even when there are no such states.
(Note that no projections of or to any states are involved.)
Still, the coordinatization plays a role in the geometric
representation of MMP hypergraphs; e.g., a 3-dim
pentagon can never be planar even when all five of its
$m>1$ vectors span a plane.

\subsection{\label{subsec:disc4}Higher dimensions}

In Secs.~\ref{subsec:5d}, \ref{subsec:6d} and \ref{subsec:78d} we
give give figures and structural properties of the smallest
MMP hypergraphs from 6- to 8-dim spaces, respectively. 

\subsection{\label{subsec:disc5}Implementation and 
  research perspectives}

In Sec.~\ref{sec:impl} we presented existing general implementation
schemes for implementation of arbitrary MMP hypergraphs and a broad
class of a qubit gate network. This opens up a research road
to possible universal automated preparation algorithms for
implementation of arbitrary contextual sets in contrast to
existing approach in the literature which usually limits itself
to experiments with the smallest contextual sets and theoretical
constructs based on them which however have not been exemplified
on arbitrarily chosen contextual sets. Actually, such
small-set-approach often yields inconsistent results, some
of which we discussed in the paper in an attempt to purify the
handling of the contextual sets and find their general features.
An abundance of available contextual sets would support achieving
this goal and therefore we also generated billions of them in odd
(3 to 9; \cite{pavicic-pra-22,pavicic-entropy-19}) and even
(4 to 32; \cite{pavicic-pra-17,pm-entropy18,pwma-19}) dimensional
spaces and stored them in a freely available data base wherefrom
one can download sets with required structure and complexity for
any application.

\vspace{6pt}

\bigskip
The programs are available at
  {\tt http://puh.srce.hr/s/Qegixzz2BdjYwFL}

\acknowledgments{Computational support was provided by the
  cluster Isabella of the Zagreb University Computing Centre.
  The support of Emir Imamagi\'c from Isabella and CRO-NGI to
  the technical work is gratefully acknowledged. Reading of
  and commenting on the paper by late Norman D.~Megill was
  particularly appreciated.}

\bigskip\bigskip
\centerline{\huge Appendices}

\appendix

\section{\label{app4}MMP hypergraph multiplicity
   with even number of hyperedges (Fig.~\ref{fig:for18})}

\begin{figure}[ht]
  \begin{center}
  \includegraphics[width=1\textwidth]{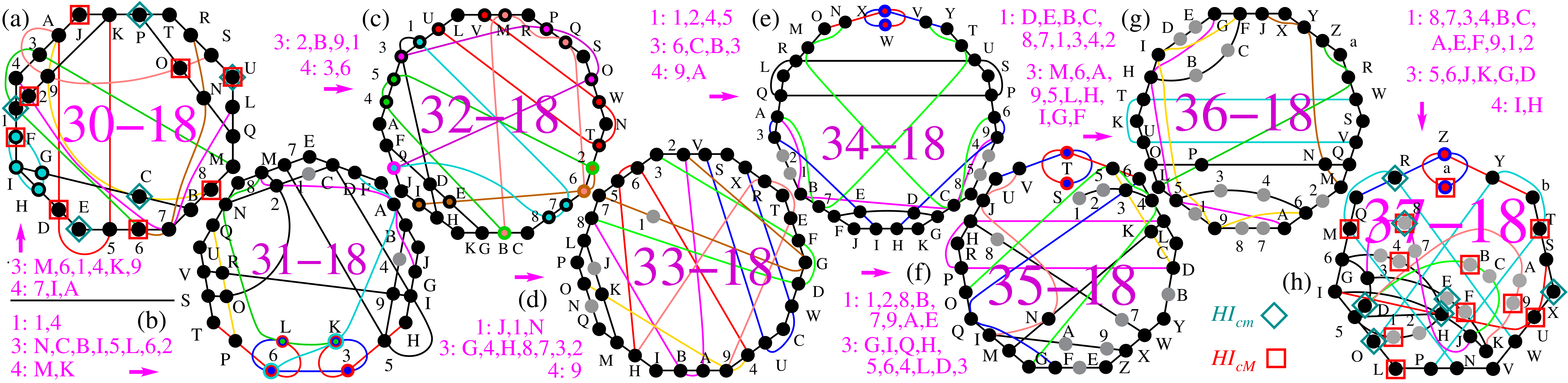}
\caption{Samples of 4-dim KS criticals with 18 hyperedges from
  the 156-249 class, from the lowest to highest number of vertices.
  Vertex multiplicities $m$ different from 2 are indicated for each
  set. There are odd multiplicities $m$ in all sets. Examples of
  distributions of the maximal and minimal numbers of ``classical 1s''
  (red squares and cyan diamonds, respectively) are shown for
  (a) 30-18 and (h) 37-18. None of the sets has a parity proof.}
\label{fig:for18}
\end{center}
\end{figure}

\section{\label{app3d}Historical 3-dim KS sets in
   a renewed MMP hypergraph representation}

 In Fig.~\ref{fig:3dup} we give the four historically known
3-dim sets in MMP hypergraph graphical presentation. 
 
\begin{figure}[ht]
  \flushleft
  \includegraphics[width=1\textwidth]{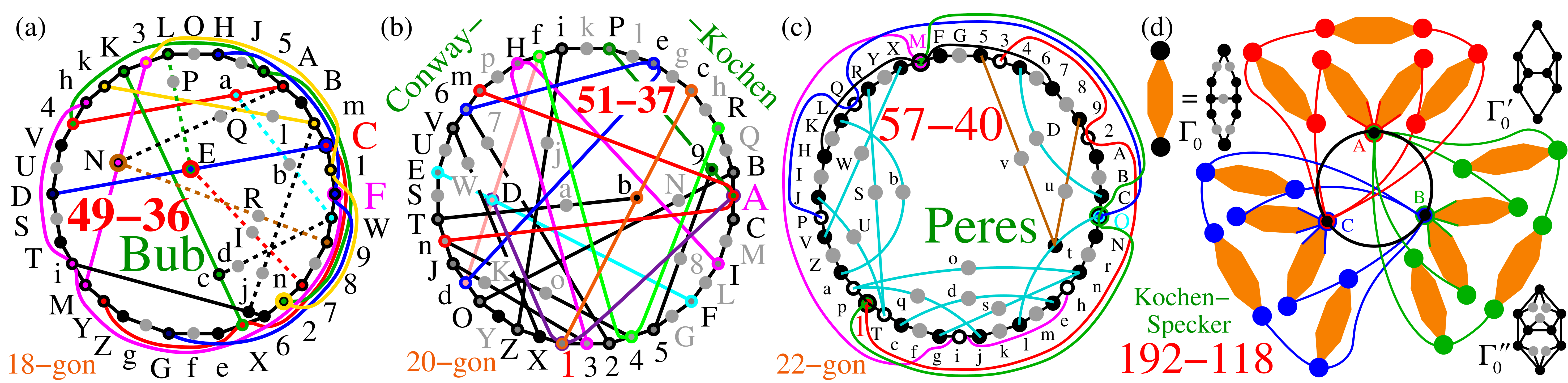}
\caption{(a) Bub's 49-36 MMP hypergraph; $m({\tt C})=4$ and
  $m({\tt F})=4$; gray dots denote vertices that belong to just
  one hyperedge, i.e., $m=1$; its maximal loop is an 18-gon;
  (b) Conway-Kochen's 51-37; $m({\tt 1})=4$ and $m({\tt A})=4$;
  (c) Peres' 57-40; it has three $m=4$ multiplicities:
  $m({\tt 1})=4$, $m({\tt M})=4$, and $m({\tt O})=4$;
  (d) the original Kochen-Specker KS set
  \cite[Fig.~on p.~69]{koch-speck} 192-118  in the MMP
  hypergraph notation \cite{pavicic-pra-17} with 15 $\Gamma_0$
  contextual non-binary MMP hexagons; $\Gamma_0'$ is a graph
  representation of $\Gamma_0$ with $m=1$ gray vertices dropped;
  $\Gamma_0''$ is a graph representation of $\Gamma_0$;
  $m({\tt A})=9$, $m({\tt B})=9$, and $m({\tt C})=9$; in
  \cite[Appendix]{pavicic-pra-22}, apparently the only
  existing coordinatization in the literature is provided by
  means of 24 components.}
\label{fig:3dup}
\end{figure}

\bigskip
\section{\label{app3string}ASCII strings and
  coordinatizations of 3-dim MMP
 hypergraphs given in Fig.~\ref{fig:3dup1}}

  \parindent=0pt

  {\bf 53-38} {\tt 213,39A,AFG,GpB,BNX,XWY,YdK,KVf,fe5,546,678,8ED,DIr,rqO,OLP,Pkl,lCH,\break HMa,aZb,bhJ,JSj,ji2,BC5,HI2,JEC,KIG,L63,MNL,LKJ,QRS,TUV,cd8,ghA,mnP,opO,\break nFE,UND,RMF.} {\tt 1}=(5,2,1);  {\tt 2}=(-1,2,1); {\tt 3}=(0,1,-2); {\tt 4}=(5,-1,2); {\tt 5}=(1,1,-2); {\tt 6}=(0,2,1);\break  {\tt 7}=(5,1,-2); {\tt 8}=(1,-1,2); {\tt 9}=(5,-2,-1); {\tt A}=(1,2,1); {\tt B}=(1,1,1); {\tt C}=(1,-1,0); {\tt D}=(-1,1,1); {\tt E}=(1,1,0); {\tt F}=(1,-1,1); {\tt G}=(1,0,-1); {\tt H}=(1,1,-1); {\tt I}=(1,0,1); {\tt J}=(0,0,1); {\tt K}=(0,1,0); {\tt L}=(1,0,0); {\tt M}=(0,1,1); {\tt N}=(0,1,-1); {\tt O}=(0,1,2); {\tt P}=(0,2,-1); {\tt Q}=(2,1,5); {\tt R}=(2,1,-1); {\tt S}=(1,-2,0); {\tt T}=(-2,5,-1); {\tt U}=(2,1,1); {\tt V}=(1,0,-2); {\tt W}=(2,5,-1); {\tt X}=(-2,1,1); {\tt Y}=(1,0,2); {\tt Z}=(-2,1,5); {\tt a}=(2,-1,1); {\tt b}=(1,2,0); {\tt c}=(1,5,2); {\tt d}=(2,0,-1); {\tt e}=(-1,5,2); {\tt f}=(2,0,1); {\tt g}=(-1,-2,5); {\tt h}=(2,-1,0); {\tt i}=(1,-2,5); {\tt j}=(2,1,0);\break  {\tt k}=(5,-1,-2);$\,${\tt l}=(1,1,2);$\,${\tt m}=(5,1,2);$\,${\tt n}=(-1,1,2);$\,${\tt o}=(5,2,-1);$\,${\tt p}=(1,-2,1);$\,${\tt q}=(5,-2,1);$\,${\tt r}=(1,2,-1).

\medskip  
{\bf 54-39} {\tt 546,6DE,EmW,WRV,VUJ,JHI,Ipq,qTs,srG,GFC,CAB,B38,879,9ZL,LMN,NOP,PbY,\break Yci,ihj,jdg,gef,fXa,a25,123,KLJ,QRP,STN,XYG,bI3,cE9,cT2,dC6,dbZ,XV8,dVT,\break klR,nSB,oU5,laZ.} {\tt 1}=(-1,1,2); {\tt 2}=(1,1,0); {\tt 3}=(1,-1,1); {\tt 4}=(5,1,-2); {\tt 5}=(1,-1,2); {\tt 6}=(0,2,1); {\tt 7}=(1,-2,1); {\tt 8}=(1,0,-1); {\tt 9}=(1,1,1); {\tt A}=(5,-2,-1); {\tt B}=(1,2,1); {\tt C}=(0,1,-2); {\tt D}=(5,-1,2);\break  {\tt E}=(1,1,-2); {\tt F}=(5,2,1);{\tt G}=(-1,2,1); {\tt H}=(-2,5,1); {\tt I}=(2,1,-1); {\tt J}=(1,0,2); {\tt K}=(2,5,-1); {\tt L}=(-2,1,1); {\tt M}=(2,-1,5); {\tt N}=(1,2,0); {\tt O}=(-2,1,5); {\tt P}=(2,-1,1); {\tt Q}=(2,5,1); {\tt R}=(1,0,-2); {\tt S}=(2,-1,0); {\tt T}=(0,0,1); {\tt U}=(2,0,-1); {\tt V}=(0,1,0); {\tt W}=(2,0,1); {\tt X}=(1,0,1); {\tt Y}=(1,1,-1); {\tt Z}=(0,1,-1); {\tt a}=(-1,1,1); {\tt b}=(0,1,1); {\tt c}=(1,-1,0); {\tt d}=(1,0,0); {\tt e}=(5,-2,1); {\tt f}=(1,2,-1); {\tt g}=(0,1,2); {\tt h}=(5,-1,-2); {\tt i}=(1,1,2);\break  {\tt j}=(0,2,-1); {\tt k}=(-2,5,-1); {\tt l}=(2,1,1); {\tt m}=(-1,5,2); {\tt n}=(-1,-2,5); {\tt o}=(1,5,2); {\tt p}=(2,1,5);\break  {\tt q}=(1,-2,0); {\tt r}=(1,-2,5); {\tt s}=(2,1,0). 

\medskip
{\bf 55-40} {\tt 213,3EF,FKL,Lpq,qaV,Vgd,dDc,cBb,biZ,ZYo,onJ,JRT,Tjk,keX,XWm,mlH,HSN,\break N9M,Mr5,546,6PO,Ot2,789,AB6,CD9,GHB,IJD,QR3,UPL,VW5,XY8,Za2,ecF,eaN,dYP,\break bWR,RPN,fgS,hiU,sQ8.} {\tt 1}=(5,-1,2), {\tt 2}=(1,1,-2), {\tt 3}=(0,2,1); {\tt 4}=(2,5,1), {\tt 5}=(2,-1,1),\break  {\tt 6}=(1,0,-2), {\tt 7}=(-1,-2,5), {\tt 8}=(1,2,1), {\tt 9}=(2,-1,0), {\tt A}=(-2,5,-1), {\tt B}=(2,1,1), {\tt C}=(1,2,5),\break  {\tt D}=(1,2,-1), {\tt E}=(5,1,-2); {\tt F}=(1,-1,2), {\tt G}=(-2,-1,5), {\tt H}=(1,-2,0), {\tt I}=(5,-2,1), {\tt J}=(0,1,2), {\tt K}=(1,5,2), {\tt L}=(2,0,-1), {\tt M}=(1,2,0), {\tt N}=(0,0,1), {\tt O}=(2,0,1), {\tt P}=(0,1,0), {\tt Q}=(0,1,-2), {\tt R}=(1,0,0), {\tt S}=(2,1,0), {\tt T}=(0,2,-1), {\tt U}=(1,0,2), {\tt V}=(1,1,-1), {\tt W}=(0,1,1), {\tt X}=(1,-1,1), {\tt Y}=(1,0,-1), {\tt Z}=(1,1,1), {\tt a}=(1,-1,0), {\tt b}=(0,1,-1); {\tt c}=(-1,1,1), {\tt d}=(1,0,1), {\tt e}=(1,1,0), {\tt f}=(1,-2,5), {\tt g}=(-1,2,1), {\tt h}=(2,5,-1), {\tt i}=(-2,1,1), {\tt j}=(5,1,2), {\tt k}=(-1,1,2), {\tt l}=(2,1,5), {\tt m}=(2,1,-1); {\tt n}=(5,2,-1), {\tt o}=(1,-2,1), {\tt p}=(-1,5,-2), {\tt q}=(1,1,2), {\tt r}=(-2,1,5), {\tt s}=(5,-2,-1), {\tt t}=(-1,5,2)

\medskip
{\bf 57-41} {\tt 213,398,876,645,5cd,dBa,abG,GWV,VTU,UAR,RSQ,QrO,OHP,Pmn,nJD,DlF,FjE,\break EpI,ItY,YCf,fe2,ABC,DC6,EB3,FGH,IJK,LK5,MK2,NH8,KHA,XSG,YZG,gZL,hbM,ijL,\break klM,opP,qTO,stN,uvN,vaJ.} {\tt 1}=(1,2,5); {\tt 2}=(2,-1,0); {\tt 3}=(1,2,-1); {\tt 4}=(-1,2,5); {\tt 5}=(2,1,0);\break  {\tt 6}=(-1,2,-1); {\tt 7}=(5,2,-1); {\tt 8}=(0,1,2); {\tt 9}=(-5,2,-1); {\tt A}=(0,1,0); {\tt B}=(1,0,1); {\tt C}=(1,0,-1); {\tt D}=(1,1,1); {\tt E}=(-1,1,1); {\tt F}=(0,1,-1); {\tt G}=(0,1,1); {\tt H}=(1,0,0); {\tt I}=(1,1,0); {\tt J}=(1,-1,0); {\tt K}=(0,0,1); {\tt L}=(-1,2,0); {\tt M}=(1,2,0); {\tt N}=(0,2,-1); {\tt O}=(0,-1,2); {\tt P}=(0,2,1); {\tt Q}=(-1,2$\omega$,$\omega$); {\tt R}=(1,0,$\omega$); {\tt S}=(1,$\omega$,-$\omega$);\break  {\tt T}=(1,2$\omega$,$\omega$); {\tt U}=(1,0,-$\omega$); {\tt V}=(1,-$\omega$,$\omega$); {\tt W}=(2,$\omega$,-$\omega$); {\tt X}=(2,-$\omega$,$\omega$); {\tt Y}=(1,-1,1); {\tt Z}=(2,1,-1);\break  {\tt a}=(1,1,-1); {\tt b}=(2,-1,1); {\tt c}=(-1,2,-5); {\tt d}=(-1,2,1); {\tt e}=(1,2,-5); {\tt f}=(1,2,1); {\tt g}=(2,1,5);\break  {\tt h}=(2,-1,-5); {\tt i}=(2,1,-5); {\tt j}=(2,1,1); {\tt k}=(2,-1,5); {\tt l}=(2,-1,-1); {\tt m}=(5,-1,2); {\tt n}=(-1,-1,2);\break  {\tt o}=(-5,-1,2); {\tt p}=(1,-1,2); {\tt q}=(-5,2$\omega$,$\omega$); {\tt r}=(5,2$\omega$,$\omega$); {\tt s}=(5,1,2); {\tt t}=(-1,1,2); {\tt u}=(-5,1,2);\break  {\tt v}=(1,1,2).

\medskip
{\bf 69-50} \ \ {\tt 451,176,6wK,KLG,G3H,HNM,MqD,DC2,2AB,BmX,XWR,RZY,Y\&8,8\%U,UVQ,QST,TId,\break dca,aef,f\$F,F\#h,hgb,bij,jt4,123,189,2EF,GIJ,HOP,QR3,ab3,ATk,AYl,BUn,CIo,\break COp,DKr,4fs,5hu,5dv,6Ox,M7y,MXh,Iz7,KUf,OYj,Ed!,Ej",X9$'$,T(9.} {\tt 1}=(0,0,$\omega$);\break  {\tt 2}=(0,$\omega$,0); {\tt G}=(0,$\omega$,$\omega$); {\tt H}=(0,$\omega$,-$\omega$); {\tt Q}=(0,$\omega$,$\omega^2$); {\tt R}=(0,$\omega$,-$\omega^2$); {\tt a}=(0,$\omega^2$,$\omega$); {\tt b}=(0,$\omega^2$,-$\omega$); {\tt 3}=($\omega$,0,0); {\tt A}=($\omega$,0,$\omega$); {\tt B}=($\omega$,0,-$\omega$); {\tt C}=($\omega$,0,$\omega^2$); {\tt D}=($\omega$,0,-$\omega^2$); {\tt 4}=($\omega$,$\omega$,0); {\tt 5}=($\omega$,-$\omega$,0);\break  {\tt 6}=($\omega$,$\omega^2$,0); {\tt M}=($\omega$,$\omega^2$,$\omega^2$); {\tt I}=($\omega$,$\omega^2$,-$\omega^2$);{\tt z}=($\omega$,$\omega^2$,2$\omega^2$); {\tt 7}=($\omega$,-$\omega^2$,0); {\tt K}=($\omega$,-$\omega^2$,$\omega^2$);\break  {\tt O}=(-$\omega$,$\omega^2$,$\omega^2$); {\tt x}=($\omega$,-$\omega^2$,2$\omega^2$); {\tt g}=(2$\omega^2$,-$\omega^2$,-$\omega$); {\tt r}=($\omega$,2$\omega^2$,$\omega^2$); {\tt p}=($\omega$,2$\omega^2$,-$\omega^2$);\break  {\tt c}=(2$\omega^2$,-$\omega^2$,$\omega$); {\tt w}=(-$\omega$,$\omega^2$,2$\omega^2$); {\tt y}=(-$\omega$,-$\omega^2$,2$\omega^2$); {\tt o}=(-$\omega$,2$\omega^2$,$\omega^2$); {\tt q}=(-$\omega$,2$\omega^2$,-$\omega^2$);\break  {\tt e}=(2$\omega^2$,$\omega^2$,-$\omega$); {\tt i}=(2$\omega^2$,$\omega^2$,$\omega$); {\tt P}=(2$\omega$,$\omega^2$,$\omega^2$); {\tt L}=(2$\omega$,$\omega^2$,-$\omega^2$); {\tt W}=(2$\omega^2$,-$\omega$,-$\omega^2$);\break  {\tt J}=(2$\omega$,-$\omega^2$,$\omega^2$); {\tt N}=(2$\omega$,-$\omega^2$,-$\omega^2$); {\tt S}=(2$\omega^2$,-$\omega$,$\omega^2$); {\tt E}=($\omega^2$,0,$\omega$); {\tt F}=($\omega^2$,0,-$\omega$); {\tt 8}=($\omega^2$,$\omega$,0); {\tt X}=($\omega^2$,$\omega$,$\omega^2$); {\tt T}=($\omega^2$,$\omega$,-$\omega^2$); {\tt (}=($\omega^2$,$\omega$,2$\omega^2$); {\tt 9}=($\omega^2$,-$\omega$,0); {\tt U}=($\omega^2$,-$\omega$,$\omega^2$); {\tt Y}=(-$\omega^2$,$\omega$,$\omega^2$);\break  {\tt \&}=($\omega^2$,-$\omega$,2$\omega^2$); {\tt n}=($\omega^2$,2$\omega$,$\omega^2$); {\tt l}=($\omega^2$,2$\omega$,-$\omega^2$); {\tt h}=($\omega^2$,$\omega^2$,$\omega$); {\tt d}=($\omega^2$,$\omega^2$,-$\omega$); {\tt v}=($\omega^2$,$\omega^2$,2$\omega$); {\tt f}=($\omega^2$,-$\omega^2$,$\omega$); {\tt j}=(-$\omega^2$,$\omega^2$,$\omega$); {\tt t}=($\omega^2$,-$\omega^2$,2$\omega$); {\tt \$}=($\omega^2$,2$\omega^2$,$\omega$); {\tt "}=($\omega^2$,2$\omega^2$,-$\omega$);\break  {\tt \%}=(-$\omega^2$,$\omega$,2$\omega^2$); {\tt $'$}=(-$\omega^2$,-$\omega$,2$\omega^2$); {\tt k}=(-$\omega^2$,2$\omega$,$\omega^2$); {\tt m}=(-$\omega^2$,2$\omega$,-$\omega^2$); {\tt s}=(-$\omega^2$,$\omega^2$,2$\omega$);\break  {\tt u}=(-$\omega^2$,-$\omega^2$,2$\omega$); {\tt !}=(-$\omega^2$,2$\omega^2$,$\omega$); {\tt \#}=(-$\omega^2$,2$\omega^2$,-$\omega$); {\tt V}=(2$\omega^2$,$\omega$,-$\omega^2$); {\tt Z}=(2$\omega^2$,$\omega$,$\omega^2$).

\section{\label{app1a}\quad
 {\boldmath$\alpha\le\alpha_r^*$} \ 
     violations}

ASCII strings of MMP hypergraphs from Fig.~\ref{fig:alpha-star}    
that violate Eq.~(\ref{eq:alpha-alpha}) from Theorem
\ref{th:alpha-star}, i.e., for which the inequality
$\alpha>\alpha_r^*$ holds.

\parindent=0pt   
(a) {\bf Bub's 49-36} {\tt 7lI,ICG,G5b,bVM,MPS,SAT,TjZ,Ze2,29B,BON,NdD,Dkn,n8g,gQY,YcH,\break H6m,mhF,FW7,aJS,Ke1,VQB,JkF,E8e,Z5W,hgT,U3M,kVj,iD3,LQW,NIT,hKb,2XJ,fRA,\break n5R,4AL,ZH3.}

$\alpha$-vertices: {\tt l,a,b,c,1,9,n,P,Q,Z,F,6,E,U,i,C,X,f,d,4,O}

(b) {\bf 26-15} {\tt NOQP,PQLM,M5GA,A89K,KJ76,6734,4BHE,E1CN,HIJK,FGLM,BCDE,59DE,2348,\break 12KO,4FIM.}

(c) {\bf 34-16} {\tt JMPSVN,N349AF,FDEIRY,Y1BKTX,X267CJ,TUVWXY,QRSWXY,KLMNOP,GHIJOP,\break BCFUXY,AJSVXY,56789J,CEIJOY,248FHM,1DLQXY,35EGJO.}

(d) {\bf 37-11} {\tt 789A56CB,BCDEFIHG,GHWXYVRP,PROQ3487,12345678,JKLMNIAC,STUVQRMN,\break ZaYULF28,ZaXTKE17,bJDI9ABC,bWSORIAB.}

(e) {\bf 9-3} {\tt 1234,4567,7891.}

(f) {\bf 10-5} {\tt 123,345,567,789,9A1.}

\vbox to 10pt{\vfill}
\section{\label{app2}\quad
 {\boldmath$\Gamma$} \ MMP
     hypergraphs}

Howard, Wallman, Veitech, and Emerson's exclusivity graph with
cliques (Fig.~\ref{fig:magic}(a)) has the following MMP hypergraph
string representation where hyperedges substitute the cliques:

\parindent=0pt
{\bf 30-108} \ \ {\tt 1234,45,5678,9ABC,CD,DEFG,HIJK,KJLM,LMNO,PQ,QR,RS,ST,TU,15,18,1F,\break 1G,1I,1K,1Q,26,27,2F,2G,2H,2J,2P,36,37,3D,3E,3I,3K,3Q,48,4D,4E,4H,4J,4P,59,\break 5A,5M,5O,5T,69,6A,6L,6N,6U,7B,7C,7M,7O,7T,8B,8C,8L,8N,8U,9E,9G,9I,9J,9R,AD,\break AF,AH,AK,AS,BE,BG,BH,BK,BR,CF,CI,CJ,CS,DL,DO,DR,EM,EN,ES,FM,FN,FR,GL,GO,GS,\break HN,HO,HQ,IN,IO,IP,JQ,KP,LT,MU,NT,OU,PS,PU,QT,RU.}
\parindent=9pt

\smallskip
Its filled KS 223-108 (Fig.~\ref{fig:magic}(b)) MMP hypergraph
ASCII string with all $m=1$ vertices is

\smallskip
\parindent=0pt
{\bf 232-108} {\tt 1234,4VW5,5678,9ABC,CXYD,DEFG,HIJK,KJLM,LMNO,PZaQ,QbcR,RdeS,SfgT,\break ThiU,1jk5,1lm8,1noF,1pqG,1rsI,1tuK,1vwQ,2xy6,2z!7,2"\#F,2\$\%G,2\&'H,2()J,\break 2*$-$P,3/:6,3;\textless 7,3=\textgreater D,3?@E,3[\textbackslash I,3]\textasciicircum K,3\_`Q,4\{$|$8,4\}$\sim$D,4+1+2E,4+3+4H,4+5+6J,\break 4+7+8P,5+9+A9,5+B+CA,5+D+EM,5+F+GO,5+H+IT,6+J+K9,6+L+MA,6+N+OL,6+P+QN,\break 6+R+SU,7+T+UB,7+V+WC,7+X+YM,7+Z+aO,7+b+cT,8+d+eB,8+f+gC,8+h+iL,8+j+kN,\break 8+l+mU,9+n+oE,9+p+qG,9+r+sI,9+t+uJ,9+v+wR,A+x+yD,A+z+!F,A+"+\#H,A+\$+\%K,\break A+\&+'S,B+(+)E,B+*+$-$G,B+/+:H,B+;+\textless K,B+=+\textgreater R,C+?+@F,C+[+\textbackslash I,C+]+\textasciicircum J,C+\_+`S,\break D+\{+$|$L,D+\}+$\sim$O,D++1++2R,E++3++4M,E++5++6N,E++7++8S,F++9++AM,F++B++CN,\break F++D++ER,G++F++GL,G++H++IO,G++J++KS,H++L++MN,H++N++OO,H++P++QQ,I++R++SN,\break I++T++UO,I++V++WP,J++X++YQ,K++Z++aP,L++b++cT,M++d++eU,N++f++gT,O++h++iU,\break P++j++kS,P++l++mU,Q++n++oT,R++p++qU.}
\parindent=9pt

\smallskip
ASCII string of the only critical KS MMP hypergraph we found in the
232-108 KS MMP hypergraph is:

\smallskip
\parindent=0pt
{\bf 152-71} {\tt 1234,4567,589A,BCDE,EFGH,FIJK,LMNO,ONPQ,PQRS,15jk,1Alm,1Jno,1Kpq,\break 1Mrs,1Otu,28xy,29z!,2J"\#,2K\$\%,2L\&',2N(),38/:,39;\textless,3F=\textgreater,3I?@,3M[\textbackslash ,3O]\textasciicircum ,4A{|,\break 4F}$\sim$,4I+1+2,4L+3+4,4N+5+6,5B+9+A,5C+B+C,5Q+D+E,5S+F+G,8B+J+K,8C+L+M,8P+N+O,\break 8R+P+Q,9D+T+U,9E+V+W,9Q+X+Y,9S+Z+a,AD+d+e,AE+f+g,AP+h+i,AR+j+k,BI+n+o,\break BK+p+q,BM+r+s,BN+t+u,CF+x+y,CJ+z+!,CL+"+\#,CO+\$+\%,DI+(+),DK+*+$-$,DL+/+:,\break DO+;+\textless,EJ+?+@,EM+[+\textbackslash,EN+]+\textasciicircum,FP+{+|,FS+}+$\sim$,IQ++3++4,IR++5++6,JQ++9++A,\break JR++B++C,KP++F++G,KS++H++I.}
\parindent=9pt

\smallskip
This critical MMP with $m=1$ vertices dropped, i.e., the non-binary
24-71 MMP diagram (Fig.~\ref{fig:magic}(c)) is :

\smallskip
\parindent=0pt
{\bf 24-71} {\tt 1234,45,59A8,8B,BI,IR,RJ,JC,CL,LD,DK,KS,SF,FP,POQN,NE,EM,M1,BCDE,EF,\break FIJK,LMNO,PQRS,15,1A,1J,1K,1O,28,29,2J,2K,2L,2N,38,39,3F,3I,3M,3O,4A,4F,4I,\break 4L,4N,5B,5C,5Q,5S,8C,8P,8R,9D,9E,9Q,9S,AD,AE,AP,AR,BK,BM,BN,CF,CO,DI,DO,EJ,\break IQ,JQ,KP.}
\parindent=9pt

\section{\label{app3}Vector representation of the
Peres-Mermin square}

\parindent=0pt
{\bf 45-18} {\tt 1AB2,2CD3,1EF3,4GH5,5IJ6,4KL6,7MN8,8OP9,7QR9,1ST4,4UV7,1WX7,2YZ5,\break 5ab8,2cd8,3ef6,6gh9,3ij9.}

{\tt 1}=(0,0,0,1); {\tt 2}=(0,0,1,0); {\tt 3}=(0,1,0,0); {\tt 4}=(-1,i,2,0); {\tt 5}=(1,i,0,2); {\tt 6}=(2,0,1,-1); {\tt 7}=(1,i,0,0); {\tt 8}=(i,1,0,0); {\tt 9}=(0,0,1,1); {\tt A}=(1,1,0,0); {\tt B}=(1,-1,0,0); {\tt C}=(1,0,0,1); {\tt D}=(1,0,0,-1); {\tt E}=(1,0,1,0); {\tt F}=(1,0,-1,0); {\tt G}=(1,i,0,-1); {\tt H}=(i,1,i,0); {\tt I}=(1,i,-3,-1); {\tt J}=(i,3,-i,i); {\tt K}=(1,-i,1,3); {\tt L}=(i,-3,-i,i); {\tt M}=(0,0,1,i); {\tt N}=(0,0,i,1); {\tt O}=(1,i,1,-1); {\tt P}=(1,i,-1,1); {\tt Q}=(i,1,i,-i); {\tt R}=(i,1,-i,i); {\tt S}=(0,2,i,0);\break  {\tt T}=(5,i,2,0); {\tt U}=(i,1,i,i); {\tt V}=(1,-i,1,-3); {\tt W}=(-1,i,1,0); {\tt X}=(1,-i,2,0); {\tt Y}=(0,2,0,i); {\tt Z}=(5,-i,0,-2); {\tt a}=(-i,1,i,i); {\tt b}=(1,i,3,-1); {\tt c}=(1,i,0,1); {\tt d}=(1,i,0,-2); {\tt e}=(1,0,0,2); {\tt f}=(-2,0,5,1); {\tt g}=(1,1,-1,1); {\tt h}=(1,-3,-1,1); {\tt i}=(0,0,1,-1); {\tt j}=(1,0,0,0)

\parindent=9pt

\section{\label{app5string}ASCII strings and
  coordinatizations of 5-dim MMP
 hypergraphs given in Fig.~\ref{fig:5dim}}

  \parindent=0pt

{\bf 29-16} {\tt 12345,56789,98ABC,CDEFG,GHI21,JBK78,LIM42,NOP32,QRP68,SAOM2,THRK8,\break HSFN2,TAEQ8,HFG2J,AECL8,DP528.} {\tt T}=(1,-1,1,0,-1), {\tt H}=(1,0,-1,0,0), {\tt S}=(1,-1,1,1,0),\break  {\tt A}=(0,1,1,0,0), {\tt D}=(0,0,1,0,0), {\tt E}=(1,0,0,0,1), {\tt F}=(0,1,0,1,0), {\tt C}=(1,0,0,0,-1),\break {\tt G}=(0,1,0,-1,0), {\tt Q}=(1,1,-1,0,-1), {\tt R}=(1,1,1,0,1), {\tt N}=(1,1,1,-1,0), {\tt O}=(1,1,-1,1,0), {\tt P}=(1,-1,0,0,0), {\tt 9}=(1,-1,1,0,1), {\tt 6}=(0,0,1,0,-1), {\tt 1}=(1,-1,1,-1,0), {\tt 3}=(0,0,1,1,0), {\tt 5}=(1,1,0,0,0), {\tt L}=(0,1,-1,0,0), {\tt I}=(1,1,1,1,0), {\tt M}=(1,0,0,-1,0), {\tt 4}=(1,-1,-1,1,0), {\tt 2}=(0,0,0,0,1), {\tt J}=(1,0,1,0,0), {\tt B}=(1,1,-1,0,1), {\tt K}=(0,1,0,0,-1), {\tt 7}=(-1,1,1,0,1), {\tt 8}=(0,0,0,1,0)

\smallskip

{\bf 30-16} {\tt 12345,56789,9ABCD,DEFGH,HIJK1,CLKM9,JNOPH,B8M94,QIGPH,8O945,\break RF8NH,REHK3,DQH92,STHK9,UT7L9,AUQ69.} {\tt A}=(0,0,1,0,-1), {\tt U}=(-1,1,1,0,1),\break  {\tt S}=(1,0,0,0,0), {\tt T}=(0,1,-1,0,0), {\tt R}=(1,-1,1,-1,0), {\tt D}=(1,-1,0,0,0), {\tt E}=(1,1,-1,-1,0), {\tt F}=(1,1,1,1,0), {\tt Q}=(1,1,0,0,0), {\tt I}=(1,-1,1,1,0), {\tt G}=(0,0,1,-1,0), {\tt B}=(1,1,1,0,1), {\tt 6}=(1,-1,1,0,1), {\tt 7}=(1,1,1,0,-1), {\tt 8}=(1,0,-1,0,0), {\tt J}=(1,1,-1,1,0), {\tt N}=(0,1,0,-1,0), {\tt O}=(1,0,1,0,0), {\tt P}=(-1,1,1,1,0), {\tt H}=(0,0,0,0,1), {\tt C}=(1,1,-1,0,-1), {\tt L}=(1,0,0,0,1), {\tt K}=(0,1,1,0,0), {\tt M}=(1,-1,1,0,-1), {\tt 9}=(0,0,0,1,0), {\tt 2}=(0,0,1,0,0), {\tt 3}=(1,0,0,1,0), {\tt 1}=(1,0,0,-1,0), {\tt 4}=(0,1,0,0,-1), {\tt 5}=(0,1,0,0,1)

\smallskip

{\bf 30-17} {\tt 89A7B,BCD23,3T41L,L1MNK,KIJH6,6QE58,12345,67845,EFGH6,OPNG1,\break QRFD6,SRJC6,TRMA1,URP91,SI684,UO135,KGB61.} {\tt 1}=(0,0,0,0,1), {\tt 2}=(0,1,0,-1,0),\break  {\tt 3}=(0,1,0,1,0), {\tt 4}=(1,0,1,0,0), {\tt 5}=(1,0,-1,0,0), {\tt 6}=(0,0,0,1,0), {\tt 7}=(0,1,0,0,-1), {\tt 8}=(0,1,0,0,1), {\tt 9}=(1,0,0,-1,0), {\tt A}=(1,0,0,1,0), {\tt B}=(0,0,1,0,0), {\tt C}=(1,0,0,0,1), {\tt D}=(1,0,0,0,-1), {\tt E}=(1,1,1,0,-1), {\tt F}=(1,1,-1,0,1), {\tt G}=(1,-1,0,0,0), {\tt H}=(0,0,1,0,1), {\tt I}=(1,-1,-1,0,1), {\tt J}=(1,-1,1,0,-1), {\tt K}=(1,1,0,0,0), {\tt L}=(1,-1,-1,1,0), {\tt M}=(1,-1,1,-1,0), {\tt N}=(0,0,1,1,0), {\tt O}=(1,1,1,-1,0), {\tt P}=(1,1,-1,1,0), {\tt Q}=(1,-1,1,0,1), {\tt R}=(0,1,1,0,0), {\tt S}=(1,1,-1,0,-1), {\tt T}=(1,1,-1,-1,0), {\tt U}=(1,-1,1,1,0) 

\smallskip

{\bf 34-17} {\tt 41235,5FGEH,HIJK9,968A7,7STDO,ODNP4,BCDEA,LMKH8,QRPD3,UTRD6,\break VTJH2,WTMH1,E8934,XSNCD,XUQBD,YVIGH,YWLFH.} {\tt 1}=(1,0,0,0,-1), {\tt 2}=(1,0,0,0,1),\break {\tt 3}=(0,1,0,1,0), {\tt 4}=(0,1,0,-1,0), {\tt 5}=(0,0,1,0,0), {\tt 6}=(1,1,0,0,0), {\tt 7}=(1,-1,0,0,0), {\tt 8}=(0,0,1,0,1), {\tt 9}=(0,0,1,0,-1), {\tt A}=(0,0,0,1,0), {\tt B}=(0,1,1,0,0), {\tt C}=(0,1,-1,0,0), {\tt D}=(0,0,0,0,1), {\tt E}=(1,0,0,0,0), {\tt F}=(0,0,0,1,1), {\tt G}=(0,0,0,1,-1), {\tt H}=(0,1,0,0,0), {\tt I}=(1,0,1,1,1), {\tt J}=(1,0,-1,1,-1), {\tt K}=(1,0,0,-1,0), {\tt L}=(1,0,1,1,-1), {\tt M}=(1,0,-1,1,1), {\tt N}=(-1,1,1,1,0), {\tt O}=(1,1,-1,1,0), {\tt P}=(1,0,1,0,0), {\tt Q}=(1,1,-1,-1,0), {\tt R}=(1,-1,-1,1,0), {\tt S}=(1,1,1,-1,0), {\tt T}=(0,0,1,1,0), {\tt U}=(1,-1,1,-1,0), {\tt V}=(1,0,1,-1,-1),\break  {\tt W}=(1,0,1,-1,1), {\tt X}=(1,0,0,1,0),{\tt Y}=(1,0,-1,0,0)

\smallskip

{\bf 58-40} {\tt CBDAE,EAJKI,InPSg,gShQY,YroHZ,ZpmOM,MwlVU,UVvqN,NaXR5,5bWGC,\break 12345,6789A,FGHI5,LMHI5,MI345,NO89A,PQRSD,TUVWK,XYNOA,ZYSOA,cdefZ,ijVMD,\break LK79A,aX245,kZCA5,lC69A,mnokZ,pqrZC,stjVW,efZYM,hC145,utlVH,hYWSO,dfZRK,\break vwshV,vwTUV,baNF5,YNJBA,qkcfZ,qnkZC.} {\tt 1}=(-1,1,1,1,0), {\tt 2}=(1,1,-1,1,0),\break  {\tt 3}=(1,1,1,-1,0), {\tt 4}=(1,-1,1,1,0), {\tt 5}=(0,0,0,0,1), {\tt 6}=(1,-1,1,0,1), {\tt 7}=(1,1,-1,0,1), {\tt 8}=(1,1,1,0,-1), {\tt 9}=(-1,1,1,0,1), {\tt A}=(0,0,0,1,0), {\tt B}=(1,-1,-1,0,1), {\tt C}=(1,1,0,0,0), {\tt D}=(0,0,1,0,1), {\tt E}=(1,-1,1,0,-1), {\tt F}=(1,1,1,1,0), {\tt G}=(1,-1,1,-1,0), {\tt H}=(0,1,0,-1,0), {\tt I}=(1,0,-1,0,0), {\tt J}=(1,1,1,0,1), {\tt K}=(0,1,0,0,-1), {\tt L}=(1,0,1,0,0), {\tt M}=(0,1,0,1,0), {\tt N}=(0,1,-1,0,0), {\tt O}=(1,0,0,0,1), {\tt P}=(1,0,1,-1,-1), {\tt Q}=(1,0,-1,-1,1), {\tt R}=(1,0,0,1,0), {\tt S}=(0,1,0,0,0), {\tt T}=(0,1,-1,1,1), {\tt U}=(0,1,1,-1,1), {\tt V}=(1,0,0,0,0), {\tt W}=(0,0,1,1,0), {\tt X}=(0,1,1,0,0), {\tt Y}=(1,0,0,0,-1), {\tt Z}=(0,0,1,0,0), {\tt a}=(1,0,0,-1,0), {\tt b}=(1,-1,-1,1,0), {\tt c}=(1,1,0,1,-1), {\tt d}=(-1,1,0,1,1), {\tt e}=(1,-1,0,1,1), {\tt f}=(1,1,0,-1,1), {\tt g}=(1,0,1,1,1), {\tt h}=(0,0,1,-1,0), {\tt i}=(0,1,-1,-1,1), {\tt j}=(0,1,1,-1,-1), {\tt k}=(1,-1,0,0,0), {\tt l}=(0,0,1,0,-1), {\tt m}=(1,1,0,-1,-1), {\tt n}=(0,0,0,1,-1), {\tt o}=(1,1,0,1,1), {\tt p}=(1,-1,0,1,-1), {\tt q}=(0,0,0,1,1), {\tt r}=(1,-1,0,-1,1), {\tt s}=(0,1,0,0,1), {\tt t}=(0,1,-1,1,-1), {\tt u}=(0,1,1,1,1), {\tt v}=(0,1,1,1,-1), {\tt w}=(0,-1,1,1,1) (the first ten hyperedges form a decagon).

\smallskip

{\bf 65-40} {\tt 8679A,ABCDE,EG254,45c3d,d\#RUT,TUQNu,ujIMZ,ZMaPb,blnkm,mk!f8,12345,\break FGH95,IJKLM,NOPH5,QRSTU,VWXYU,efLF5,ghijk,on125,pqrnA,odUM5,raKMG,STUJ9,\break XYUJE,daIMD,ocOF5,stucA,vwxyM,z!kcJ,"\textbackslash \#\textbackslash \$kG,\$ikfD,xyuMG,wydM8,vyPJM,\#hkUA,\break BC67A,vwrjM,rWYUN,\#oVYU,kfUN5.} \hfil (the first ten hyperedges form a decagon)\break {\tt 1}=(1,1,0,1,1), {\tt 2}=(1,-1,0,1,-1), {\tt 3}=(1,-1,0,-1,1), {\tt 4}=(1,1,0,-1,-1), {\tt 5}=(0,0,1,0,0),\break  {\tt 6}=(0,1,-1,1,-1), {\tt 7}=(0,1,1,1,1), {\tt 8}=(0,0,1,0,-1), {\tt 9}=(0,1,0,-1,0), {\tt A}=(1,0,0,0,0), {\tt B}=(0,1,-1,-1,1), {\tt C}=(0,1,1,-1,-1), {\tt D}=(0,0,1,0,1), {\tt E}=(0,1,0,1,0), {\tt F}=(1,1,0,1,-1), {\tt G}=(1,0,0,0,1), {\tt H}=(-1,1,0,1,1), {\tt I}=(1,0,-1,-1,1), {\tt J}=(1,0,1,0,0), {\tt K}=(1,0,-1,1,-1), {\tt L}=(0,0,0,1,1), {\tt M}=(0,1,0,0,0), {\tt N}=(1,1,0,0,0), {\tt O}=(1,-1,0,1,1), {\tt P}=(0,0,0,1,-1), {\tt Q}=(1,-1,1,1,0), {\tt R}=(1,1,1,-1,0), {\tt S}=(1,1,-1,1,0), {\tt T}=(-1,1,1,1,0), {\tt U}=(0,0,0,0,1), {\tt V}=(1,1,1,1,0), {\tt W}=(1,-1,1,-1,0), {\tt X}=(1,1,-1,-1,0), {\tt Y}=(1,-1,-1,1,0), {\tt Z}=(1,0,1,1,1),\break {\tt a}=(1,0,1,-1,-1), {\tt b}=(1,0,-1,0,0), {\tt c}=(0,1,0,0,1), {\tt d}=(1,0,0,1,0), {\tt e}=(1,1,0,-1,1), {\tt f}=(1,-1,0,0,0), {\tt g}=(1,-1,1,0,1), {\tt h}=(0,1,1,0,0), {\tt i}=(1,1,-1,0,1), {\tt j}=(1,0,0,0,-1), {\tt k}=(0,0,0,1,0), {\tt l}=(1,-1,1,0,-1), {\tt m}=(1,1,1,0,1), {\tt n}=(0,1,0,0,-1), {\tt o}=(1,0,0,-1,0), {\tt p}=(0,1,-1,1,1), {\tt q}=(0,1,1,-1,1), {\tt r}=(0,0,1,1,0), {\tt s}=(0,1,1,1,-1), {\tt t}=(0,-1,1,1,1), {\tt u}=(0,0,1,-1,0), {\tt v}=(1,0,-1,1,1), {\tt w}=(1,0,1,-1,1), {\tt x}=(1,0,1,1,-1), {\tt y}=(-1,0,1,1,1), {\tt z}=(1,-1,-1,0,1), {\tt !}=(1,1,-1,0,-1), {\tt "}=(-1,1,1,0,1), {\tt \#}=(0,1,-1,0,0),\break {\tt \$}=(1,1,1,0,-1) .

\smallskip

{\bf 105-136} {\tt 12345,12367,12489,12AB5,134CD,13EF5,1GH45,1GH67,1G6IJ,1GKL7,1H6MN,\break 1HOP7,1EF89,1E8IP,1E9KN,1F8ML,1F9OJ,1ABCD,1ACJP,1ADLN,1QRST,1QUVW,1XYSZ,\break 1XabW,1BCMK,1BDOI,1cYVd,1caeT,1fRbd,1fUeZ,1OIJP,1MKLN,234gh,23ij5,2kl45,\break 2kl67,2k6mn,2kop7,2l6qr,2lst7,2ij89,2i8mt,2i9or,2j8qp,2j9sn,2ABgh,2Agtn,\break 2Ahpr,2uvSw,2uxyW,2z!S",2z\#\$W,2Bgoq,2Bhsm,2\%!yd,2\%\#ew,2\&v\$d,2\&xe",2smtn,\break 2oqpr,'(345,'(367,'(489,'(AB5,'36)*,'3$-$/7,'48:;,'4\textless =9,'A\textgreater ?5,'@[B5,(36\textbackslash ],\break (3\textasciicircum \_7,(48`\{,(4$|$\}9,(A$\sim$+15,(+2+3B5,3ijCD,3iC\_*,3iD$-$\textbackslash ,3jC/],3jD\textasciicircum ),3EFgh,\break 3Eg\_),3Eh/\textbackslash ,3+4+5Vw,3+4+6yT,3+7+8V",3+7+9\$T,3Fg$-$],3Fh\textasciicircum *,3+A+8yZ,3+A+9bw,\break 3+B+5\$Z,3+B+6b",3\textasciicircum \_)*,3$-$/\textbackslash ],kl4CD,klEF5,k4C\};,k4D\textless `,kE+3?5,kF@$\sim$5,l4C=\{,\break l4D$|$:,lE[+15,lF+2\textgreater 5,GH4gh,GHij5,G4g\}:,G4h=`,Gi+3\textgreater 5,Gj[$\sim$5,+C4+5Rx,+C4+6vU,\break +C+7X\%5,+C+Azc5,+D4+8R\#,+D4+9!U,+D+4X\&5,+D+Buc5,H4g\textless \{,H4h$|$;,Hi@+15,Hj+2?5,\break +E4+8va,+E4+9Yx,+E+4zf5,+E+BQ\%5,+F4+5!a,+F4+6Y\#,+F+7uf5,+F+AQ\&5,4$|$\}:;,\break 4\textless =`\{,+2+3\textgreater ?5,@[$\sim$+15.} {\tt 1}=(0,0,0,0,1), {\tt 2}=(0,0,0,1,0), {\tt 3}=(0,0,1,0,0), {\tt 4}=(0,1,0,0,0),\break {\tt 5}=(1,0,0,0,0), {\tt 6}=(1,1,0,0,0), {\tt 7}=(1,-1,0,0,0), {\tt 8}=(1,0,1,0,0), {\tt 9}=(1,0,-1,0,0), {\tt A}=(0,1,1,0,0), {\tt B}=(0,1,-1,0,0), {\tt C}=(1,0,0,1,0), {\tt D}=(1,0,0,-1,0), {\tt E}=(0,1,0,1,0), {\tt F}=(0,1,0,-1,0), {\tt G}=(0,0,1,1,0), {\tt H}=(0,0,1,-1,0),$\,${\tt I}=(1,-1,-1,1,0),$\,${\tt J}=(1,-1,1,-1,0),$\,${\tt K}=(1,1,1,-1,0),$\,${\tt L}=(1,1,-1,1,0),$\,${\tt M}=(-1,1,1,1,0),\break  {\tt N}=(1,-1,1,1,0), {\tt O}=(1,1,1,1,0), {\tt P}=(1,1,-1,-1,0), {\tt Q}=(0,1,1,1,0), {\tt R}=(1,0,1,-1,0), {\tt S}=(1,1,-1,0,0), {\tt T}=(1,-1,0,1,0), {\tt U}=(1,0,-1,1,0), {\tt V}=(1,1,0,-1,0), {\tt W}=(1,-1,1,0,0), {\tt X}=(0,1,1,-1,0), {\tt Y}=(1,0,1,1,0), {\tt Z}=(-1,1,0,1,0), {\tt a}=(-1,0,1,1,0), {\tt b}=(1,1,0,1,0), {\tt c}=(0,1,-1,1,0), {\tt d}=(-1,1,1,0,0), {\tt e}=(1,1,1,0,0), {\tt f}=(0,-1,1,1,0), {\tt g}=(1,0,0,0,1), {\tt h}=(1,0,0,0,-1), {\tt i}=(0,1,0,0,1), {\tt j}=(0,1,0,0,-1), {\tt k}=(0,0,1,0,1), {\tt l}=(0,0,1,0,-1), {\tt m}=(1,-1,-1,0,1), {\tt n}=(1,-1,1,0,-1), {\tt o}=(1,1,1,0,-1), {\tt p}=(1,1,-1,0,1),\break  {\tt q}=(-1,1,1,0,1), {\tt r}=(1,-1,1,0,1), {\tt s}=(1,1,1,0,1), {\tt t}=(1,1,-1,0,-1), {\tt u}=(0,1,1,0,1), {\tt v}=(1,0,1,0,-1), {\tt w}=(1,-1,0,0,1), {\tt x}=(1,0,-1,0,1), {\tt y}=(1,1,0,0,-1), {\tt z}=(0,1,1,0,-1), {\tt !}=(1,0,1,0,1), {\tt "}=(-1,1,0,0,1), {\tt \#}=(-1,0,1,0,1), {\tt \$}=(1,1,0,0,1), {\tt \%}=(0,1,-1,0,1), {\tt \&}=(0,-1,1,0,1), {\tt '}=(0,0,0,1,1), {\tt (}=(0,0,0,1,-1), {\tt )}=(1,-1,0,1,-1), {\tt *}=(1,-1,0,-1,1), {\tt $-$}=(1,1,0,1,-1), {\tt /}=(1,1,0,-1,1), {\tt :}=(1,0,-1,1,-1),\break  {\tt ;}=(1,0,-1,-1,1), {\tt \textless }=(1,0,1,1,-1), {\tt =}=(1,0,1,-1,1), {\tt \textgreater }=(0,1,-1,1,-1), {\tt ?}=(0,1,-1,-1,1),\break  {\tt @}=(0,1,1,1,-1), {\tt [}=(0,1,1,-1,1), {\tt \textbackslash }=(1,-1,0,1,1), {\tt ]}=(-1,1,0,1,1), {\tt \textasciicircum }=(1,1,0,1,1),\break {\tt \_}=(1,1,0,-1,-1), {\tt `}=(1,0,-1,1,1), {\tt \{}=(-1,0,1,1,1), {\tt $|$}=(1,0,1,1,1), {\tt \}}=(1,0,1,-1,-1), {\tt $\sim$}=(0,1,-1,1,1), {\tt +1}=(0,-1,1,1,1), {\tt +2}=(0,1,1,1,1), {\tt +3}=(0,1,1,-1,-1), {\tt +4}=(0,1,0,1,1), {\tt +5}=(1,0,0,1,-1),\break {\tt +6}=(1,0,0,-1,1),  {\tt +7}=(0,1,0,1,-1), {\tt +8}=(1,0,0,1,1), {\tt +9}=(-1,0,0,1,1), {\tt +A}=(0,1,0,-1,1),\break {\tt +B}=(0,-1,0,1,1),  {\tt +C}=(0,0,1,1,1), {\tt +D}=(0,0,1,1,-1), {\tt +E}=(0,0,1,-1,1), {\tt +F}=(0,0,-1,1,1)

\section{\label{app6string}ASCII strings and
  coordinatizations of 6-dim MMP
  hypergraphs given in Fig.~\ref{fig:6d-triangl} and
  Tables \ref{T:6dc} and \ref{T:6da}}

  \parindent=0pt

  {\bf 27-9} {\tt 123456,1789AB,27CEHR,3C8DGQ,4ED9FJ,5HGFAI,BIJKLM,MLKPON,NOPQR6.}\break {\tt 1}=(0,0,1,1,$\omega$,$\omega$), {\tt 2}=(0,1,0,$\omega$,1,$\omega$), {\tt 3}=(0,1,$\omega$,0,$\omega$,1), {\tt 4}=(0,$\omega$,1,$\omega$,0,1), {\tt 5}=(0,$\omega$,$\omega$,1,1,0),\break  {\tt 6}=(1,0,0,0,0,0), {\tt 7}=(1,0,0,$\omega$,$\omega$,1), {\tt 8}=(1,0,$\omega$,0,1,$\omega$), {\tt 9}=($\omega$,0,$\omega$,1,0,1), {\tt A}=($\omega$,0,1,$\omega$,1,0),\break  {\tt B}=(0,1,0,0,0,0), {\tt C}=($\omega$,$\omega$,0,0,1,1), {\tt D}=($\omega$,1,1,0,0,$\omega$), {\tt E}=(1,$\omega$,0,1,0,$\omega$), {\tt F}=(1,1,$\omega$,$\omega$,0,0),\break  {\tt G}=(1,$\omega$,1,0,$\omega$,0), {\tt H}=($\omega$,1,0,1,$\omega$,0), {\tt I}=(0,0,0,0,0,1), {\tt J}=(0,0,0,0,1,0), {\tt K}=($\omega$,0,1,1,0,0),\break  {\tt L}=(1,0,$\omega$,1,0,0), {\tt M}=(1,0,1,$\omega$,0,0), {\tt N}=(0,1,0,0,$\omega$,1),\quad{\tt O}=(0,$\omega$,0,0,1,1),\quad{\tt P}=(0,1,0,0,1,$\omega$),\break {\tt Q}=(0,0,0,1,0,0),\quad{\tt R}=(0,0,1,0,0,0),
  
\smallskip
  
{\bf 32-11} {\tt 123456,1789AB,CDE4AF,CGH96I,JGKHLM,JNO85P,QDKRSM,QN3RTB,\break UVWHTI,UVOESP,27WHLF.} {\tt 1}=($\omega$,1,1,1,1,1), {\tt C}=($\omega^2$,$\omega$,$\omega$,1,1,1), {\tt J}=($\omega^2$,$\omega$,1,$\omega$,1,1),\break  {\tt Q}=($\omega^2$,$\omega$,1,1,$\omega$,1),{\tt U}=(1,$\omega^2$,$\omega$,1,$\omega$,1),{\tt 2}=(1,$\omega^2$,$\omega$,1,1,$\omega$),{\tt D}=(1,$\omega^2$,1,$\omega^2$,1,1),{\tt 7}=(1,$\omega^2$,1,$\omega$,$\omega$,1),\break  {\tt V}=(1,$\omega^2$,1,$\omega$,1,$\omega$),{\tt G}=(1,$\omega^2$,1,1,$\omega^2$,1),{\tt N}=(1,$\omega^2$,1,1,1,$\omega^2$),{\tt W}=($\omega^2$,1,$\omega$,$\omega$,1,1),{\tt O}=($\omega^2$,1,$\omega$,1,1,$\omega$),\break {\tt 3}=(1,$\omega$,$\omega^2$,$\omega$,1,1),$\,${\tt K}=(1,$\omega$,$\omega^2$,1,1,$\omega$),$\,${\tt E}=($\omega^2$,1,1,$\omega$,$\omega$,1),$\,${\tt R}=($\omega^2$,1,1,$\omega$,1,$\omega$),$\,${\tt H}=($\omega^2$,1,1,1,$\omega$,$\omega$),\break {\tt 8}=(1,$\omega$,$\omega$,1,$\omega^2$,1),$\,${\tt 9}=(1,$\omega$,1,$\omega^2$,1,$\omega$),$\,${\tt 4}=(1,$\omega$,1,1,$\omega$,$\omega^2$),{\tt S}=(1,1,1,1,$\omega^2$,$\omega^2$),{\tt 5}=(1,1,1,$\omega$,$\omega^2$,$\omega$),\break {\tt L}=(1,1,1,$\omega^2$,1,$\omega^2$),{\tt T}=(1,1,1,$\omega^2$,$\omega^2$,1),{\tt A}=(1,1,$\omega$,$\omega$,1,$\omega^2$),{\tt 6}=(1,1,$\omega$,$\omega^2$,$\omega$,1),{\tt I}=(1,1,$\omega^2$,1,1,$\omega^2$),\break {\tt B}=(1,1,$\omega^2$,1,$\omega$,$\omega$),{\tt F}=(1,1,$\omega^2$,1,$\omega^2$,1),{\tt P}=(1,1,$\omega^2$,$\omega^2$,1,1),{\tt M}=(1,1,$\omega^2$,$\omega$,$\omega$,1)

  \medskip

{\bf 36-13} {\tt 123456,1789AB,CD34EF,CGHIJK,CGLMEA,C7NIOP,QRHS56,QRTUEF,\break VDWSXY,V2WUJK,VZL9XO,Va8MYP,ZaTNEB.} {\tt 1}=($\omega$,1,1,1,$\omega$,$\omega^2$), {\tt C}=($\omega^2$,$\omega$,$\omega$,1,1,1),\break  {\tt Q}=($\omega$,1,1,$\omega^2$,$\omega$,1), {\tt V}=($\omega$,1,1,$\omega$,1,$\omega^2$), {\tt R}=($\omega^2$,$\omega$,1,$\omega$,1,1), {\tt D}=($\omega$,1,$\omega^2$,1,$\omega$,1), {\tt 2}=($\omega^2$,$\omega$,1,1,1,$\omega$), {\tt Z}=(1,$\omega^2$,$\omega$,$\omega$,1,1), {\tt G}=($\omega$,1,$\omega^2$,$\omega$,1,1), {\tt a}=($\omega$,1,$\omega$,1,$\omega^2$,1), {\tt 7}=(1,$\omega^2$,1,1,$\omega^2$,1),{\tt H}=($\omega$,$\omega^2$,1,1,1,$\omega$),\break {\tt L}=($\omega$,$\omega^2$,1,1,$\omega$,1), {\tt W}=(1,$\omega$,$\omega^2$,$\omega$,1,1), {\tt 3}=($\omega$,$\omega^2$,1,$\omega$,1,1), {\tt T}=(1,$\omega$,$\omega^2$,1,$\omega$,1), {\tt N}=($\omega^2$,1,1,$\omega$,$\omega$,1), {\tt 8}=(1,$\omega$,$\omega$,1,1,$\omega^2$), {\tt U}=($\omega$,$\omega^2$,$\omega$,1,1,1), {\tt 4}=(1,$\omega$,1,$\omega^2$,$\omega$,1), {\tt M}=(1,$\omega$,1,$\omega$,$\omega^2$,1),{\tt 9}=($\omega^2$,1,$\omega^2$,1,1,1), {\tt I}=(1,$\omega$,1,$\omega$,1,$\omega^2$), {\tt S}=(1,$\omega$,1,1,$\omega$,$\omega^2$), {\tt E}=(1,1,1,1,1,$\omega$), {\tt J}=(1,1,1,1,$\omega$,1), {\tt 5}=(1,1,1,$\omega$,$\omega^2$,$\omega$),\break  {\tt X}=(1,1,1,$\omega^2$,$\omega^2$,1), {\tt 6}=(1,1,$\omega$,1,1,1), {\tt O}=(1,1,$\omega$,1,$\omega$,$\omega^2$), {\tt Y}=($\omega^2$,$\omega^2$,1,1,1,1), {\tt F}=(1,1,$\omega$,$\omega$,$\omega^2$,1), {\tt K}=(1,1,$\omega$,$\omega^2$,1,$\omega$), {\tt A}=(1,1,$\omega$,$\omega^2$,$\omega$,1), {\tt P}=(1,1,$\omega^2$,$\omega^2$,1,1), {\tt B}=($\omega$,$\omega$,1,$\omega^2$,1,1)

  \medskip

{\bf 39-13} {\tt 123456,1789AB,CDEFGH,CIJKH6,LIMKAB,LNOFPQ,RST95U,RVO4PW,\break XNYTZa,X7b8cW,DVY3Ga,dSEMcU,d2bJZQ.} {\tt 1}=($\omega$,1,1,1,1,1), {\tt C}=($\omega$,1,1,$\omega^2$,$\omega$,1),\break  {\tt L}=($\omega$,1,1,$\omega$,$\omega^2$,1), {\tt R}=(1,$\omega^2$,$\omega^2$,1,1,1), {\tt X}=($\omega$,1,$\omega^2$,1,1,$\omega$), {\tt D}=($\omega^2$,$\omega$,1,$\omega$,1,1),{\tt d}=(1,$\omega^2$,$\omega$,1,1,$\omega$), {\tt I}=($\omega$,1,$\omega$,1,1,$\omega^2$), {\tt S}=($\omega$,1,$\omega$,1,$\omega^2$,1), {\tt V}=($\omega$,1,$\omega$,$\omega^2$,1,1), {\tt 2}=(1,$\omega^2$,1,$\omega$,$\omega$,1), {\tt N}=(1,$\omega^2$,1,$\omega$,1,$\omega$), {\tt 7}=(1,$\omega^2$,1,1,$\omega$,$\omega$), {\tt Y}=($\omega$,$\omega^2$,1,1,$\omega$,1), {\tt b}=($\omega^2$,1,$\omega$,1,$\omega$,1), {\tt E}=(1,$\omega$,$\omega^2$,1,$\omega$,1), {\tt J}=($\omega^2$,1,1,$\omega$,1,$\omega$), {\tt M}=($\omega^2$,1,1,1,$\omega$,$\omega$), {\tt O}=($\omega^2$,1,1,1,1,$\omega^2$),{\tt 3}=(1,$\omega$,$\omega$,1,$\omega^2$,1),{\tt T}=($\omega^2$,1,1,$\omega^2$,1,1),{\tt 8}=(1,$\omega$,$\omega$,1,1,$\omega^2$), {\tt F}=($\omega$,$\omega^2$,$\omega$,1,1,1), {\tt 4}=(1,$\omega$,1,$\omega^2$,1,$\omega$), {\tt 9}=(1,$\omega$,1,$\omega$,$\omega^2$,1), {\tt G}=(1,1,1,1,1,$\omega$), {\tt Z}=(1,1,1,1,$\omega^2$,$\omega^2$), {\tt K}=(1,$\omega$,1,1,1,1), {\tt c}=(1,1,1,$\omega$,1,1), {\tt P}=(1,1,$\omega$,1,$\omega^2$,$\omega$), {\tt 5}=(1,1,$\omega$,$\omega$,1,$\omega^2$), {\tt H}=(1,1,$\omega$,$\omega$,$\omega^2$,1), {\tt A}=(1,1,$\omega$,$\omega^2$,$\omega$,1), {\tt U}=($\omega$,$\omega$,1,1,1,$\omega^2$), {\tt W}=($\omega$,$\omega$,1,1,$\omega^2$,1),{\tt 6}=(1,1,$\omega^2$,1,$\omega$,$\omega$),{\tt Q}=(1,1,$\omega^2$,$\omega^2$,1,1),\break {\tt B}=(1,1,$\omega^2$,$\omega$,1,$\omega$), {\tt a}=(1,1,$\omega^2$,$\omega$,$\omega$,1)

\bigskip

  {\bf 81-162 master} {\tt 123456,12789A,1BCD5E,1B7FGH,1ICJ9K,1I3LGM,1NODAP,1NQ4HR,\break 1SOJ6T,1SU8MR,1VQLET,1VUFKP,WXY45Z,WX7abA,Wcde5E,Wc7Ffg,WIdJbh,WIYifM,\break WjkeAP,WjQ4gl,WSkJZm,WSnaMl,WoQiEm,WonFhP,pqY89Z,pq3ab6,pcre9K,pc3Lsg,\break pBrDbh,pBYisH,pjte6T,pjU8gu,pNtDZm,pNnaHu,pvnLhT,pvUiKm,wxyD5Z,wx7azH,\break w!"e56,w!78\#g,wI"Lzh,wIyi\#K,w\$keHR,w\$ODgl,wVkLZ\%,wV\&aKl,woOi6\%,wo\&8hR,\break '(yDbA,'(Y4zH,'!)Jb6,'!Y8*M,'c)LzE,'cyF*K,'-kJHu,'-tDMl,'vkLA/,'v:4Kl,\break 'otF6/,'o:8Eu,;("e9A,;(34\#g,;x)J9Z,;x3a*M,;B)i\#E,;B"F*h,;\textless teMR,;\textless OJgu,\break ;vOiA=,aA/;v\textgreater 4hR,;VtFZ=,;V\textgreater aEu,?!reGM,?!CJsg,?qyFGZ,?qCazE,?2yisA,?2r4zh,\break ?\$:eET,?\$UFg/,?-\&JhT,?-UiM\%,?N:4Z\%,?N\&aA/,@(deGH,@(CDfg,@X)LGZ,\break @XCa*K,@2)if6,@2d8*h,@\textless :eKP,@\textless QLg/,@-\textgreater DhP,@-QiH=,@S:8Z=,@S\textgreater a6/,[xdJsH,\break [xrDfM,[X"LsA,[Xr4\#K,[q"Ff6,[qd8\#E,[\textless \&JKm,[\textless nLM\%,[\$\textgreater DEm,[\$nFH=,[j\textgreater 46\%,\break [j\&8A=,(S"Uzm,(Syn\#T,(VdUb\%,(VY\&fT,(oCn9\%,(o3\&Gm,xj)UzP,xjyQ*T,xvdU5/,\break xv7:fT,xorQ9/,xo3:sP,!N)n\#P,!N"Q*m,!vCn5=,!v7\textgreater Gm,!VrQb=,!VY\textgreater sP,X\$)UbR,\break X\$YO*T,X-"U5u,X-7t\#T,XorOGu,XoCtsR,q\textless yQbR,q\textless YOzP,q-"Q9l,q-3k\#P,qvdOGl,\break qvCkfR,2\textless yn5u,2\textless 7tzm,2\$)n9l,2\$3k*m,2Vdtsl,2Vrkfu,c-3\&5=,c-7\textgreater 9\%,cN)\&fR,\break cNdO*\%,cSy\textgreater sR,cSrOz=,B\textless Y\&5/,B\textless 7:b\%,Bj)\&Gl,BjCk*\%,BS":sl,BSrk\#/,I\$Y\textgreater 9/,\break I\$3:b=,Ijy\textgreater Gu,IjCtz=,IN":fu,INdt\#/.} {\tt 1}=($\omega$,1,1,1,1,1), {\tt W}=($\omega$,1,1,1,$\omega^2$,$\omega$),\break  {\tt p}=($\omega$,1,1,1,$\omega$,$\omega^2$), {\tt w}=($\omega$,1,1,$\omega^2$,1,$\omega$), {\tt '}=($\omega^2$,$\omega$,$\omega$,1,1,1), {\tt ;}=($\omega$,1,1,$\omega^2$,$\omega$,1),{\tt ?}=($\omega$,1,1,$\omega$,1,$\omega^2$), {\tt @}=($\omega$,1,1,$\omega$,$\omega^2$,1), {\tt [}=(1,$\omega^2$,$\omega^2$,1,1,1), {\tt (}=($\omega$,1,$\omega^2$,1,1,$\omega$),{\tt x}=($\omega^2$,$\omega$,1,$\omega$,1,1),{\tt !}=($\omega$,1,$\omega^2$,1,$\omega$,1), {\tt X}=($\omega^2$,$\omega$,1,1,$\omega$,1), {\tt q}=($\omega^2$,$\omega$,1,1,1,$\omega$), {\tt 2}=(1,$\omega^2$,$\omega$,$\omega$,1,1), {\tt c}=($\omega$,1,$\omega^2$,$\omega$,1,1),{\tt B}=(1,$\omega^2$,$\omega$,1,$\omega$,1), {\tt I}=(1,$\omega^2$,$\omega$,1,1,$\omega$), {\tt \textless }=($\omega$,1,$\omega$,1,1,$\omega^2$), {\tt \$}=($\omega$,1,$\omega$,1,$\omega^2$,1),{\tt -}=(1,$\omega^2$,1,$\omega^2$,1,1),{\tt j}=($\omega$,1,$\omega$,$\omega^2$,1,1), {\tt N}=(1,$\omega^2$,1,$\omega$,$\omega$,1), {\tt S}=(1,$\omega^2$,1,$\omega$,1,$\omega$),{\tt v}=(1,$\omega^2$,1,1,$\omega^2$,1),{\tt V}=(1,$\omega^2$,1,1,$\omega$,$\omega$),{\tt o}=(1,$\omega^2$,1,1,1,$\omega^2$),\break {\tt )}=($\omega$,$\omega^2$,1,1,1,$\omega$),{\tt "}=($\omega^2$,1,$\omega$,$\omega$,1,1),\break {\tt y}=($\omega$,$\omega^2$,1,1,$\omega$,1), {\tt d}=($\omega^2$,1,$\omega$,1,$\omega$,1), {\tt r}=($\omega^2$,1,$\omega$,1,1,$\omega$), {\tt C}=(1,$\omega$,$\omega^2$,$\omega$,1,1), {\tt Y}=($\omega$,$\omega^2$,1,$\omega$,1,1), {\tt 3}=(1,$\omega$,$\omega^2$,1,$\omega$,1), {\tt 7}=(1,$\omega$,$\omega^2$,1,1,$\omega$), {\tt k}=($\omega^2$,1,1,$\omega$,$\omega$,1), {\tt t}=($\omega^2$,1,1,$\omega$,1,$\omega$), {\tt O}=(1,$\omega$,$\omega$,$\omega^2$,1,1), {\tt :}=($\omega^2$,1,1,1,$\omega$,$\omega$),{\tt \textgreater }=($\omega^2$,1,1,1,1,$\omega^2$),{\tt \&}=($\omega^2$,1,1,1,$\omega^2$,1),{\tt Q}=(1,$\omega$,$\omega$,1,$\omega^2$,1),{\tt n}=($\omega^2$,1,1,$\omega^2$,1,1),\break {\tt U}=(1,$\omega$,$\omega$,1,1,$\omega^2$), {\tt a}=($\omega$,$\omega^2$,$\omega$,1,1,1), {\tt 8}=(1,$\omega$,1,$\omega^2$,$\omega$,1), {\tt 4}=(1,$\omega$,1,$\omega^2$,1,$\omega$), {\tt F}=(1,$\omega$,1,$\omega$,$\omega^2$,1), {\tt i}=($\omega^2$,1,$\omega^2$,1,1,1), {\tt L}=(1,$\omega$,1,$\omega$,1,$\omega^2$), {\tt D}=(1,$\omega$,1,1,$\omega^2$,$\omega$), {\tt J}=(1,$\omega$,1,1,$\omega$,$\omega^2$), {\tt *}=(1,1,1,1,1,$\omega$), {\tt z}=(1,1,1,1,$\omega$,1), {\tt \#}=(1,1,1,1,$\omega^2$,$\omega^2$), {\tt e}=(1,$\omega$,1,1,1,1), {\tt b}=(1,1,1,$\omega$,1,1), {\tt 5}=(1,1,1,$\omega$,$\omega$,$\omega^2$),\break  {\tt 9}=(1,1,1,$\omega$,$\omega^2$,$\omega$), {\tt f}=(1,1,1,$\omega^2$,1,$\omega^2$), {\tt G}=(1,1,1,$\omega^2$,$\omega$,$\omega$), {\tt s}=(1,1,1,$\omega^2$,$\omega^2$,1), {\tt Z}=(1,1,$\omega$,1,1,1), {\tt A}=(1,1,$\omega$,1,$\omega$,$\omega^2$),{\tt 6}=(1,1,$\omega$,1,$\omega^2$,$\omega$),{\tt H}=(1,1,$\omega$,$\omega$,1,$\omega^2$),{\tt g}=($\omega^2$,$\omega^2$,1,1,1,1),{\tt M}=(1,1,$\omega$,$\omega$,$\omega^2$,1),\break  {\tt E}=(1,1,$\omega$,$\omega^2$,1,$\omega$), {\tt K}=(1,1,$\omega$,$\omega^2$,$\omega$,1), {\tt h}=($\omega$,$\omega$,$\omega^2$,1,1,1), {\tt =}=($\omega$,$\omega$,1,1,1,$\omega^2$), {\tt \%}=($\omega$,$\omega$,1,1,$\omega^2$,1), {\tt l}=(1,1,$\omega^2$,1,1,$\omega^2$),{\tt R}=(1,1,$\omega^2$,1,$\omega$,$\omega$),{\tt u}=(1,1,$\omega^2$,1,$\omega^2$,1),{\tt /}=(1,1,$\omega^2$,$\omega^2$,1,1),{\tt m}=($\omega$,$\omega$,1,$\omega^2$,1,1),\break {\tt P}=(1,1,$\omega^2$,$\omega$,1,$\omega$), {\tt T}=(1,1,$\omega^2$,$\omega$,$\omega$,1)

\bigskip

  {\bf 34-16} {\tt 123456,789456,ABCDEF,GHIJEF,JCDF93,KLMI86,NOM256,PJ9356,QRSTUJ,\break OLIJE6,VWUPMD,XWTMHC,XRSOJ5,YKB756,YNA156,VQLGJE.} {\tt 1}=(0,1,0,0,0,1),\break  {\tt 2}=(1,0,0,0,1,0), {\tt 3}=(1,-1,0,0,-1,1), {\tt 4}=(1,1,0,0,-1,-1), {\tt 5}=(0,0,0,1,0,0), {\tt 6}=(0,0,1,0,0,0),\break  {\tt 7}=(1,0,0,0,0,1), {\tt 8}=(0,1,0,0,1,0), {\tt 9}=(1,-1,0,0,1,-1), {\tt A}=(1,-1,0,0,1,1), {\tt B}=(-1,1,0,0,1,1),\break  {\tt C}=(1,1,1,1,0,0), {\tt D}=(0,0,1,-1,0,0), {\tt E}=(0,0,0,0,1,-1), {\tt F}=(1,1,-1,-1,0,0), {\tt G}=(0,1,1,0,0,0),\break  {\tt H}=(1,-1,1,-1,0,0), {\tt I}=(1,0,0,1,0,0), {\tt J}=(0,0,0,0,1,1), {\tt K}=(0,1,0,0,-1,0), {\tt L}=(1,0,0,-1,0,0),\break  {\tt M}=(0,0,0,0,0,1), {\tt N}=(1,0,0,0,-1,0), {\tt O}=(0,1,0,0,0,0), {\tt P}=(1,1,0,0,0,0), {\tt Q}=(1,1,-1,1,0,0),\break  {\tt R}=(1,0,1,0,-1,1), {\tt S}=(1,0,1,0,1,-1), {\tt T}=(0,1,0,-1,0,0), {\tt U}=(-1,1,1,1,0,0), {\tt V}=(1,-1,1,1,0,0),\break  {\tt W}=(0,0,0,0,1,0), {\tt X}=(1,0,-1,0,0,0), {\tt Y}=(1,1,0,0,1,-1)

\bigskip

  {\bf 35-16} {\tt 123456,789AB6,CDEFB5,GHIJKA,LMEFA5,NJK9A4,ONIFAB,PQRSD5,TUH823,\break ULMGA1,VWSC15,XYQRM5,ZW7893,TOE723,ZYVL15,XP1356.} {\tt 1}=(0,0,0,0,1,-1),\break  {\tt 2}=(0,0,1,0,0,0), {\tt 3}=(0,0,0,1,0,0), {\tt 4}=(1,0,0,0,0,0), {\tt 5}=(0,0,0,0,1,1), {\tt 6}=(0,1,0,0,0,0),\break  {\tt 7}=(1,0,0,0,-1,0), {\tt 8}=(1,0,0,0,1,0), {\tt 9}=(0,0,0,0,0,1), {\tt A}=(0,0,1,1,0,0), {\tt B}=(0,0,1,-1,0,0),\break  {\tt C}=(1,1,-1,-1,0,0), {\tt D}=(1,1,1,1,0,0), {\tt E}=(1,-1,0,0,1,-1), {\tt F}=(1,-1,0,0,-1,1), {\tt G}=(1,-1,0,0,1,1),\break  {\tt H}=(1,1,0,0,-1,1), {\tt I}=(1,0,0,0,0,-1), {\tt J}=(0,1,-1,1,1,0), {\tt K}=(0,1,1,-1,1,0), {\tt L}=(1,1,1,-1,0,0),\break  {\tt M}=(1,1,-1,1,0,0), {\tt N}=(0,1,0,0,-1,0), {\tt O}=(1,1,0,0,1,1), {\tt P}=(1,0,-1,0,0,0), {\tt Q}=(0,1,0,-1,-1,1),\break  {\tt R}=(0,1,0,-1,1,-1), {\tt S}=(1,-1,1,-1,0,0), {\tt T}=(0,1,0,0,0,-1), {\tt U}=(-1,1,0,0,1,1), {\tt V}=(1,0,0,1,0,0),\break  {\tt W}=(0,1,1,0,0,0), {\tt X}=(1,0,1,0,0,0), {\tt Y}=(-1,1,1,1,0,0), {\tt Z}=(0,1,-1,0,0,0)

  \bigskip

  {\bf 37-16} {\tt 123456,789ABC,DEFGC6,HIJK56,LMNKFG,OPAB46,QRSTUN,VWXUJ5,PLMIE3,\break ONHI89,YZXTNH,aZWTN2,bYSTNK,baRJ15,JD7126,VQI345.} {\tt 1}=(0,1,1,-1,-1,0),\break  {\tt 2}=(0,1,1,1,1,0), {\tt 3}=(1,-1,1,0,0,-1), {\tt 4}=(1,1,-1,0,0,-1), {\tt 5}=(1,0,0,0,0,1), {\tt 6}=(0,0,0,1,-1,0),\break  {\tt 7}=(1,-1,1,0,0,1), {\tt 8}=(1,0,-1,1,-1,0), {\tt 9}=(1,0,-1,-1,1,0), {\tt A}=(0,1,0,1,1,1), {\tt B}=(0,1,0,-1,-1,1), \break {\tt C}=(1,1,1,0,0,-1), {\tt D}=(1,1,-1,0,0,1), {\tt E}=(0,0,1,0,0,1), {\tt F}=(1,-1,0,1,1,0), {\tt G}=(-1,1,0,1,1,0), \break {\tt H}=(0,1,0,0,0,0), {\tt I}=(0,0,0,1,1,0), {\tt J}=(1,0,0,0,0,-1), {\tt K}=(0,0,1,0,0,0),  {\tt L}=(1,1,0,1,-1,0), \break {\tt M}=(1,1,0,-1,1,0), {\tt N}=(0,0,0,0,0,1), {\tt O}=(1,0,1,0,0,0), {\tt P}=(1,-1,-1,0,0,1), {\tt Q}=(0,1,1,-1,1,0),\break  {\tt R}=(0,0,1,1,0,0), {\tt S}=(0,1,0,0,-1,0), {\tt T}=(1,0,0,0,0,0), {\tt U}=(0,1,-1,1,1,0),  {\tt V}=(0,1,1,1,-1,0), \break {\tt W}=(0,1,0,-1,0,0), {\tt X}=(0,0,1,0,1,0), {\tt Y}=(0,0,0,1,0,0), {\tt Z}=(0,0,1,0,-1,0), {\tt a}=(0,1,-1,1,-1,0),\break  {\tt b}=(0,1,0,0,1,0)

  \smallskip

  {\bf 37-17} {\tt 123456,789A56,BCDE34,FGE124,HIJKD3,LMJKC3,NOPQME,RSQIE3,TUVWB2,\break XYVWCD,SHEA13,ROPBE2,ZaUG9A,baYF9A,NLFGEA,ZT89A6,bX8245.} {\tt 1}=(1,1,0,0,0,0),\break  {\tt 2}=(0,0,0,0,0,1), {\tt 3}=(0,0,1,-1,0,0), {\tt 4}=(0,0,0,0,1,0), {\tt 5}=(1,-1,1,1,0,0),  {\tt 6}=(-1,1,1,1,0,0), \break {\tt 7}=(1,1,1,-1,0,0), {\tt 8}=(1,1,-1,1,0,0), {\tt 9}=(0,0,0,0,1,-1), {\tt A}=(0,0,0,0,1,1), {\tt B}=(1,0,0,0,0,0),\break  {\tt C}=(0,1,0,0,0,1), {\tt D}=(0,1,0,0,0,-1), {\tt E}=(0,0,1,1,0,0), {\tt F}=(1,-1,1,-1,0,0), {\tt G}=(1,-1,-1,1,0,0),\break  {\tt H}=(1,-1,0,0,1,-1), {\tt I}=(1,1,0,0,1,1), {\tt J}=(1,0,1,1,-1,0), {\tt K}=(-1,0,1,1,1,0), {\tt L}=(1,1,0,0,1,-1),\break  {\tt M}=(1,-1,0,0,1,1), {\tt N}=(1,1,0,0,-1,1), {\tt O}=(0,1,1,-1,1,0), {\tt P}=(0,1,-1,1,1,0),{\tt Q}=(1,0,0,0,0,-1),\break  {\tt R}=(0,1,0,0,-1,0), {\tt S}=(1,-1,0,0,-1,1), {\tt T}=(0,1,0,-1,0,0), {\tt U}=(0,1,0,1,0,0), {\tt V}=(0,0,1,0,-1,0),\break  {\tt W}=(0,0,1,0,1,0), {\tt X}=(1,0,0,-1,0,0), {\tt Y}=(1,0,0,1,0,0), {\tt Z}=(1,0,1,0,0,0), {\tt a}=(1,1,-1,-1,0,0),\break  {\tt b}=(0,1,1,0,0,0)

\smallskip

{\bf 37-18} {\tt 123456,789ABC,DEFGHI,JKLBC6,MNOHI5,PQLIBC,RSTUO4,VTUNG4,WXSUF4,\break XUME34,YUQ345,BC1246,ZKA235,ZYI935,abPJ8C,abU7C1,ZVDI9C,WRPLIC.}\break  {\tt 1}=(0,1,0,0,0,0), {\tt 2}=(1,0,1,0,0,0), {\tt 3}=(0,0,0,0,1,0), {\tt 4}=(0,0,0,1,0,0), {\tt 5}=(0,0,0,0,0,1),\break  {\tt 6}=(1,0,-1,0,0,0), {\tt 7}=(1,0,0,-1,0,0), {\tt 8}=(0,1,-1,0,0,0), {\tt 9}=(1,1,1,1,0,0), {\tt A}=(1,-1,-1,1,0,0),\break  {\tt B}=(0,0,0,0,1,-1), {\tt C}=(0,0,0,0,1,1), {\tt D}=(1,-1,0,0,-1,1), {\tt E}=(0,1,0,0,0,1), {\tt F}=(1,0,0,0,1,0),\break  {\tt G}=(1,1,0,0,-1,-1), {\tt H}=(0,0,1,1,0,0), {\tt I}=(0,0,1,-1,0,0), {\tt J}=(1,1,1,-1,0,0), {\tt K}=(0,1,0,1,0,0),\break  {\tt L}=(1,-1,1,1,0,0), {\tt M}=(1,0,0,0,0,0), {\tt N}=(0,1,0,0,1,0), {\tt O}=(0,1,0,0,-1,0), {\tt P}=(-1,1,1,1,0,0),\break  {\tt Q}=(1,1,0,0,0,0), {\tt R}=(1,1,0,0,1,-1), {\tt S}=(-1,1,0,0,1,1), {\tt T}=(1,0,0,0,0,1), {\tt U}=(0,0,1,0,0,0),\break  {\tt V}=(1,-1,0,0,1,-1), {\tt W}=(1,1,0,0,-1,1), {\tt X}=(0,1,0,0,0,-1), {\tt Y}=(1,-1,0,0,0,0), {\tt Z}=(1,1,-1,-1,0,0),\break  {\tt a}=(1,0,0,1,1,-1),\quad{\tt b}=(1,0,0,1,-1,1)

\smallskip

{\bf 38-18} {\tt 123456,789ABC,DEFGBC,HIJKG6,LMNOC5,POFGAC,QRS946,SJKF36,TN8ABC,\break UVWXO2,WXPMO2,XTR826,YVR234,ZXI156,ZQH7AB,aYUNEC,aLNDC1,bcXE25.}\break  {\tt 1}=(0,1,-1,-1,1,0), {\tt 2}=(1,0,0,0,0,1), {\tt 3}=(1,-1,0,0,1,-1), {\tt 4}=(1,1,0,0,-1,-1), {\tt 5}=(0,1,1,1,1,0), \break {\tt 6}=(0,0,1,-1,0,0), {\tt 7}=(1,0,0,0,-1,0), {\tt 8}=(0,1,0,0,0,0), {\tt 9}=(1,0,0,0,1,0), {\tt A}=(0,0,0,1,0,0),\break  {\tt B}=(0,0,1,0,0,0), {\tt C}=(0,0,0,0,0,1), {\tt D}=(0,0,0,1,1,0), {\tt E}=(0,0,0,1,-1,0),  {\tt F}=(1,1,0,0,0,0),\break  {\tt G}=(1,-1,0,0,0,0), {\tt H}=(1,1,0,0,1,-1), {\tt I}=(1,1,0,0,-1,1), {\tt J}=(0,0,1,1,1,1), {\tt K}=(0,0,1,1,-1,-1),\break  {\tt L}=(0,1,-1,1,-1,0), {\tt M}=(0,1,0,-1,0,0), {\tt N}=(1,0,0,0,0,0), {\tt O}=(0,0,1,0,-1,0), {\tt P}=(0,0,1,0,1,0),\break  {\tt Q}=(0,1,0,0,0,1), {\tt R}=(0,0,1,1,0,0), {\tt S}=(1,-1,0,0,-1,1), {\tt T}=(0,0,0,0,1,0), {\tt U}=(0,-1,1,1,1,0),\break  {\tt V}=(0,1,1,-1,1,0), {\tt W}=(0,1,0,1,0,0), {\tt X}=(1,0,0,0,0,-1), {\tt Y}=(0,1,-1,1,1,0),  {\tt Z}=(1,-1,0,0,1,1),\break {\tt a}=(0,1,1,0,0,0), {\tt b}=(0,1,1,-1,-1,0), {\tt c}=(0,1,-1,0,0,0))

\section{\label{app78string}ASCII strings and
  coordinatizations of 7- and 8-dim MMP
  hypergraphs given in Fig.~\ref{fig:78d} and
  Table \ref{T:78d}}

\centerline{\bf 7-dim}

  \parindent=0pt
  
{\bf 34-14}\hfil {\tt 4567231,19KHBL8,8WYVJPA,AJPRNSQ,QTU5674,189A5BC,189DE7F,189GHIJ,\break 2MNDOIP,2MNEOCL,2MNGK6F,RTV9567,WXMS567,XY3U567.} {\tt 1}=(0,0,0,1,0,0,0),\break  {\tt 2}=(0,0,1,0,0,0,0), {\tt 3}=(1,-1,0,0,0,0,0), {\tt 4}=(1,1,0,0,0,0,0), {\tt 5}=(0,0,0,0,0,0,1), {\tt 6}=(0,0,0,0,1,1,0),\break  {\tt 7}=(0,0,0,0,1,-1,0),$\,${\tt 8}=(0,1,-1,0,0,0,0),$\,${\tt 9}=(0,1,1,0,0,0,0),$\,${\tt A}=(0,0,0,0,1,0,0),{\tt B}=(1,0,0,0,0,-1,0),\break  {\tt C}=(1,0,0,0,0,1,0),{\tt D}=(1,0,0,0,1,1,-1),{\tt E}=(-1,0,0,0,1,1,1),{\tt F}=(1,0,0,0,0,0,1),$\,${\tt G}=(1,0,0,0,1,-1,-1),\break  {\tt H}=(1,0,0,0,1,1,1),{\tt I}=(1,0,0,0,-1,0,0),{\tt J}=(0,0,0,0,0,1,-1),{\tt K}=(1,0,0,0,-1,1,-1),{\tt L}=(0,0,0,0,1,0,-1),\break  {\tt M}=(0,1,0,1,0,0,0),{\tt N}=(0,1,0,-1,0,0,0),{\tt O}=(1,0,0,0,1,-1,1), {\tt P}=(0,0,0,0,0,1,1), {\tt Q}=(-1,1,1,1,0,0,0),\break  {\tt R}=(1,1,-1,1,0,0,0),{\tt S}=(1,0,1,0,0,0,0),{\tt T}=(1,-1,1,1,0,0,0),{\tt U}=(0,0,1,-1,0,0,0),$\,${\tt V}=(1,0,0,-1,0,0,0),\break  {\tt W}=(1,-1,-1,1,0,0,0), {\tt X}=(1,1,-1,-1,0,0,0), {\tt Y}=(1,1,1,1,0,0,0)

\bigskip

\centerline{\bf 8-dim}

\medskip
  \parindent=0pt

{\bf 34-9} {\tt 9ABC5687,78DE34GF,FGHILMKJ,JKVWUPA9,12345678,NOPQRMEC,STUQRLDB,\break XYWTOI28,XYVSNH17.}{\tt 1}=(0,0,1,1,1,1,0,0), {\tt 2}=(0,0,1,-1,1,-1,0,0), {\tt 3}=(0,0,0,1,0,-1,0,0),\break  {\tt 4}=(0,0,1,0,-1,0,0,0), {\tt 5}=(0,1,0,0,0,0,0,0), {\tt 6}=(1,0,0,0,0,0,0,0), {\tt 7}=(0,0,0,0,0,0,0,1),\break  {\tt 8}=(0,0,0,0,0,0,1,0), {\tt 9}=(0,0,0,1,0,0,0,0), {\tt A}=(0,0,1,0,0,0,0,0), {\tt B}=(0,0,0,0,0,1,0,0),\break  {\tt C}=(0,0,0,0,1,0,0,0), {\tt D}=(1,-1,1,0,1,0,0,0), {\tt E}=(1,1,0,1,0,1,0,0), {\tt F}=(1,1,0,-1,0,-1,0,0),\break  {\tt G}=(-1,1,1,0,1,0,0,0), {\tt H}=(0,1,-1,1,0,0,1,0), {\tt I}=(1,0,1,1,0,0,0,-1), {\tt J}=(1,0,0,0,1,1,0,1),\break  {\tt K}=(0,1,0,0,-1,1,-1,0), {\tt L}=(0,0,1,0,-1,0,1,1), {\tt M}=(0,0,0,1,0,-1,-1,1), {\tt N}=(1,0,1,0,0,-1,1,0),\break  {\tt O}=(0,-1,1,0,0,1,0,1), {\tt P}=(-1,1,0,0,0,0,1,1), {\tt Q}=(1,0,-1,-1,0,0,0,1), {\tt R}=(0,1,1,-1,0,0,-1,0),\break  {\tt S}=(1,0,0,1,-1,0,-1,0), {\tt T}=(0,1,0,1,1,0,0,1), {\tt U}=(1,1,0,0,0,0,1,-1), {\tt V}=(0,1,0,0,1,-1,-1,0),\break  {\tt W}=(1,0,0,0,-1,-1,0,1), {\tt X}=(1,1,0,-1,0,1,0,0), {\tt Y}=(1,-1,-1,0,1,0,0,0)

\medskip

{\bf 36-9 star}\hfil {\tt 12345678,89ABCDEF,FGHI4JKL,L7MNBOPQ,QERSI3TU,UK6VNAWX,XPDYSH2Z,\break ZTJ5VM9a,aWOCYRG1.}{\tt 1}=(0,0,0,0,0,0,0,1), {\tt 2}=(0,0,0,0,0,0,1,0), {\tt 3}=(0,0,0,0,0,1,0,0),\break  {\tt 4}=(0,0,0,0,1,0,0,0), {\tt 5}=(0,0,1,1,0,0,0,0), {\tt 6}=(0,0,1,-1,0,0,0,0), {\tt 7}=(1,1,0,0,0,0,0,0),\break  {\tt 8}=(1,-1,0,0,0,0,0,0), {\tt 9}=(1,1,0,0,0,0,-1,1), {\tt A}=(0,0,1,1,1,-1,0,0), {\tt B}=(0,0,0,0,0,0,1,1),\break  {\tt C}=(0,0,1,-1,1,1,0,0), {\tt D}=(0,0,0,1,0,1,0,0), {\tt E}=(0,0,1,0,-1,0,0,0), {\tt F}=(1,1,0,0,0,0,1,-1),\break  {\tt G}=(0,0,1,0,0,-1,0,0), {\tt H}=(1,0,0,0,0,0,0,1), {\tt I}=(0,0,0,1,0,0,0,0), {\tt J}=(1,-1,0,0,0,0,-1,-1),\break  {\tt K}=(0,1,0,0,0,0,-1,0), {\tt L}=(0,0,1,0,0,1,0,0), {\tt M}=(0,0,1,-1,1,-1,0,0), {\tt N}=(1,-1,0,0,0,0,-1,1),\break  {\tt O}=(0,0,0,1,1,0,0,0), {\tt P}=(0,0,1,1,-1,-1,0,0), {\tt Q}=(1,-1,0,0,0,0,1,-1), {\tt R}=(1,0,0,0,0,0,-1,0),\break  {\tt S}=(0,0,1,0,1,0,0,0), {\tt T}=(0,1,0,0,0,0,0,-1), {\tt U}=(1,1,0,0,0,0,1,1), {\tt V}=(0,0,0,0,1,1,0,0),\break  {\tt W}=(0,0,1,1,-1,1,0,0), {\tt X}=(1,0,0,0,0,0,0,-1), {\tt Y}=(0,1,0,0,0,0,0,0), {\tt Z}=(0,0,1,-1,-1,1,0,0),\break  {\tt a}=(1,0,0,0,0,0,1,0)

\medskip

{\bf 36-9 hexagon}\hfil {\tt 34125687,78XYVWKJ,JKHILMON,NOPQRSGF,FG9ADECB,BCZaTU43,\break TUVWLMDE,ZaRSHI56,XYPQ9A12.} {\tt 1}=(0,0,0,0,0,0,0,1), {\tt 2}=(0,0,0,0,0,1,0,0),\break  {\tt 3}=(0,0,1,1,0,0,0,0), {\tt 4}=(1,1,0,0,0,0,0,0), {\tt 5}=(0,0,1,-1,0,0,0,0), {\tt 6}=(1,-1,0,0,0,0,0,0),\break  {\tt 7}=(0,0,0,0,0,0,1,0), {\tt 8}=(0,0,0,0,1,0,0,0), {\tt 9}=(1,0,0,0,0,0,-1,0), {\tt A}=(0,0,1,0,1,0,0,0),\break  {\tt B}=(1,-1,0,0,0,0,1,-1), {\tt C}=(0,0,1,-1,-1,1,0,0), {\tt D}=(0,0,1,1,-1,-1,0,0), {\tt E}=(1,1,0,0,0,0,1,1),\break  {\tt F}=(0,0,0,1,0,1,0,0), {\tt G}=(0,1,0,0,0,0,0,-1), {\tt H}=(1,1,0,0,0,0,1,-1), {\tt I}=(0,0,1,1,-1,1,0,0),\break  {\tt J}=(0,0,1,0,0,-1,0,0), {\tt K}=(1,0,0,0,0,0,0,1), {\tt L}=(0,1,0,0,0,0,-1,0), {\tt M}=(0,0,0,1,1,0,0,0),\break  {\tt N}=(0,0,1,-1,1,1,0,0), {\tt O}=(1,-1,0,0,0,0,-1,-1), {\tt P}=(0,0,1,0,-1,0,0,0), {\tt Q}=(1,0,0,0,0,0,1,0),\break  {\tt R}=(1,1,0,0,0,0,-1,1), {\tt S}=(0,0,1,1,1,-1,0,0), {\tt T}=(0,0,1,-1,1,-1,0,0), {\tt U}=(1,-1,0,0,0,0,-1,1),\break  {\tt V}=(0,0,1,0,0,1,0,0), {\tt W}=(1,0,0,0,0,0,0,-1), {\tt X}=(0,0,0,1,0,0,0,0), {\tt Y}=(0,1,0,0,0,0,0,0),\break  {\tt Z}=(0,0,0,0,0,0,1,1), {\tt a}=(0,0,0,0,1,1,0,0)

\medskip
{\bf 37-11} {\tt 789A56CB,BCDEFIHG,GHWXYVRP,PROQ3487,12345678,JKLMNIAC,STUVQRMN,\break ZaYULF28,ZaXTKE17,bJDI9ABC,bWSORIAB.} {\tt 1}=(0,0,1,-1,1,0,0,1), {\tt 2}=(0,0,1,1,-1,0,0,1),\break  {\tt 3}=(0,0,0,1,1,0,0,0), {\tt 4}=(1,-1,0,0,0,0,0,0), {\tt 5}=(1,1,0,0,0,0,0,0), {\tt 6}=(0,0,1,0,0,0,0,-1),\break  {\tt 7}=(0,0,0,0,0,0,1,0), {\tt 8}=(0,0,0,0,0,1,0,0), {\tt 9}=(-1,1,1,0,0,0,0,1), {\tt A}=(0,0,0,0,1,0,0,0),\break  {\tt B}=(0,0,0,1,0,0,0,0), {\tt C}=(1,-1,1,0,0,0,0,1), {\tt D}=(0,0,1,0,0,1,-1,-1), {\tt E}=(1,0,0,0,1,-1,0,-1),\break  {\tt F}=(0,1,0,0,1,0,-1,1), {\tt G}=(0,1,1,0,-1,-1,0,0), {\tt H}=(-1,0,1,0,1,0,1,0), {\tt I}=(1,1,0,0,0,1,1,0),\break  {\tt J}=(0,0,1,0,0,-1,1,-1), {\tt K}=(0,1,0,1,0,-1,0,1), {\tt L}=(1,0,0,1,0,0,-1,-1), {\tt M}=(0,1,1,-1,0,0,-1,0),\break  {\tt N}=(-1,0,1,1,0,1,0,0), {\tt O}=(0,0,1,0,0,0,0,0), {\tt P}=(1,1,0,-1,1,0,0,0), {\tt Q}=(1,1,0,1,-1,0,0,0),\break  {\tt R}=(0,0,0,0,0,0,0,1), {\tt S}=(1,-1,0,0,0,1,-1,0), {\tt T}=(0,1,-1,0,1,1,0,0), {\tt U}=(1,0,1,0,1,0,1,0),\break  {\tt V}=(0,0,0,1,1,-1,-1,0), {\tt W}=(1,-1,0,0,0,-1,1,0), {\tt X}=(1,0,1,1,0,1,0,0), {\tt Y}=(0,1,-1,1,0,0,1,0),\break  {\tt Z}=(1,1,0,-1,-1,0,0,0), {\tt a}=(1,-1,-1,0,0,0,0,1), {\tt b}=(1,1,0,0,0,-1,-1,0)

\medskip
{\bf 52-13} {\tt 9ABCDEFG,GFIKHLMJ,JLMH41aY,YaQONUXZ,ZXRfgVih,hiopklnj,jkln52md,\break dTb3Pce9,12345678,KD67AINO,PQRSTUVW,EBogqpfe,CmS8Wcqb.} {\tt H}=(0,1,-1,1,0,0,0,1),\break  {\tt K}=(1,-1,-1,0,0,1,0,0), {\tt e}=(0,0,0,0,0,0,0,1), {\tt Q}=(1,0,0,1,0,1,0,1), {\tt R}=(0,1,0,0,1,1,0,-1),\break  {\tt E}=(0,0,0,0,0,1,0,0), {\tt q}=(0,0,0,0,0,0,1,0), {\tt 5}=(0,1,1,0,-1,0,1,0), {\tt 7}=(1,0,1,0,1,0,0,-1),\break  {\tt U}=(1,0,1,-1,0,0,1,0), {\tt i}=(0,0,0,0,0,1,0,1), {\tt o}=(0,0,0,0,1,0,0,0), {\tt A}=(1,1,0,0,0,0,1,1),\break  {\tt h}=(0,0,1,0,0,0,1,0), {\tt f}=(0,0,0,1,0,0,0,0), {\tt V}=(0,1,1,0,-1,0,-1,0), {\tt 3}=(1,0,-1,-1,0,0,1,0),\break  {\tt 4}=(0,-1,0,1,0,1,1,0), {\tt T}=(0,1,0,1,0,-1,1,0), {\tt S}=(0,1,-1,-1,0,0,0,1), {\tt P}=(1,0,0,0,1,-1,-1,0),\break  {\tt b}=(0,0,-1,1,1,1,0,0), {\tt k}=(0,0,1,0,0,0,-1,0), {\tt Y}=(1,0,-1,0,1,0,0,-1), {\tt m}=(0,1,0,0,1,-1,0,-1),\break  {\tt 1}=(1,0,0,0,-1,1,-1,0), {\tt B}=(0,0,1,0,0,0,0,0), {\tt c}=(1,1,1,0,0,1,0,0), {\tt M}=(1,0,0,0,-1,-1,1,0),\break  {\tt F}=(1,1,0,0,0,0,-1,-1), {\tt J}=(1,0,1,0,1,0,0,1), {\tt L}=(0,1,0,-1,0,1,1,0), {\tt 2}=(0,1,0,0,1,1,0,1),\break  {\tt p}=(0,1,0,0,0,0,0,0), {\tt X}=(0,1,0,0,1,-1,0,1), {\tt O}=(-1,0,0,0,1,1,1,0), {\tt n}=(0,0,0,0,0,1,0,-1),\break  {\tt W}=(-1,0,1,0,1,0,0,1), {\tt 8}=(1,0,0,1,0,-1,0,1), {\tt N}=(0,1,0,-1,0,1,-1,0), {\tt G}=(0,0,0,1,1,0,1,-1),\break  {\tt 9}=(1,-1,0,1,-1,0,0,0), {\tt 6}=(0,1,-1,1,0,0,0,-1), {\tt l}=(1,0,0,1,0,0,0,0), {\tt a}=(0,1,1,1,0,0,0,-1),\break  {\tt Z}=(0,1,-1,0,-1,0,1,0), {\tt D}=(0,0,0,1,1,0,-1,1), {\tt d}=(0,-1,1,0,1,0,1,0), {\tt j}=(1,0,0,-1,0,0,0,0),\break  {\tt I}=(0,0,1,1,-1,1,0,0), {\tt g}=(1,0,0,0,0,0,0,0), {\tt C}=(1,-1,0,-1,1,0,0,0)

\bigskip\bigskip

\bibliographystyle{quantum}



\end{document}